\documentclass[journal,onecolumn]{IEEEtran}
\usepackage{notation}
\begin{document}

\title{Exact Reconstruction of Euclidean Distance Geometry Problem Using Low-rank Matrix Completion}
\author{Abiy Tasissa,
        \thanks{A. Tasissa is with the Department of Mathematics, Rensselaer Polytechnic Institute, Troy, NY 12180,  ({\tt tasisa@rpi.edu}).}
        \and
        Rongjie Lai
            \thanks{R. Lai is with the Department of Mathematics, Rensselaer Polytechnic Institute, Troy, NY 12180, 
         U.S.A. ({\tt lair@rpi.edu}). This work is partially supported by NSF DMS--1522645 and an NSF CAREER award, DMS--1752934.}
        }
\date{}
\maketitle
%\IEEEpeerreviewmaketitle
%\maketitle

\begin{abstract}

The Euclidean distance geometry problem arises in a wide variety of applications, from determining molecular conformations in computational chemistry to localization in sensor networks. When the distance information
is incomplete, the problem can be formulated as a nuclear norm minimization problem. 
In this paper, this minimization program is recast as a matrix completion problem of a
low-rank $r$ Gram matrix with respect to a suitable basis. 
The well known restricted isometry property can not be satisfied in this scenario. 
Instead, a dual basis approach is introduced to theoretically analyze the reconstruction problem.  
If the Gram matrix satisfies certain coherence conditions with parameter $\nu$, the main result shows
that the underlying configuration of $n$ points
can be recovered with very high probability from $O(nr\nu\log^{2}(n))$ uniformly random samples. 
Computationally, simple and fast algorithms are designed to solve the Euclidean distance geometry problem. 
Numerical tests on different three dimensional data and protein molecules validate
effectiveness and efficiency of the proposed algorithms.
\end{abstract}


\section{Introduction}

The Euclidean distance geometry problem, EDG hereafter, aims at constructing the configuration of points given partial information on pairwise inter-point distances. The problem has applications in diverse areas, such as molecular conformation in computational chemistry
\cite{glunt1993molecular,trosset1997applications,fang2013using}, dimensionality reduction in machine learning \cite{tenenbaum2000global}
and localization in sensor networks \cite{ding2010sensor,biswas2006semidefinite}. For instance, in the molecular conformation problem,
the goal is to determine the structure of protein given partial inter-atomic distances obtained
from nuclear magnetic resonance (NMR) spectroscopy experiments. Since the structure of protein
determines its physical and chemical properties, the molecular conformation problem
is crucial for biological applications such as drug design. In recent works, novel applications of EDG like solving
partial differential equations (PDEs) on manifolds represented as incomplete inter-point distance information have been explored in~\cite{lai2017solve}.

For the mathematical setup of the problem, consider a set of $n$ points $\Pb=\{\p_{1},\p_{2},...,\p_{n}\}^{T} \in 
\real^{n\times r}$.
The squared distance between any two points $\p_{i}$ and $\p_{j}$ is given by $d_{i,j}^{2} = ||\p_{i}-\p_{j}||_{2}^{2}
= \p_{i}^{T}\p_{i}+\p_{j}^{T}\p_{j}-2\p_{i}^{T}\p_{j}$. The Euclidean distance matrix $\D =[d_{i,j}^{2}]$ can be written compactly
as follows
\begin{equation}\label{eq:dist_matrix}
\D = \bm{1}\,\textrm{diag}(\Pb\Pb^{T})^{T}+\textrm{diag}(\Pb\Pb^{T})\,\bm{1}^{T}-2\Pb\Pb^{T}
\end{equation}
where $\bm{1}$ is a column vector of ones and $\textrm{diag}(\cdot)$ is a column vector
of diagonal entries of the matrix in consideration. $\X=\Pb\Pb^{T}$ is the inner product matrix well known as the Gram matrix. 
If the complete distance information is available, the coordinates can be recovered from the
eigendecomposition of the Gram matrix using classical multidimensional scaling (MDS)~\cite{torgerson1952multidimensional}. It should be noted that this solution is unique up to rigid transformations and translations.

In practice, the distance matrix is incomplete and finding the underlying point
coordinates with partial information is in general not possible.  
One approach proposed in \cite{lai2017solve} is based on a matrix completion applied to the Gram matrix
which we briefly summarize here.
Assume that the entries of $\D$ are sampled uniformly at random.
By construction, $\X$ is symmetric and positive semidefinite. The solution for $\X$ is unique up to translations. This ambiguity is fixed by
considering the constraint that the centroid of the points
is located at the origin, $\sum_{k=1}^{n}\p_{k} = \bm{0}$, which leads
to $\X\cdot \bm{1} = \bm{0}$.
Assuming that $\X$ is low-rank, $r\ll n$, the work in \cite{lai2017solve} considers the following nuclear
norm minimization program to recover $\X$.
\begin{align}
\label{eq:nnm_EDG}
\textrm{minimize } \quad &||\X||_{*} \nonumber\\
 \textrm{subject to }\quad &\X_{i\,,i}+\X_{j\,,j}-2\X_{i\,,j} = \D_{i\,,j}\quad (i,j)\in \Omega
\\
&\X\cdot \bm{1}=\bm{0}\,;\,\X =\X^{T}\,;\, \X\succeq \bm{0} \nonumber
\end{align}
Here $||\X||_{*}$ denotes the nuclear norm and $\Omega\subset\{(i,j)|i,j=1,...,n, i<j\}$, $|\Omega|=m$, denotes the random set 
that consists of all the sampled indices. One characterization of the Euclidean distance matrix, due to Gower \cite{gower1985properties}, states that $\D$ is an Euclidean distance matrix
if and only if $\D = \X_{i\,,i}+\X_{j\,,j}-2\X_{i\,,j} = \D_{i\,,j}\,\,\forall i,j$
for some positive semidefinite matrix $\X$ satisfying
$\X\cdot \bm{1}=\bm{0}$. As such, the
above nuclear norm minimization can be interpreted as a regularization of the rank
with a prior assumption that the true Gram matrix is low rank. An alternative approach based on the matrix completion method is direct completion
of the distance matrix \cite{candes2009exact,moreira2017novel,recht2010guaranteed}.
Compared to this approach,
an advantage of the above minimization program can be seen by comparing the rank of
the Gram matrix $\X$ and the distance matrix $\D$. Using \eqref{eq:dist_matrix}, the rank of $\D$ is at most $r+2$ while
the rank of $\X$ is simply $r$. Using matrix completion theory, it can be surmised that
relatively less number of samples are required
for the Gram matrix completion. Numerical experiments in \cite{lai2017solve} confirm
this observation. In \cite{javanmard2013localization}, the authors consider
a theoretical analysis of a specific instance of localization problem and propose an algorithm similar to 
\eqref{eq:nnm_EDG}. The paper considers a random geometric model and derives 
interesting results of bound of errors in reconstructing point coordinates. While
the localization problem and EDG problem share a common theme, we remark that the EDG problem is different
and our analysis adopts the matrix completion framework. 
The main task of this paper is a theoretical analysis of the above minimization problem. 
In particular, under appropriate conditions, we will show that the above nuclear norm minimization 
recovers the underlying inner product matrix.

\paragraph{Related Work}
The EDG problem has been well studied in various contexts. Early works establish mathematical properties of Euclidean distance
matrices (EDM) \cite{gower1985properties} and prove conditions for a symmetric hollow matrix to be EDM
\cite{schoenberg1935remarks,young1938discussion}. Particularly, the important result due to Schoenberg states that
a symmetric hollow matrix $\D$ is EDM if and only if the matrix $-\frac{1}{2}\bm{J}\bm{D}\bm{J}$
is positive semidefinite. $\bm{J}$ is the geometric centering matrix defined as $\bm{J}=\bm{I}-\frac{1}{n}\bm{1}\bm{1}^{T}$
\cite{schoenberg1935remarks}. In follow-up theoretical papers, the EDM completion problem
was considered. Given a symmetric hollow matrix with some entries specified,
the EDM completion problem asks if the partial matrix can be completed to an EDM. 
A work of Hendrickson relates the EDM completion problem to a graph realization problem
and provides necessary graph theoretic conditions on the partial matrix that ensures a unique solution~\cite{hendrickson1992conditions}.
Bakonyi and Johnson establish a link between EDM completion and the graph of partial distance matrix.
Specifically, they show that there is a completion to a distance matrix if the graph is chordal~\cite{bakonyi1995euclidian}.
The connection between the EDM completion problem and the positive semidefinite matrix completion
was established in \cite{johnson1995connections} and \cite{laurent1998connection}.  
In \cite{alfakih2003uniqueness}, Alfakih establishes necessary and sufficient conditions for the solution of EDM completion problem to be unique. The emphasis in the above theoretical works is analytic characterization of 
EDMs and the EDM completion problem mostly employing graph theoretic methods. 
In the numerical side, a variety of algorithms using different optimization techniques have been developed to solve the EDM completion problem ~\cite{fang2013using,alfakih1999solving,fang2012euclidean,hendrickson1995molecule,more1997global,trosset2000distance,zou1997stochastic}.
The above review is by no means exhaustive and we recommend the interested
reader to refer to \cite{dokmanic2015euclidean} and \cite{liberti2014euclidean}. 
This paper adopts the low-rank matrix recovery framework. While the well-known matrix completion
theory can be applied to reconstruct $\D$ \cite{candes2009exact,recht2010guaranteed}, it
can not be directly used to analyze \eqref{eq:nnm_EDG}. In particular,
the measurements in \eqref{eq:nnm_EDG} are with respect to the
Gram matrix while the measurements in the work of \cite{candes2009exact,recht2010guaranteed}
are entry wise sampling of the distance matrix. 
The emphasis in this paper is on theoretical understanding of the nuclear norm minimization
formulation for the EDG problem as stated in \eqref{eq:nnm_EDG}. In particular,
the goal is to provide rigorous probabilistic guarantees that give
precise estimates for the minimum number of measurements needed for a certain success probability
of the recovery algorithm.

\paragraph{Main Challenges}
The random linear constraints in \eqref{eq:nnm_EDG} can be equivalently written as
a linear map $\mathcal{L}$ acting on $\X$. A first approach to showing uniqueness for the problem 
\eqref{eq:nnm_EDG} is to check if $\mathcal{L}$ satisfies the restricted isometry property (RIP)~\cite{recht2010guaranteed}. However, the RIP does not hold for the completion problem~\eqref{eq:nnm_EDG}. This can be simply verified by choosing any $(i,j)\notin \Omega$ and construct a matrix $\X$ with $\X_{i,j} = \X_{j,i} = \X_{i,i}=\X_{j,j}=1$
and zero everywhere else. Then, it is clear that $\mathcal{L}(\X) = \bm{0}$ implying that the RIP
condition does not hold. In general, RIP based analysis is not suitable for deterministic structured measurements.
Adopting the framework introduced in \cite{gross2011recovering},
the nuclear norm minimization problem in \eqref{eq:nnm_EDG}
can be written as a matrix completion problem with respect to a basis
$\{\w_{\alphab}\}$. It, however, turns out that $\w_{\alphab}$ is not
orthogonal. With non-orthogonal basis, the measurements $\langle \w_{\alphab},\X \rangle$ are not
compatible with the expansion coefficients of the true Gram matrix. A possible remedy is orthogonalization
of the basis, say using the Gram-Schmidt orthogonalization.
Unfortunately, the orthogonalization process does not preserve the structure of the basis $\w_{\alphab}$.
This has the implication that the modified minimization problem \eqref{eq:nnm_EDG} no longer corresponds
to the original problem. As such, the lack of orthogonality is critical in this problem.
In addition, it is necessary that the solution of \eqref{eq:nnm_EDG}
is symmetric positive semidefinite satisfying the constraint $\X\cdot \bm{1}=\bm{0}$.
On the basis of the above considerations, an alternative approach
is considered to show that \eqref{eq:nnm_EDG} admits a unique solution.
The analysis presented in this paper is not based on RIP but
on the dual certificate approach introduced in \cite{candes2009exact}.
Our proof was inspired by the work of David Gross \cite{gross2011recovering}
where the author generalizes the matrix completion problem to any orthonormal basis.
In the case of the EDG problem, 
one main challenge is that the sampling operator, an important operator in matrix completion,
is no longer self-adjoint. This necessitates modifications and alternative proofs to some of the technical statements that appear
in \cite{gross2011recovering}.

\paragraph{Major Contributions}
In this paper, a dual basis approach is introduced to show that \eqref{eq:nnm_EDG}
has a unique solution under appropriate sampling conditions. First, the
minimization problem in \eqref{eq:nnm_EDG} is written as matrix
completion problem with respect to a basis $\w_{\alphab}$.
Second, by introducing a dual basis $\{\v_{\alphab}\}$ to
$\{\w_{\alphab}\}$, one can ensure that the measurements $\langle \X\,,\w_{\alphab}\rangle$
in \eqref{eq:nnm_EDG} are compatible with expansion coefficients
of the true Gram matrix $\M$.  
The two main contributions of this paper are as follows.
\begin{enumerate}
  \item A dual basis approach is introduced to address the EDG problem. Under certain assumptions, we show
that the nuclear norm minimization problem succeeds in recovering the underlying low rank solution.
Precisely, the main result states that if $|\Omega| = m\ge O(nr\log^{2}n)$, the nuclear
norm minimization program \eqref{eq:nnm_EDG} recovers the underlying inner product matrix
with very high probability, $1 - n^{-\beta}$ for $\beta>1$(see more details in Theorem \ref{thm1}). 
The minimization problem has a positive semidefinite constraint.
The proof describes how this constraint is handled.  In addition to the Bernstein bound which appears in matrix completion analysis, 
a crucial part in the proof uses the operator Chernoff bound.  We would like to emphasize that this part of the proof and the estimate using Chernoff  bound
is simple and could be useful in a broader setting. 

\item We develop simple and fast algorithms to solve the EDG problem under two instances.
The first instance considers the case of exact partial information. The second
instance considers the more realistic setup of a noisy partial information. Numerical tests on various data show that the algorithms are accurate and efficient.

 \end{enumerate}
%\paragraph{Outline}
The outline of the paper is as follows. In section \ref{sec:Dualbasis}, we introduce a dual basis approach and formulate a well-defined matrix completion 
problem for the EDG problem.  We conclude the section by proposing coherence conditions for the EDG problem and explaining
the sampling scheme. In section \ref{sec:Proof}, the proof of exact reconstruction is presented.
In brief terms, the main parts are summarized as follows. From convex optimization theory,
showing uniqueness of the EDG problem is equivalent to showing that there exists
a dual certificate, denoted by $\Y$, satisfying certain conditions. $\Y$ is constructed
using the golfing scheme proposed in \cite{gross2011recovering}.
Next,  we show that these conditions
hold with very high probability by employing concentration inequalities. This implies that there is a unique solution
with very high probability. 
In section \ref{sec:Algorithms}, fast numerical algorithms for the EDG matrix completion problem
are proposed. Section \ref{sec:Results} validates the efficiency and effectiveness of the proposed numerical algorithms. Finally, we conclude the work in the last section.

\paragraph{Notation}
To make our notation consistent in the paper, we summarize notations used in this paper in Table \ref{tab:notation}.
\begin{table*}[h!]
\renewcommand{\arraystretch}{1.3}
\caption{Notations}
\label{tab:notation}
\centering
\begin{tabular}{|cc||cc|}
\hline
$\x$ & Vector & $||\x||_{2}$ & $l_2$ norm\\
$\X$ & Matrix & $||\X||_{F}$  &  Frobenius norm\\
$\mathcal{X}$ & Operator & $||\X||$ & Operator norm \\ 
$\X^T$ & Transpose & $||\X||_{*}$ & Nuclear norm\\
$\textrm{Tr}(\X)$ & Trace & $||\mathcal{A}||$ & $\sup_{||\X||_{F}=1}||\mathcal{A}\X||_{F}$.\\
$\langle \X\,,\Y\rangle$ & $\textrm{Trace}(\X^T\Y)$ & $\lambda_{\max},\lambda_{\min}$ & Maximum, Minimum eigenvalue\\
$\bm{1}$ & A vector or matrix of ones & $\textrm{Sgn }\X$ & $\U\textrm{sign }(\Sigma) \V^T$; here $[\U,\Sigma,\V]= \text{svd}(\X)$\\
$\bm{0}$ & A vector or matrix of zeros & $\Omega$, $\mathbb{I}$ & Random sampled set, Universal set\\\hline
\end{tabular}
\end{table*}

\section{Dual basis formulation}
\label{sec:Dualbasis}
The aim of this section is to show that the EDG problem \eqref{eq:nnm_EDG} can be equivalently stated as a matrix completion problem with
respect to a special designed basis. The matrix completion problem is an instance of the low rank matrix recovery problem
where the goal is to recover a low rank matrix given random partial linear observations. The natural minimization problem for low rank recovery problem 
is rank minimization with linear constraints. 
Unfortunately, this problem is NP hard \cite{recht2010guaranteed}
motivating other solutions. In a series of remarkable theoretical papers \cite{candes2009exact,recht2010guaranteed,gross2011recovering,candes2010power,recht2011simpler},
it was shown that, under certain conditions, the solutions of the NP hard rank minimization problem can be
obtained by solving the convex nuclear norm minimization problem which is computationally tractable~\cite{cai2010singular,carmi2013compressed}. Our theoretical analysis is inspired by the
work~\cite{gross2011recovering}  where the author extends the nuclear norm minimization formulation to
recovering a low rank matrix given that a few of its coefficients in a fixed orthonormal basis
are known. 

\paragraph{Dual basis}
To write the EDG problem \eqref{eq:nnm_EDG} as a matrix completion problem with
respect to an appropriate operator basis, let us introduce few notations. 
We write $\Us =\{ \alphab=(\alpha_1,\alpha_2)~|~1\le \alpha_1<\alpha_2\le n\}$ and define  
\begin{equation}\label{eq:basis}
\w_{\alphab} = \e_{\alpha_{1},\,\alpha_{1}}+\e_{\alpha_{2},\,\alpha_{2}}-\e_{\alpha_{1},\,\alpha_{2}}-\e_{\alpha_{2},\,\alpha_{1}} 
\end{equation}
where $\e_{\alpha_1,\,\alpha_2}$ is a matrix whose entries are all
zero except a $1$ in the $(\alpha_1,\,\alpha_2)$-th entry. It is clear that $\{\w_{\alphab}\}$ forms a basis of the linear space
$\S =\{\X\in \real^{n\times n} ~|~ \X=\X^{T}\,\&\,\X\cdot \bm{1}=\bm{0} \}$ and 
the number of basis is $L = \frac{n(n-1)}{2}$.
For conciseness and ease in later analysis, we further define
$\Sp = \S \cap \{\X\in \real^{n\times n} ~|~ \X\succeq \bm{0}\}$. 
Therefore, for any given subsample $\Omega$ of $\Us$, the EDG problem \eqref{eq:nnm_EDG}
can be written as the following nuclear norm minimization problem
\begin{align}
\underset{\X\in \Sp}{\textrm{minimize }} \quad &||\X||_{*} \nonumber\\
  \textrm{subject to }\quad &\langle \w_{\alphab},\X \rangle = \langle \w_{\alphab}, \,\M \rangle \quad 
	\forall \alphab \in \Omega \label{eq:nnm_EDG_basis}
\end{align}
where $\M$ is the true underlying rank $r$ Gram matrix satisfying $\M_{i,\,i}+\M_{j,\,j}-2\M_{i,\,j}=\D_{i,\,j}$ $\forall (i,j)$.
The EDG problem can now be equivalently
interpreted as a matrix completion problem with respect to the basis  $\w_{\alphab}$.

By construction, $\w_{\alphab}$ is symmetric and satisfies $\w_{\alphab}\cdot \bm{1}=\bm{0}$.
The latter condition naturally enforces the constraint $\X\cdot \bm{1}=\bm{0}$.
It is clear that any $\X\in\S$ can be expanded in the basis 
$\{\w_{\alphab}\}$ as $\X = \sum_{\betab\in \Us}\c_{\betab}\w_{\betab}$. 
After minor algebra, $\c_{\betab} = \sum_{\alphab\in \Us} \langle \X\,,\w_{\alphab}\rangle\H^{\alphab,\,\betab}$
where we define $\H_{\alphab,\,\betab} = \langle \w_{\alphab}\,,\w_{\betab}\rangle$ 
and $\H^{\alphab,\,\betab} = \H^{-1}_{\alphab,\,\betab}$. Note that, since 
$\{\w_{\alphab}\}$ is a basis, the inverse $\H^{-1}$ is well defined.
This results in the following representation of $\X$.
\begin{equation}\label{eq:x_expansion}
\X = \sum_{\betab\in\Us} \sum_{\alphab\in\Us} \H^{\alphab,\,\betab}\langle \w_{\alphab}\,,\X\rangle \w_{\betab}
\end{equation}
The crux of the dual basis approach is to simply consider \eqref{eq:x_expansion}
and rewrite it as follows 
\begin{equation}\label{eq:x_expansion_dual}
\X = \sum_{\alphab\in\Us} \langle \X\,,\w_{\alphab}\rangle \v_{\alphab}
\end{equation}
where $\v_{\alphab}= \sum_{\betab} \H^{\alphab,\,\betab} \w_{\betab}$. It can
be easily verified that $\{\v_{\alphab}\}$ is a dual basis of $\{\w_{\alphab}\}$ satisfying $\langle \v_{\alphab}\,,\w_{\betab} \rangle =
\delta_{\alphab,\,\betab}$. Let $\W =[\w_1,\w_2,...,\w_L]$
and $\V = [\v_1,\v_2,...,\v_{L}]$ denote the matrix of vectorized basis matrices and vectorized dual basis matrices
respectively. The following basic relations are useful in later analysis, $\H = \W^T\W$, $\H^{-1} = \V^T\V$ and $\V = \W\H^{-1}$.
In the context of the EDG completion problem, the dual basis approach ensures that the expansion coefficients
match the measurements while preserving the condition that the matrix in
consideration is symmetric with zero row sums. With this,
\eqref{eq:nnm_EDG_basis} turns into a well formulated matrix completion problem with respect to
the basis $\w_{\alphab}$. An alternative way to rewrite \eqref{eq:nnm_EDG_basis}
makes use of the sampling operator defined as follows.
\begin{equation} \label{eq:dual_sampling}
\Po:\X \in \S\longrightarrow \frac{L}{m}\sum_{\alphab\in\Omega}\langle \X\,,\w_{\alphab} \rangle \v_{\alphab}
\end{equation}
where we assume that $\Omega$ are sampled uniformly at random from $\Us$ with replacement and 
$m = |\Omega|$ is the size of $\Omega$.
The scaling factor $\displaystyle \frac{L}{m}$ is for ease in later analysis.
It can be easily verified that $\displaystyle\frac{m^2}{L^2}\Po^{2} = \frac{m}{L}\Po$.
The adjoint operator of $\Po$, $\Po^{*}$, can be simply derived and is given by
\begin{equation} 
\Po^{*}:\X \in \S\longrightarrow \frac{L}{m}\sum_{\alphab\in\Omega} \langle \X\,,\v_{\alphab} \rangle \w_{\alphab}
\end{equation}
Using the sampling operator $\Po$, we can write $\eqref{eq:nnm_EDG_basis}$ as follows 
\begin{align}
&\underset{\X\in \Sp}{\textrm{minimize }} \quad \|\X\|_{*} \nonumber\\
&\textrm{subject to }\quad \Po(\X) = \Po(\M)  \label{eq:m_completion}
\end{align}
%For the convenience of later analysis, we also use the following alternative representation of \eqref{eq:dual_sampling}. Let $\x$ denote the vectorized form of $\X$.
%Introduce a sampling matrix $\mathcal{S}\in \real^{L\times L}$,
%a diagonal matrix of $0$'s and $1$'s with the $1$'s 
%located at the indices belonging to $\Omega$. With this, a compact form of the sampling operator is given by
%\begin{equation} \label{eq:dual_sampling_compact}
%\Po(\x) = \frac{L}{m} \V\S\W^{T}\x
%\end{equation}
%Here on, the two forms \eqref{eq:dual_sampling} and \eqref{eq:dual_sampling_compact}
%will be used interchangeably depending on the context. 
In addition to the sampling operator,
the restricted frame operator appears in the analysis and is defined as follows
\begin{equation} \label{eq:frame_operator}
\F:\X \in \S \longrightarrow \frac{L}{m}\sum_{\alphab\in\Omega}\langle \X\,,\w_{\alphab} \rangle \w_{\alphab}
\end{equation}
Note that the restricted frame operator is self-adjoint and positive semidefinite. 

\paragraph{Coherence}

Pathological situations can arise in \eqref{eq:m_completion} when
$\M$ has very few non-zero coefficients. In the context of EDG,
in extreme cases, this happens if only the diagonal elements of $\D$ are  sampled and/or
there are many overlapping points. This motivates the
notion of coherence first introduced in \cite{candes2009exact}.
Let $\M = \sum_{k=1}^{r} \lambda_{k} \u_{k} \u_{k}^{T}$,
where the eigenvectors have been chosen to be orthonormal and $\lambda_k \geq 0 $ as $\M\in\Sp$. 
We write $\U = \textrm{span } \{\u_1,...,\u_r\}$ as the column space
of $\M$, $\U^{\perp} = \textrm{span} \{\u_{r+1} , . . . , \u_n \}$ as the orthogonal complement of $\U$ and further 
denote $\P_{\U}$ and $\P_{\U^{\perp}}$ as the orthogonal projections
onto $\U$ and $\U^{\perp}$, respectively. 
Let $\T=\{\U\Z^T+\Z\U^{T}:\Z\in \real^{n\times r} \}$ be the tangent space of the rank $r$ matrix in $\Sp$ at $\M$.
The orthogonal projection onto $\T$ is given by
\begin{equation} \label{eq:pt_defn}
\P_{\T}\,\X = \P_{\U}\X + \X \P_{\U} - \P_{\U} \X \P_{\U}
\end{equation}
It can be readily verified that $\P_{\T^{\perp}}\,\X = \X - \P_{\T}\X = \P_{\U^{\perp}}\X\P_{\U^{\perp}}$.
The coherence conditions can now be defined as follows. 
\begin{definition}
The aforementioned rank r matrix $\M\in \Sp$ has coherence $\nu$ with respect to basis $\{\w_{\alphab}\}_{\alphab\in \Us}$ if
the following estimates hold
\begin{align}
\underset{\alphab\in \Us}{\max} \,\, &\sum_{\betab\in \Us}\langle \P_{\T}\, \w_{\alphab}\,,\w_{\betab} \rangle^{2} \le 2\nu \frac{r}{n}
\label{eq:coherencew1}\\
\underset{\alphab\in \Us}{\max} \,\, &\sum_{\betab\in \Us}\langle \P_{\T}\, \v_{\alphab}\,,\w_{\betab} \rangle^{2} \le 4\nu \frac{r}{n}
\label{eq:coherencev1}\\
\underset{\alphab\in \Us}{\max} \,\, &\langle \w_{\alphab}\,,\U\U^{T}\rangle^{2} \le \frac{1}{4}\nu \frac{r}{n^{2}} \label{eq:jointcoherence1}
\end{align}
where $\{\v_{\alphab}\}_{\alphab\in \Us}$ is the dual basis of $\{\w_{\alphab}\}_{\alphab\in \Us}$.
\end{definition}

\noindent
\begin{remark}
If $\{\w_{\alphab}\}$ is an orthonormal basis, it follows trivially that the dual basis
$\{\v_{\alphab}\}$ is also $\{\w_{\alphab}\}$. For this specific case, the above coherence
conditions are equivalent, up to constants, to the coherence conditions in
\cite{gross2011recovering}. However, in the most general setting, 
the coherence conditions \eqref{eq:coherencew1} and \eqref{eq:coherencev1} depart from their
orthonormal counterparts. This is because these conditions depend on the spectrum
of the correlation matrix $\H$. Since the analysis in the sequel makes repeated use
of the coherence conditions, the above equations are further simplified to convenient forms
as allows. First, using \eqref{eq:coherencew1}, consider a bound on $||\P_{\T}\, \w_{\alphab}||_{F}^{2}$. Using Lemma \ref{H_prop}
and Lemma \ref{norm_equivalent}, one obtains 
\begin{equation*} 
||\P_{\T}\, \w_{\alphab}||_{F}^{2}\le 
\underset{\alphab\in \Us}{\max} \,\, \sum_{\betab\in \Us}\langle \P_{\T}\, \w_{\alphab}\,,\w_{\betab} \rangle^{2} \le 2\nu \frac{r}{n}
\longrightarrow ||\P_{\T}\, \w_{\alphab}||_{F}^{2}\le 2\nu \frac{r}{n}
\end{equation*}
Using the above inequality and Lemma \ref{H_prop}, one obtains the following bound
for $||\P_{\T}\, \v_{\alphab}||_{F}$.
\begin{equation*}
||\P_{\T}\,\v_{\alphab}||_{F} \le \sum_{\betab\in \Us} ||\H^{\alphab,\,\betab}\,\P_{\T}\, \w_{\betab}||_{F} = 
\sum_{\betab\in \Us} |\H^{\alphab,\,\betab}|\,\,||\P_{\T}\, \w_{\betab}||_{F} \le 2\sqrt{\frac{2\nu r}{n}}                   
  \end{equation*}
	Next, consider a bound on $\langle \v_{\alphab}\,,\U\U^{T}\rangle^{2} $.
	Using \eqref{eq:jointcoherence1} and Lemma \ref{H_prop}, one arrives at the
	following inequality.
	\[
 \langle \v_{\alphab}\,,\U\U^{T}\rangle^{2}	= \langle \sum_{\betab\in \Us} \H^{\alphab,\,\betab}\w_{\betab}\,,\U\U^{T}\rangle^{2}	\le
\underset{\betab\in \Us}{\max} \,\, \langle \w_{\betab}\,,\U\U^{T}\rangle^{2}\,\left(\sum_{\betab\in \Us} |\H^{\alphab,\,\betab}|\right)^{2}
\le \frac{\nu r}{n^2}
	\]
All in all, the simplified forms of the coherence conditions are summarized as follows.
\begin{align}
\underset{\alphab\in \Us}{\max} \,\, &||\P_{\T}\, \w_{\alphab}||_{F}^{2} \le 2\nu \frac{r}{n}\label{eq:coherencew}\\
\underset{\alphab\in \Us}{\max} \,\, &||\P_{\T} \, \v_{\alphab}||_{F}^{2} \le 8\nu \frac{r}{n}\label{eq:coherencev}\\
\underset{\alphab\in \Us}{\max} \,\, &\langle \v_{\alphab}\,,\U\U^{T}\rangle^{2} \le \frac{\nu r}{n^{2}} \label{eq:jointcoherence}
\end{align}
\end{remark}

\noindent
The simplified coherence conditions presented above are the same as the standard coherence assumptions up to constants 
(see \cite{gross2011recovering,recht2011simpler}). If the matrix $\M$ has coherence $\nu$ with respect to the standard basis, comparable bounds
could be derived for the above coherence conditions. This is true since the basis $\{\w_{\alphab}\}$ is not ``far'' from the standard basis.
For a precise statement of this fact, we refer the interested reader to Lemma \ref{coherence_conditions_derivation}.
The implication of this fact is that an incoherent $\M$ with respect to the standard basis is also incoherent with respect to the EDG basis. Intuitively, the coherence parameter is fundamentally about the concentration of information in the underlying matrix. For a matrix with low coherence(incoherent), each measurement
is equally as informative as the other. In contrast, the information is concentrated on few measurements for a coherent matrix. 

\begin{remark}[Coherence and Geometry]
In the context of the EDG problem, an interesting question is the relationship, if any, between coherence and the geometry of
the underlying points. One analytical task is to relate the coherence, using the compressive sensing definition, of the
Gram matrix to the coherence conditions \eqref{eq:coherencew1}-\eqref{eq:jointcoherence1}. 
Computationally, numerical tests indicate that points arising from random distributions or points which are
a random sample of structured points have comparable coherence values irrespective of the underlying distribution. However, since estimating coherence
accurately and efficiently for large $n$ is computationally expensive, the numerical tests were limited to relatively small values of $n$.
The question of the geometry of points and how it relates to coherence is important to us and future work will explore
this problem through analysis and extensive numerical tests. 

\end{remark}

\paragraph{Sampling Model} The sampling scheme is an important element of the matrix
completion problem. For the EDG completion problem, it is assumed that the basis
vectors are sampled uniformly at random with replacement. Previous
works have considered uniform sampling with replacement \cite{gross2011recovering,recht2011simpler}
and Bernoulli sampling \cite{candes2009exact}. The implication of our choice is that
the sampling process is independent. This property is essential as we make use of concentration inequalities for i.i.d matrix valued random variables. 

\begin{remark}
For the EDG completion problem, we have considered uniform sampling with out replacement. In this case, the sampling process is no longer independent. The implication is that one could not simply use concentration inequalities for i.i.d matrix valued random
variables. However, as noted in the classic work of Hoeffding \cite[section 6]{hoeffding1963probability}, the results derived for the case of the sampling with replacement
also hold true for the case of sampling without replacement. An analogous 
argument for matrix concentration inequalities, resulting from the operator Chernoff bound technique
\cite{ahlswede2002strong}, under uniform sampling with out replacement is shown in \cite{gross2010note}.
The analysis based on this choice leads to a ``slightly lower'' number of measurements $m$. Since
the gain is minimal, which will be made apparent in later analysis, we use the uniform sampling with
replacement model.   
\end{remark}

\section{Proof of Main Result}
\label{sec:Proof}
The main goal of this section is to show that the nuclear norm minimization problem~\eqref{eq:m_completion} admits a unique solution. Thus, the optimization problem provides the desired reconstruction.
%Since $\X\in \Sp$, the nuclear norm of $\X$ equates to its trace and \eqref{eq:m_completion}
%can be written as a trace minimization problem.
%\begin{align}
%&\underset{\X\in \Sp}{\textrm{minimize }} \quad  \Tr(\X) \nonumber\\
%&\textrm{subject to }\quad \Po(\X) = \Po(\M)  \label{eq:m_completion_trace}
%\end{align}
Under certain conditions, our main result guarantees a unique solution for the minimization problem \eqref{eq:m_completion}. 
More precisely, we have the following theorem. 
\begin{theorem}\label{thm1}
Let $\M\ \in \real^{n\times n}$ be a matrix of rank $r$ that obeys the coherence conditions
\eqref{eq:coherencew1}, \eqref{eq:coherencev1} and \eqref{eq:jointcoherence1} with coherence 
$\nu$. 
Assume $m$ measurements, $\{\langle \M\,,\w_{\alphab}\rangle\}_{\alphab\in\Omega}$, are sampled uniformly at random with replacement.
For $\beta>1$, if
  \[
  m\ge nr \log_{2}\left(4\sqrt{8Lr}n\right)\left[96\left(\nu+\frac{1}{nr}\right)\left(\beta\log(n)+\log\left(4\log_{2}(4L\sqrt{r})\right)\right)\right]
 \]
 the solution to  \eqref{eq:m_completion}, equivalently \eqref{eq:nnm_EDG} and \eqref{eq:nnm_EDG_basis},
is unique and equal to $\M$ with probability at least $1-n^{-\beta} $. 
\end{theorem}

\noindent

The optimization problem in \eqref{eq:m_completion} is a convex minimization problem for which
the optimality of a feasible solution $\X$ is equivalent to the sufficient KKT conditions.
For details, we refer the interested reader to \cite{candes2009exact} where the
the authors derive a compact form of these optimality
conditions. The over all structure of the proof is as follows. 
First, we show that if certain deterministic optimality and uniqueness conditions hold, 
then it certifies that $\M$ is a unique solution to the minimization problem.
The remaining part of the proof will focus on 
showing that these conditions hold with very high probability
under certain conditions. This in turn will imply that $\M$ is a
unique solution with very high probability for an appropriate choice
of $m$.  

The proof of Theorem \ref{thm1} closely follows the approach in \cite{gross2011recovering}. For ease of
presentation, the proof is divided into several intermediate results.
Readers familiar with matrix completion proof in \cite{gross2011recovering,recht2011simpler}
will recall that one crucial part of the proof is bounding the spectral norm of
$||\P_{\T}\Po\P_{\T}-\P_{\T}||$. The interpretation of this bound is that the operator
$\P_{\T}\Po\P_{\T}$ is almost isometric to $\P_{\T}$. In the case of
our problem, $\Po$ is no longer self-adjoint and the equivalent 
statement is a bound on $||\P_{\T}\Po^{*}\Po\P_{\T}-\P_{\T}||$.
Unfortunately, the term $||\P_{\T}\Po^{*}\Po\P_{\T}-\P_{\T}||$ is
not amenable to simpler analysis as standard concentration inequalities result sub-optimal
success rates. The idea is instead to consider the operator $\P_{\T}\F\P_{\T}$ and show
that the minimum eigenvalue is bounded away from $0$ with high probability using the operator Chernoff
bound. The implication is that 
the operator $\P_{\T}\F\P_{\T}$ is full rank on the space $\T$.  For non orthogonal measurements, 
$\P_{\T}\F\P_{\T}$ can be interpreted as the analogue of the operator $\P_{\T}\Po\P_{\T}$.  In \cite{candes2010power},
for the standard matrix completion problem with measurement basis $\bm{e}_{ij}$, it
was argued theoretically that a lower bound for $m$ is $O(nr\nu\log\,n)$. 
Theorem \ref{thm1} requires on the order of $nr\nu\log^{2}\,n$ measurements
which is $\log(n)$ multiplicative factor away from
this sharp lower bound. We remark that, despite the aforementioned technical challenges, our result is of the same order as the result in \cite{gross2011recovering,recht2011simpler} which consider the low rank recovery problem with any orthogonal basis and the matrix completion problem respectively. 
We start the proof by stating and proving Theorem \ref{uniqueness_conditions} which 
rigorously shows that if appropriate conditions as discussed earlier hold, then $\M$ is a unique
solution to  \eqref{eq:m_completion}.

\begin{theorem}\label{uniqueness_conditions}
Given $\X\in \Sp$, we write $\bDelta = \X - \M$ as the deviation from the underlying low rank matrix $\M$. Let $\bDelta_{\T}$ and $\bDelta_{\T^{\perp}}$ be the orthogonal projection of
 $\bDelta$ to $\T$ and $\T^{\perp}$, respectively. For any given $\Omega$ with $|\Omega|=m$, the following two statements hold.
\begin{enumerate}
\item[(a).]
If $\displaystyle ||\bDelta_{\T}||_{F}^{2} \ge \left(n\sqrt{8L}  
||\bDelta_{\T^{\perp}}||_{F}\right)^{2}$ and $ \lambda_{\min}\left(\P_{\T}\,\F\,\P_{\T}\right) >\frac{1}{2}$, then $\Po \bDelta \neq 0$.
\item [(b).]
  If $\displaystyle ||\bDelta_{\T}||_{F}^{2} < \left(n\sqrt{8L}||\bDelta_{\T^{\perp}}||_{F}\right)^{2}$
  for $\bDelta \in \ker \Po$, and there exists a  $\Y \in \textrm{range}\,\, \Po^{*}$ satisfying,
\begin{equation}\label{eq:yconds}
||\P_{\T}\Y-\textrm{Sgn}\,\M||_{F} \le \frac{1}{4n\sqrt{8L} }\quad \textrm {and} \quad ||\P_{\T^{\perp}}\Y||\le \frac{1}{2}
\end{equation}
then $\|\X\|_{*}= \|\M + \bDelta\|_{*}> \|\M\|_{*}$.

\end{enumerate}
\end{theorem}

\noindent
To interpret the above theorem, note
that Theorem $\ref{uniqueness_conditions}$(a) states
that any deviation from $\M$
fails to be in the null space of the sampling operator for ``large" $\bDelta_{\T}$.
For ``small" $\bDelta_{\T}$, Theorem $\ref{uniqueness_conditions}$(b) claims that
any deviation from $\M$ increases the nuclear norm
thereby violating the minimization condition. We would like to emphasize that this is a deterministic theorem which does not depend on the random sampling of $\Omega$ as long as the conditions in the statement are satisfied. In fact, we will later show that these conditions will hold with very high probability 
under certain sampling conditions and an appropriate choice of $m = |\Omega|$.

\subsection{Proof of Theorem \ref{uniqueness_conditions}}
\begin{proof} [Proof of Theorem \ref{uniqueness_conditions}(a)]
Suppose $\|\bDelta_{\T}\|_{F}^{2} \ge \left(n\sqrt{8L} \|\bDelta_{\T^{\perp}}\|_{F}\right)^{2}$. 
To prove $\Po \bDelta \neq 0$, it suffices to show $\|\Po\bDelta\|_{F} >0$. 
Note that $\|\Po\,\bDelta\|_{F} =\|\Po\,\bDelta_{\T} + \Po\,\bDelta_{\T^{\perp}}\|_{F} \ge \| \Po\, \bDelta_{\T}\|_{F}-
\| \Po\,\bDelta_{\T^{\perp}}\|_{F}$.
This motivates finding a lower bound for $||\Po\, \bDelta_{\T}||_{F}$ and an upper bound for
$||\Po\,\bDelta_{\T^{\perp}}||_{F}$. For any $\X$, $||\Po\,\X||_{F}^{2}$ can be bounded as follows
\begin{equation*}
||\Po\,\X||_{F}^{2} = \langle \Po\,\X\,,\Po\,\X\rangle =  \langle \X\,,\Po^{*}\Po\,\X\rangle= 
\frac{L^2}{m^2} \sum_{\betab\in \Omega}\sum_{\alphab \in \Omega}\langle \X\,,\w_{\alphab}\rangle \langle \v_{\alphab}\,,\v_{\betab}\rangle\langle \X\,, \w_{\betab}\rangle
= \frac{L^2}{m^2}\sum_{\betab\in \Omega}\sum_{\alphab \in \Omega}\langle \X\,,\w_{\alphab}\rangle \H^{\alphab,\,\betab} \langle \X\,, \w_{\betab}\rangle
\end{equation*}
noting that $\langle \v_{\alphab}\,,\v_{\betab}\rangle= \H^{\alphab,\,\betab}$. Using the min-max theorem in the above equation, one obtains 
\begin{equation}\label{eq:minmax_bound}
\frac{L^2}{m^2}\lambda_{\min}(\H^{-1})\sum_{\alpha \in \Omega} \langle \X\,,\w_{\alphab}\rangle ^{2} \le ||\Po\,\X||_{F}^{2} \le \frac{L^2}{m^2}\lambda_{\max}(\H^{-1})\sum_{\alpha \in \Omega} \langle \X\,,\w_{\alphab}\rangle ^{2}
\le \frac{L^2}{m^2}\lambda_{\max}(\H^{-1})\sum_{\alpha\in \Us } \langle \X\,,\w_{\alphab}\rangle ^{2}
\end{equation}
Using the fact $\lambda_{\max}(\H^{-1})= 1$ (Lemma \ref{H_prop}) and $\sum_{\alpha\in \Us } \langle \X\,,\w_{\alphab}\rangle ^{2}
\leq 2n\|\X\|_F^2$ (Lemma \ref{norm_equivalent}) and setting
$\X = \bDelta_{\T^{\perp}}$ in the right of the above inequality results
\begin{equation} \label{eq:pdelta_perp_upper_bound}
||\Po\,\bDelta_{\T^{\perp}}||_{F}^{2} \le \frac{2L^2}{m^2}nm  ||\bDelta_{\T^{\perp}}||_{F}^{2}
\end{equation}
In the above inequality, the constant $m$ bounds the maximum number of repetitions for any given measurement.
Next, the lower bound for  $||\Po\,\bDelta_{\T}||_{F}$ is considered using the left inequality in \eqref{eq:minmax_bound}.
\begin{equation} \label{eq:pdelta_lower_bound_init}
||\Po\,\bDelta_{\T}||_{F}^{2} \ge 
  \frac{L^2}{2m^2n} \sum_{\alphab \in \Omega} \, \langle \bDelta_{\T}\,,\w_{\alphab}\rangle ^{2} = \frac{L}{2mn}   \langle \bDelta_{\T}\,,\frac{L}{m}\sum_{\alphab \in \Omega}\langle\bDelta_{\T}\,,\w_{\alphab}\rangle\w_{\alphab}\rangle = \frac{L}{2mn}  \langle \bDelta_{\T}\,,\F\,\bDelta_{\T}\rangle 
  \end{equation}
with the fact that $\lambda_{\min}(\H^{-1})= \displaystyle\frac{1}{2n}$
	(Lemma \ref{H_prop}).
Using the min-max theorem on the projection onto $\T$ of the restricted frame operator, we have
\begin{align} 
||\Po\,\bDelta_{\T}||_{F}^{2} \ge 
    \frac{L}{2mn}  \langle \bDelta_{\T}\,,\F\,\bDelta_{\T}\rangle   
  =  \frac{L}{2mn}  \langle \bDelta_{\T}\,,\P_{\T}\,\F\,\P_{\T}\,\bDelta_{\T}\rangle    \ge \frac{L}{2mn}  \lambda_{\min}(\P_{\T}\,\F\,\P_{\T})\,||\bDelta_{T}||_{F}^{2}    \label{eq:pdelta_lower_bound}
  \end{align}
Above, the first equality holds since $\P_{\T}\,\bDelta_{\T} = \bDelta_{\T}$ and $\P_{\T}$ is self-adjoint. The last inequality simply applies
the min-max theorem on the self-adjoint operator $\P_{\T}\,\F\,\P_{\T}$. Using the assumption that 
$ \lambda_{\min}\left(\P_{\T}\,\F\,\P_{\T}  \right)>\frac{1}{2}$,
the inequality in \eqref{eq:pdelta_lower_bound} reduces to
\begin{equation} \label{eq:pdelta_lower_bound_final_estimate}
||\Po\,\bDelta_{\T}||_{F}^{2} > \frac{L}{4mn} || \bDelta_{\T}||_{F}^{2}
\end{equation}
Combining \eqref{eq:pdelta_perp_upper_bound} and \eqref{eq:pdelta_lower_bound_final_estimate}, we have
\begin{align*}
\|\Po\,\bDelta\|_{F}  &> 
\sqrt{\frac{L}{4mn}} ||\bDelta_{\T}||_{F} - \frac{L}{m}\sqrt{2n}
||\bDelta_{\T^{\perp}}||_{F}\\
&\ge \sqrt{\frac{L}{4mn}} \left( n\sqrt{8L}\right)|| \bDelta_{\T^{\perp}}||_{F} -
\frac{L}{m} \sqrt{2nm}  ||\bDelta_{\T^{\perp}}||_{F} = 0
\end{align*}
where the last inequality follows from the theorem's assumption. This concludes proof of Theorem \ref{uniqueness_conditions}(a).
\end{proof}
\begin{remark}
($1$) The estimate of the upper bound of $||\Po\,\bDelta_{\T^{\perp}}||_{F}^{2}$ is not sharp. Specifically,
the constant $m$ can be much lowered. This can be achieved, for example, by employing standard concentration
inequalities and arguing that the expected number of duplicate measurements is substantially smaller than $m$.
Alternatively, as noted earlier in the sampling model section, uniform sampling with out replacement can be adopted
which has the implication that the constant $m$ is not necessary. Both analysis lead to a better estimate of the upper bound. However, the gain in terms of measurements, $m$, is
minimal, resulting lower constants inside of a $\log_2$. We use the current estimate, albeit not sharp, for sake of more compact analysis.
($2$) Lower bounding the term 
$\sum_{\alphab \in \Omega} \, \langle \bDelta_{\T}\,,\w_{\alphab}\rangle ^{2}$
is not amenable to simpler analysis. In fact, a mere application of the standard concentration inequalities result an undesirable probability of failure that increases with $n$. 
Note that 
$||\Po\,\bDelta_{\T}||_{F}^{2} = 
\langle \bDelta_{\T}\,,\P_{\T}\Po^{*}\Po\P_{\T}\bDelta_{\T}-\P_{\T}\bDelta_{\T}\rangle+
||\P_{\T}\bDelta_{\T}||_{F}^{2}$. A lower bound for $||\Po\,\bDelta_{\T}||_{F}^{2}$
motivates an upper bound of $||\P_{\T}\Po^{*}\Po\P_{\T}\bDelta_{\T}-\P_{\T}||$.
This bound can be achieved, albeit long calculations, but it relies on the special structure of
$\w_{\alphab}$. To make the technical analysis simple and more general, we use an alternative approach shown in the above proof. This is a different way to handle the lower bound estimation of 
$||\Po\,\bDelta_{\T}||_{F}^{2}$.
\end{remark}

\begin{proof} [Proof of Theorem \ref{uniqueness_conditions}(b)]
 Consider a feasible solution $\X = \M + \bDelta \in \Sp$ satisfying $||\bDelta_{\T}||_{F} < n\sqrt{8L} \|\bDelta_{\T^{\perp}}\|_{F}$ and 
$\bDelta \in \ker \Po$. 
We need to show that $\|\M + \bDelta\|_{*}> \|\M\|_{*}$ which implies that nuclear minimization is violated. The proof of this requires the 
construction of a dual certificate $\Y$ satisfying certain conditions. The proof is similar to the proof
in section $2$E of \cite{gross2011recovering} but it is shown below for completeness and ease of reference in later analysis.
Using the pinching inequality \cite{bhatia2013matrix}, $ \|\M+\bDelta\|_{*} \ge ||\P_{\U}(\M+\bDelta)\P_{\U}||_{*}+||\P_{\U^{\perp}}(\M+\bDelta)\P_{\U^{\perp}}||_{*}$. 
Noting that $\P_{\U}\M\P_{\U} =\M$, $\P_{\U^{\perp}}\M\P_{\U^{\perp}}=\bm{0}$ and $\P_{\U^{\perp}}\bDelta \P_{\U^{\perp}} = \bDelta_{\T^{\perp}}$, the above inequality reduces to
\begin{equation}\label{eq:ineq2}
\|\M+\bDelta\|_{*} \ge ||\M+\P_{\U}\bDelta\P_{\U}||_{*}+||\bDelta_{\T^{\perp}}||_{*}
\end{equation}
Note that $||\bDelta_{\T^{\perp}}||_{*} = \langle \textrm{Sgn}\,\bDelta_{\T^{\perp}}\,,\bDelta_{\T^{\perp}} \rangle$. 
The first term on right hand side can be lower bounded using the fact that the nuclear norm and the
spectral norm are dual to one another. Stated precisely, $||\bm{X}_{2}||_{*}= \underset{||\bm{X}_1||\le 1}{\sup}
\langle \bm{X}_{1}\,,\bm{X}_{2}\rangle$ for any $\bm{X}_{1}$. 
Using this inequality with  $\bm{X}_{1} =\textrm{Sgn}\,\M$ and $\bm{X}_{2} = \M + \P_{\U}\bDelta\P_{\U}$ in \eqref{eq:ineq2} results
\begin{align*}
||\M+\P_{\U}\bDelta\P_{\U}||_{*}+||\bDelta_{\T^{\perp}}||_{*} &\ge 
\langle \textrm{Sgn}\,\M \,,\M+\P_{\U}\bDelta\P_{\U} \rangle + 
\langle \textrm{Sgn}\,\bDelta_{\T^{\perp}}\,,\bDelta_{\T^{\perp}} \rangle\\
&= ||\M||_{*}+\langle \textrm{Sgn}\,\M\,,\P_{\U}\bDelta\P_{\U} \rangle+
\langle \textrm{Sgn}\,\bDelta_{\T^{\perp}}\,,\bDelta_{\T^{\perp}} \rangle
\end{align*}
Note that $\textrm{Sgn}\,\M \in \T$ (see Lemma \ref{signx_prop}). 
Using this fact, it can be verified that $\langle \textrm{Sgn}\,\M\,,\P_{\U}\bDelta_{\T^{\perp}}\P_{\U} \rangle = 0$ 
and $\langle \textrm{Sgn}\,\M\,,\P_{\U}\bDelta_{\T}\,\P_{\U} \rangle = \langle \textrm{Sgn}\,\M \,,\bDelta_{\T} \rangle$.
The above inequality can now be written as
\begin{align}
||\M+\P_{\U}\bDelta\P_{\U}||_{*}+||\bDelta_{\T^{\perp}}||_{*}\ge  
||\M||_{*}+\langle \textrm{Sgn}\,\M\,,\bDelta_{\T}\rangle
+\langle \textrm{Sgn}\,\bDelta_{\T^{\perp}}\,,\bDelta_{\T^{\perp}} \rangle  \label{eq:ineq4}
\end{align}
Using \eqref{eq:ineq2} and \eqref{eq:ineq4}, it can be concluded that $||\M+\bDelta||_{*}
> ||\M||_{*} $ if it can be shown that $\langle \textrm{Sgn}\,\M\,,\bDelta_{\T}\rangle
+\langle \textrm{Sgn}\,\bDelta_{\T^{\perp}}\,,\bDelta_{\T^{\perp}}\rangle > 0$.
Since $\langle \Y\,,\bDelta \rangle = 0$ for any $\Y \in \textrm{range}\,\, \Po^{*}$, we have
\begin{align*}
\langle \textrm{Sgn}\,\M\,,\bDelta_{\T}\rangle
+\langle \textrm{Sgn}\,\bDelta_{\T^{\perp}}\,,\bDelta_{\T^{\perp}} \rangle
&= \langle \textrm{Sgn}\,\M\,,\bDelta_{\T}\rangle
+\langle \textrm{Sgn}\,\bDelta_{\T^{\perp}}\,,\bDelta_{\T^{\perp}} \rangle-\langle \Y\,,\bDelta\rangle\\
&= \langle \textrm{Sgn}\,\M - \P_{\T}\Y\,,\bDelta_{\T} \rangle + 
\langle  \textrm{Sgn}\, \bDelta_{\T^{\perp}}\,,\bDelta_{\T^{\perp}}\rangle -
\langle \P_{\T^{\perp}}\Y\,,\bDelta_{\T^{\perp}} \rangle 
\end{align*}
Assuming the conditions in the statement of the theorem and considering
the last equation, we obtain
\begin{align*}
\langle \textrm{Sgn}\,\M\,,\bDelta_{\T}\rangle+\langle \textrm{Sgn}\,\bDelta_{\T^{\perp}}\,,\bDelta_{\T^{\perp}} \rangle
&\ge  -\frac{1}{4n\sqrt{8L}} ||\bDelta_{\T}||_{F}
+||\bDelta_{\T^{\perp}}||_{*}-\frac{1}{2}||\bDelta_{\T^{\perp}}||_{*}
\\
&\ge -\frac{1}{4n\sqrt{8L}} ||\bDelta_{\T}||_{F}
+\frac{1}{2}||\bDelta_{\T^{\perp}}||_{F}>  -\frac{1}{4n\sqrt{8L}} \left( n\sqrt{8L} \|\bDelta_{\T^{\perp}}\|_{F}    \right)
+\frac{1}{2}||\bDelta_{\T^{\perp}}||_{F}=\frac{1}{4}||\bDelta_{\T^{\perp}}||_{F}
\end{align*}
Above, the first inequality follows from the duality of the spectral norm and the nuclear norm.
It has been shown that $\langle \textrm{Sgn}\,\M\,,\bDelta_{\T}\rangle
+\langle \textrm{Sgn}\,\bDelta_{\T^{\perp}}\,,\bDelta_{\T^{\perp}} \rangle > 0$
under the assumptions of Theorem \ref{uniqueness_conditions}(b). Using \eqref{eq:ineq2} and \eqref{eq:ineq4}, 
it follows that $\|\M+\bDelta\|_{*} > \|\M\|_{*}$ concluding the proof of Theorem \ref{uniqueness_conditions}(b).
\end{proof}

Consequently, if the deterministic conditions in Theorem \ref{uniqueness_conditions} hold,  $\M$ is a unique solution to \eqref{eq:m_completion}. We formally write it as a corollary which can help us to prove the probabilistic statement in Theorem \ref{thm1}.
\begin{corollary}\label{uniqueness_corollary}
If the conditions of Theorem \ref{uniqueness_conditions} hold, $\M$ is a unique solution to \eqref{eq:m_completion}.
\end{corollary}
\begin{proof}
For any $\X \in \Sp$, define $\bDelta= \M-\X$. Using Theorem \ref{uniqueness_conditions}(a), $\Po\,\bDelta\neq \bm{0}$
if $||\bDelta_{\T}||_{F}^{2} \ge \left(n\sqrt{8L} ||\bDelta_{\T^{\perp}}||_{F}\right)^{2}$.
It then suffices to consider the case $||\bDelta_{\T}||_{F}^{2} < \left(n\sqrt{8L} ||\bDelta_{\T^{\perp}}||_{F}\right)^{2}$
for $\bDelta \in \ker \Po$. For this case, the proof of Theorem \ref{uniqueness_conditions}(b) 
shows that $\|\X\|_{*}> \|M\|_{*}$. It can then be concluded that $\M$ is the unique solution to \eqref{eq:m_completion}.
\end{proof}

\subsection{Proof of Theorem \ref{thm1}}
It follows that certifying that the two conditions in Theorem \ref{uniqueness_conditions} hold implies a unique solution to \eqref{eq:m_completion}.
Hence, the main task in the proof is to show that, under the assumptions of Theorem \ref{thm1}, 
these two conditions hold with very high probability. 
If this can be achieved, it implies that the conclusion of Theorem \ref{thm1} holds true with the same high probability. 
The first condition in Theorem \ref{uniqueness_conditions} is that 
$ \lambda_{\min}\left(\P_{\T}\,\F\,\P_{\T}\right)>\frac{1}{2}$.
The goal is to show that the minimum eigenvalue of the operator $\P_{\T}\,\F\,\P_{\T}$ is bounded away from zero.
This will be proven in Lemma \ref{min_eig}  to follow shortly. The lemma makes use of the matrix Chernoff bound in \cite{tropp2012user}
and is restated below for convenience. 
\begin{theorem}
Let $\{\Lc_{k}\}$ be a finite sequence of independent, random, self-adjoint operators satisfying
\[
\Lc_{k}\succeq \bm{0} \quad \textrm{ and } \quad ||\Lc_{k}|| \le R \quad \textrm{almost surely}
\] 
with the minimum and maximum eigenvalues of the sum of the expectations,
\[
\mu_{\min}:=\lambda_{\min} \left(\sum_{k} E[\Lc_{k}]\right) \quad \textrm{ and } \quad \mu_{\max}:=\lambda_{\max} \left(\sum_{k} E[\Lc_{k}]\right)
\]
Then, we have
\[
 Pr\,\bigg[ \lambda_{\min}\left(\sum_{k}\Lc_{k}\right) \le (1-\delta)\,\mu_{\min}\bigg] \le n\,\bigg[ \frac{\exp(-\delta)}{(1-\delta)^{1-\delta}}  \bigg]^{\frac{\mu_{\min}}{R}}   
\quad \textrm{for } \delta\in [0,1]
\]
\end{theorem}

\noindent
For $\delta \in [0,1]$, using Taylor series of $\log(1-\delta)$, note that $(1-\delta)\log(1-\delta) \ge -\delta+\frac{\delta^2}{2}$. 
This results the following simplified estimate
\[
Pr\,\bigg[ \lambda_{\min}\left(\sum_{k}\Lc_{k}\right) \le (1-\delta)\,\mu_{\min}\bigg] \le n\,\exp\left(-\delta^2\,\frac{\mu_{\min}}{2R}\right) 
\quad \textrm{for } \delta\in [0,1]
\]

\begin{lemma} \label{min_eig}
Given the operator $\P_{\T}\,\F\,\P_{\T}: \T \rightarrow \T$ with $\kappa = \frac{m}{nr}$, the following estimate holds. 
\[
\textrm{Pr}\,\left( \lambda_{\min}\left(\P_{\T}\,\F\,\P_{\T}\right)\le \frac{1}{2}\right)
\le n\exp\left(-\frac{\kappa}{8\nu}\right) 
\] 
\end{lemma}

\begin{proof}
Using the definition of $\F(\X) =  \frac{L}{m}\sum_{\alphab\in\Omega}\langle \X\,,\w_{\alphab} \rangle \w_{\alphab}$, for $\X \in \T$, $\P_{\T}\,\F\,\P_{\T}\,\X$ can be written as follows
\[
\P_{\T}\,\F\,\P_{\T}\,\X = \sum_{\alphab \in \Omega} \frac{L}{m}\langle \X\,,\P_{\T}\,\w_{\alphab}\rangle\, \P_{\T}\,\w_{\alphab}
\]
Let $\Lc_{\alphab}$ denote the operator in the summand, $\Lc_{\alphab} =\frac{L}{m}\langle \cdot\,,\P_{\T}\,\w_{\alphab}\rangle\, \P_{\T}\,\w_{\alphab}$. 
It can be easily verified that $\Lc_{\alphab}$ is positive semidefinite. The operator Chernoff bound requires estimate of $R$ and $\mu_{\min}$. 
$R$ can be estimated as follows. 
\begin{align*}
\bigg|\bigg|\frac{L}{m}\langle \cdot\,,\P_{\T}\,\w_{\alphab}\rangle\, \P_{\T}\,\w_{\alphab}\bigg|\bigg| = 
 \frac{L}{m}||\P_{\T}\,\w_{\alphab}||_{F}^{2} \le \frac{n\nu r}{m}
\end{align*}
The last inequality follows from the coherence estimate in \eqref{eq:coherencew}. 
Therefore, $R = \displaystyle\frac{n\nu r}{m}$. Next, we consider $\mu_{\min}$. The first step is to compute $\sum_{\alphab\in \Omega} E[\Lc_{\alphab}]$.
\[
\sum_{\alphab\in \Omega} E[\Lc_{\alphab}] = \sum_{\alphab \in \Omega} \bigg[ \sum_{\alphab } \frac{1}{m}\langle \cdot\,,\P_{\T}\,\w_{\alphab}\rangle\, \P_{\T}\,\w_{\alphab}\bigg]
= \sum_{\alphab } \langle \cdot\,,\P_{\T}\,\w_{\alphab}\rangle\, \P_{\T}\,\w_{\alphab}
\]
For any $\X \in \T$, lower bound $\displaystyle \langle \X\,, \sum_{\alphab\in \Omega} E[\Lc_{\alphab}](\X) \rangle$ as follows.
\[
\langle \X\,, \sum_{\alphab\in \Omega} E[\Lc_{\alphab}](\X) \rangle = \sum_{\alphab} \langle \X,\,\w_{\alphab}\rangle ^{2} \ge ||\X||_{F}^{2}\]
The  inequality results from Lemma \ref{H_prop} and Lemma \ref{norm_equivalent}. Noting that $\sum_{\alphab\in \Omega} E[\Lc_{\alphab}]$ is a self-adjoint operator, and using the variational characterization of the minimum eigenvalue, it can be concluded that
the minimum eigenvalue of $ \sum_{\alphab\in \Omega} E[\Lc_{\alphab}] $ is at least $1$. With this, set $\mu_{\min} = 1$. 
Finally, apply the matrix Chernoff bound
with $R = \frac{n\nu r}{m}$ and $\mu_{\min} = 1$. Setting $\delta = \frac{1}{2}$, 
$ \lambda_{\min}\left(\P_{\T}\,\F\,\P_{\T}\right)>\frac{1}{2}$
with probability of failure at most $p_1$ given by
\[
p_1 = n\exp\left(-\frac{\kappa}{8\nu}   \right)
\]
This concludes the proof. 
\end{proof}

%$||\P_{\T}\Po^{*}\P_{\T}\X-\P_{\T}\X||_{F}< \frac{1}{2}||\X||_{F}$
%for $\X\in \S$. This can be established from the following Lemma \ref{injective} which
%bounds $||\P_{\T}\Po^{*}\P_{\T}\X-\P_{\T}\X||_{F}$. 
%This lemma uses the vector Bernstein inequality in \cite{gross2011recovering}.
%%An easier but slightly weaker version is restated below for convenience of
%reference. 

%Under the assumptions of Theorem \ref{thm1} and using Lemma \ref{injective}, 
%$||\P_{\T}\Po^{*}\P_{\T}\X-\P_{\T}\X||_{F}< \frac{1}{2}||\X||_{F}$
%holds with probability $1-p_1$ where the probability of failure is 
%$\displaystyle p_{1}=\exp\left(-\frac{\kappa}{64\left(\nu+\frac{1}{nr}\right)}+\frac{1}{4}\right)$%with $\displaystyle \kappa = \frac{m}{nr}$. 

Under the assumptions of Theorem \ref{thm1} and using Lemma \ref{min_eig}, 
$\lambda_{\min}\left(\P_{\T}\,\F\,\P_{\T}\right)>\frac{1}{2}$ holds with
probability $1-p_1$ where the probability of failure is 
$\displaystyle p_1 = n\exp\left(-\frac{\kappa}{8\nu}\right)$ with 
 $\displaystyle \kappa = \frac{m}{nr}$. 

The proof of the second condition in Theorem \ref{uniqueness_conditions}
is more technically involved and requires the construction of $\Y$ satisfying the conditions in \eqref{eq:yconds}.
The construction follows the ingenious golfing scheme introduced in \cite{gross2011recovering}.
Partition the basis elements in $\Omega$ into $l$ batches
with the $i$-th batch $\Omega_{i}$ containing $m_i$ elements and $\displaystyle \sum_{i=1}^l m_i = m$. 
The restriction operator for a batch $\Omega_{i}$ is defined as
$\displaystyle\R_{i} = \frac{L}{m_i} \sum_{\alphab \in \Omega_{i}}{} \langle \X \,, \w_{\alphab} \rangle \v_{\alpha}$.
Then $\Y = \Y_{l}$ is constructed using the following recursive scheme 
\begin{equation} \label{eq:golfing_scheme}
\Q_{0} = \textrm{sgn } \M , \quad \Y_{i} = \sum_{j=1}^{i} \R^{*}_{j}  \Q_{j-1}, \quad \Q_{i} = \textrm{sgn } \M - \P_{\T}\Y_{i} , \qquad i = 1,\cdots,l
\end{equation}
Using the golfing scheme, it will be shown that the conditions in \eqref{eq:yconds} hold with very high probability.
The technical analysis of the scheme, particularly certifying the first condition in \eqref{eq:yconds}, requires a probabilistic estimate of $||\P_{\T}\Po^{*}\P_{\T}\X-\P_{\T}\X||_{F}\ge t$ for a fixed matrix $\X \in \S$.  
Using Lemma \ref{injective} to follow, a suitable estimate can be achieved. The lemma  
uses the vector Bernstein inequality in \cite{gross2011recovering}. 
An easier but slightly weaker version is restated below for convenience of reference. 

\begin{theorem}[Vector Bernstein inequality] 

Let $\x_1,...,\x_m$ be independent zero-mean vector valued random variables. 
Assuming that $\underset{i}{\max}\, ||\x_{i}||_{2}\le R$ and let $\sum_{i} E[||\x_i||_{2}^{2}]\le \sigma^{2}$, then
for any $t\le \frac{\sigma^2}{R}$, the following inequality holds
\[
\textrm{Pr}\bigg[\bigg|\bigg|\sum_{i=1}^{m}  \x_i\bigg|\bigg|_{2}\ge t\bigg] \le \exp\left(-\frac{t^2}{8\sigma^2}+\frac{1}{4}\right),
\]

\end{theorem}

\begin{lemma} \label{injective}
For a fixed matrix $\X \in\S$, the following estimate holds
\begin{equation}
\textrm{Pr}\,(||\P_{\T}\Po^{*}\P_{\T}\X-\P_{\T}\X||_{F}\ge t||\X||_{F}) \le
\exp\left(-\frac{t^2\kappa}{16\left(\nu+\frac{1}{nr}\right)}+\frac{1}{4}\right)
\end{equation}
for all $\displaystyle t \le1$ with $\displaystyle \kappa = \frac{m}{nr}$. 
\end{lemma}

\begin{proof}
It is clear that $\P_{\T}\X\in \S$ from the definition of $\P_{\T}$. Thus, we can expand $\P_{\T} \Po^{*}\P_{\T}\X$ as follows.
\begin{equation} \label{eq:ptrpt}
\P_{\T} \Po^{*}  \P_{\T}\X = \frac{L}{m} \sum_{\alphab\in\Omega} \langle \P_{\T}\X\,,\v_{\alphab} \rangle \,
\P_{\T}\,\w_{\alphab}
\end{equation}
With this, $\P_{\T} \Po^{*}\P_{\T}\X-\P_{\T}\X$
can be written as follows
\begin{equation}\label{eq:ptrpt_pt}
\P_{\T} \Po^{*}  \P_{\T}\X-\P_{\T}\X = \sum_{{\alphab\in\Omega}}\bigg[
\frac{L}{m} \langle \P_{\T}\X\,,\v_{\alphab} \rangle \,\P_{\T}\,\w_{\alphab}
-\frac{1}{m}\P_{\T}\X\bigg]
\end{equation}
Let $\X_{\alphab} = \frac{L}{m} \langle \P_{\T}\X\,,\v_{\alphab} \rangle \P_{\T}\,\w_{\alphab}$. Noting that $E[\frac{L}{m} \langle \P_{\T}\X\,,\v_{\alphab} \rangle \P_{\T}\,\w_{\alphab}]= \frac{1}{m}\P_{\T}\X$, 
let $\Y_{\alphab} = \X_{\alphab}-E[\X_{\alphab}]$ denote the summand. It is clear that $E[\Y_{\alphab}]=\bm{0}$ by construction.
The vector Bernstein inequality requires a suitable bound for 
$||\Y_{\alphab}||_{F}$ and $E[||\Y_{\alphab}||_{F}^{2}]$. Without loss of generality, assume $||\X||_{F}=1$.
The first step is to bound $||\Y_{\alphab}||_{F}$.
\begin{align}
||\Y_{\alphab}||_{F}= \bigg|\bigg|\frac{L}{m}\langle \P_{\T}\X\,,\v_{\alphab} \rangle\, \P_{\T}\,\w_{\alphab}
-\frac{1}{m}\P_{\T}\X\bigg|\bigg|_{F}  &\le
\bigg|\bigg|\frac{L}{m}\langle \P_{\T}\X\,,\v_{\alphab} \rangle \P_{\T}\,\w_{\alphab}\bigg|\bigg|_{F}
+\bigg|\bigg|\frac{1}{m}\P_{\T}\X\bigg|\bigg|_{F} \label{eq:norm_Y_squared}\\
& < \frac{n^2}{2m}\max_{\alphab\in \Us} ||\P_{\T}\,\w_{\alphab}\,||_{F}\max_{\alphab\in \Us} ||\P_{\T}\,\v_{\alphab}||_{F}
+\frac{1}{m}\nonumber\\
& \le \frac{1}{m}(2n\nu r+1)
\end{align}
Above, the last inequality follows from the coherence estimates \eqref{eq:coherencew} and \eqref{eq:coherencev}. 
Using the above estimate, set $R$ in the vector Bernstein inequality as $R = \frac{1}{m}(2n\nu r+1)$.
Using the parallelogram identity and monotonicity of expectation, $E[||\Y_{\alphab}||_{F}^{2}] \le 
2E[||\X_{\alphab}||_{F}^{2}]+2||E[\X_{\alphab}]||_{F}^{2}$.
With this, $E[||\Y_{\alphab}||_{F}^{2}]$ can be upper bounded
as follows
\begin{align*}
E[||\Y_{\alphab}||_{F}^{2}] &\le \frac{2L^2}{m^2} E[\langle \P_{\T}\X\,,\v_{\alphab} \rangle^{2}||\P_{\T}\,\w_{\alphab}||_{F}^{2}]
+\frac{2}{m^2}||\P_{\T}\X||_{F}^{2}\\
& \le \frac{2L}{m^2} \max_{\alphab\in \Us}||\P_{\T}\w_{\alphab}||_{F}^{2} \sum_{\alphab\in \Us}\langle \P_{\T}\X\,,\v_{\alphab} \rangle^{2}
+\frac{2}{m^2}\\
& \le \frac{n^2}{m^2}\frac{2\nu r}{n}+\frac{2}{m^2}
\le \frac{2}{m^2}\left(1+n\nu r\right)
\end{align*} 
Above, the third inequality follows from the coherence estimates, \eqref{eq:coherencew} and \eqref{eq:coherencev},
and Lemma \ref{norm_equivalent}. Using the above estimate, set $\sigma^{2}$ in vector Bernstein inequality as $\sigma^{2} = \frac{2}{m}\left(1+n\nu r\right)$.
Finally, apply the vector Bernstein inequality
with $R= \frac{1}{m}(2n\nu r+1)$ and $\sigma^{2}= \frac{2}{m}\left(1+n\nu r\right)$.
For $t\le \frac{\sigma^{2}}{R} > 1 $, 
\begin{equation}
\textrm{Pr}\,(||\P_{\T}\Po^{*}\P_{\T}\X-\P_{\T}\X||_{F}\ge t) \le 
\exp\left(-\frac{t^2\kappa}{16\left(\nu+\frac{1}{nr}\right)}+\frac{1}{4}\right)
\end{equation}
where $\displaystyle\kappa= \frac{m}{nr}$. This concludes the proof of Lemma \ref{injective}.
\end{proof}

Having proved Lemma \ref{injective}, the next step is to show that the golfing scheme \eqref{eq:golfing_scheme} certifies the conditions
in \eqref{eq:yconds} with very high probability. The precise statement is stated in Lemma \ref{construction_of_y} below. 
 
\begin{lemma} \label{construction_of_y}
 $\Y_{l}$ obtained from the golfing scheme \eqref{eq:golfing_scheme} satisfies the conditions
 in \eqref{eq:yconds} with failure probability at most 
$\displaystyle p = \sum_{i=1}^{l} (p_2(i) +  p_3(i)+ p_4(i))$
where  
$\displaystyle p_2(i) = \exp\left(\frac{-\kappa_i}{64\left(\nu+\frac{1}{nr}\right)}+\frac{1}{4}\right)$,
$\displaystyle p_3(i) = n \exp\left(\frac{-3k_i}{128\nu}\right) $ and 
$\displaystyle p_{4}(i) = n^{2} \exp\left(\frac{-3\kappa_i}{64\left(\nu+\frac{1}{nr}\right)}\right)$ with $\displaystyle k_{i} = \frac{m_i}{nr}$.

\end{lemma}

\begin{proof}
Note that $\textrm{sgn } \M$ is symmetric. It is easy to verify that $\Q_i$ is symmetric as 
$\Y_{i}$ and $\P_{\T}\Y_{i}$ are symmetric. In addition, using Lemma \ref{signx_prop}, $\Q_{i}$ is in $\T$ since $\textrm{sgn } \M \in \T$. To show that the first condition in \eqref{eq:yconds} holds,
we derive a recursive formula for $\Q_i$ as follows. 
\begin{align}
  \Q_{i} & =  \textrm{sgn } \M - \P_{\T} \left(\sum_{j=1}^{i} \R^{*}_{j} \Q_{j-1}\right)= 
	\textrm{sgn } \M - \P_{\T} \left(\sum_{j=1}^{i-1} \R^{*}_{j} \Q_{j-1} + \R^{*}_{i}\Q_{i-1}\right)\nonumber\\
  & = \textrm{sgn } \M - \P_{\T}  \sum_{j=1}^{i-1} \R^{*}_{j} \Q_{j-1} -\P_{\T}  \R^{*}_{i}\Q_{i-1}
  =  \textrm{sgn } \M - \P_{\T}  \Y_{i-1} -\P_{\T}  \R^{*}_{i}\Q_{i-1}\nonumber\\
   & = \Q_{i-1} -\P_{\T}  \R^{*}_{i}\Q_{i-1}= (\P_{\T} - \P_{\T} \R^{*}_{i} \P_{\T} ) \Q_{i-1} \label{eq:qi_update}
\end{align}
Observe that the first condition in \eqref{eq:yconds} is equivalent to a bound on $\|\Q_{l}\|_{F}$.
Using the bound $\|(\P_{\T} - \P_{\T} \R^{*}_{i} \P_{\T})\Q_{i-1}\|_{F}< t_{2,i}\|\Q_{i-1}\|_{F}$
with failure probability $p_2(i)$ obtained from Lemma \ref{injective} and setting $t_{2,i} = 1/2$,
$\displaystyle p_{2}(i) = \exp\left(\frac{-\kappa_i}{64\left(\nu+\frac{1}{nr}\right)}+\frac{1}{4}\right)$
with $\displaystyle \kappa_i = \frac{m_i}{nr}$, we have the following bound of $\|\Q_{i}\|_{F}$ using the above recursive formula.
\begin{equation}\label{eq:qi_norm}
  \|\Q_{i}\|_{F} <  \left(\prod_{k=1}^{i} t_{2,k} \right) \|\Q_0\|_F = 2^{-i} \sqrt{r}
\end{equation}
To satisfy the first condition in \eqref{eq:yconds}, set $l = \log_{2}\left(4n\sqrt{8Lr}\right)$.
Using the union bound on the failure probabilities $p_2(i)$, we have $\displaystyle ||\P_{\T}\Y-\textrm{Sgn}\,\M||_{F} = ||\Q_l||_{F} \le \sqrt{r}2^{-l} = \frac{1}{4n\sqrt{8L}}$ holds
true with failure probability at most $\sum_{i=1}^{l} p_2(i)$. 

To complete the proof, it remains to show that $\Y$ satisfies the second condition in \eqref{eq:yconds}.
The condition is equivalent to controlling the operator norm of
$\P_{\T^{\perp}} \Y_l$. First, it is clear that $ \|\P_{\T^{\perp}} \Y_{l} \|  
=  \|  \sum_{j=1}^{l} 
\P_{\T^{\perp}} \R^{*}_j \Q_{j-1}\| \le   \sum_{j=1}^{l} \|\P_{\T^{\perp}} \R^{*}_j \Q_{j-1}\| $.
This motivates bounding the operator norm of $\|\P_{\T^{\perp}} \R^{*}_j \Q_{j-1}\|$
which will be the focus of Lemma \ref{golfing_size} below.
Before proving the lemma, we start by addressing an assumption the lemma entails 
on the size of  $\eta(\Q_i) = \max_{\betab} \,|\langle \Q_{i}\,,\v_{\betab} \rangle|$.
Specifically, at the $i$-th stage of the golfing scheme, 
we need to show that the assumption $\displaystyle \eta(\Q_{i})^{2}\le \frac{\nu}{n^2} c_{i}^2$
with $||\Q_{i}||_{F}^{2}\le c_i^2$ holds with very high probability. 
To enforce this, let $\eta(\Q_i)< t_{4,i}$ with probability $1-p_4(i)$. We set
$t_{4,i} = \frac{1}{2}\eta(\Q_{i-1})$ to obtain
\[
\eta(\Q_{i})^{2} < 2^{-2} \eta(\Q_{i-1})^{2} < 2^{-2i} \eta(\textrm{sgn}\,\M)^2 \le  \frac{\nu r}{n^{2}}(2^{-2i})  = \frac{\nu}{n^2}(2^{-2i}r) 
\]
Above, the second inequality is obtained by applying the inequality $\eta(\Q_i)< \frac{1}{2}\eta(\Q_{i-1})$ recursively 
and the last inequality follows from the coherence condition in \eqref{eq:jointcoherence}.
Using \eqref{eq:qi_norm},  at the $i$-th stage of the golfing scheme,  
the above inequality  precisely enforces the condition that
$\displaystyle\eta(\Q_{i})^{2}\le \frac{\nu}{n^2}c_i^2$ with $c_i = 2^{-i}\sqrt{r}$. Using Lemma \ref{joint_coherence_size_golfing},
noting that $\eta(\Q_i) = \eta(\Q_{i-1}-\P_{T}\R^{*}_{i}\Q_{i-1}),$ 
the failure probability $p_{4}(i)$ is given by
\[
p_{4}(i) =  n^{2} \exp\left(\frac{-3\kappa_i}{64\left(\nu+\frac{1}{nr}\right)}\right) \quad \quad \forall i\in[1,l]
\]
This concludes the analysis on the assumption of the size of $\eta(\Q_i)$
which will be used freely in Lemma \ref{golfing_size}. 
Using the bound $||\P_{\T^{\perp}}\R^{*}_j\Q_{j-1}|| \le t_{3,j}\,c_{j-1}$, where once again $c_{j-1}$ satisfies $||\Q_{j-1}||_{F}\le c_{j-1}=2^{-(j-1)}$, with failure probability $p_3(j)$ obtained from Lemma \ref{golfing_size} and setting
$\displaystyle t_{3,j}=\frac{1}{4\sqrt{r}}$, 
$\displaystyle p_{3}(j) = n \exp\left(\frac{-3\kappa_j}{128\nu}\right)$
with $\displaystyle \kappa = \frac{m_j}{nr}$, it follows that
\begin{align*}
  ||\P_{\T^{\perp}} \Y_{l}|| \le  \sum_{k=1}^{l} ||\P_{\T^{\perp}}\R^{*}_{k}\Q_{k-1}||
  \le \sum_{k=1}^{l} t_{3,\,k}\,c_{k-1} = \frac{1}{4\sqrt{r}} \sum_{k=1}^{l} c_{k-1}
  = \frac{1}{4\sqrt{r}} \sum_{k=1}^{l} \sqrt{r}2^{-(k-1)} < \frac{1}{2}
  \end{align*}
where the second to the last inequality uses \eqref{eq:qi_norm} to bound $||\Q_{k-1}||_{F}$.
Using the union bound of failure probabilities, $||\P_{\T^{\perp}} \Y_{l}||<\frac{1}{2}$ holds true with failure probability 
which is at most $\sum_{j=1}^{l}[p_3(j)+p_4(j)]$. 
The second condition in \eqref{eq:yconds} now holds with
this failure probability. This completes the proof of Lemma \ref{construction_of_y}.
\end{proof}

The proof of Lemma \ref{golfing_size} uses the Bernstein inequality.  A simplified version of the
inequality was derived in \cite{tropp2012user}. Our analysis uses
this simpler version and is restated below for ease of reference.

\begin{theorem}[Bernstein inequality]
 Consider a finite sequence $\{\X_i\}$ of independent, random, self-adjoint matrices with dimension $n$.
  Assuming that 
  \[
E[\X_i] =0 \quad \textrm{and} \quad \lambda_{\max}(\X_{i})\le R \quad \textrm{ almost surely.}
\]
with the norm of the total variance
$\sigma^{2} = \bigg|\bigg|\sum_{i} E(\X_{i}^{2})\,\bigg|\bigg|$, then, 
for all $t\ge 0$, the following estimate holds:
\begin{equation} \label{eq:bern}
  Pr\bigg[\big|\big|\sum_{i}\X_{i}\big|\big|>t\bigg] \le
\begin{cases}
\displaystyle n \exp\left(-\frac{3t^2}{8\sigma^{2}}\right) & t \le \frac{\sigma^{2}}{R} \vspace{0.2cm}\\
 \displaystyle n \exp\left(-\frac{3t}{8R}\right)            & t\ge \frac{\sigma^{2}}{R} 
\end{cases}
\end{equation}
\end{theorem}

Now we are ready to estimate $||\P_{\T^{\perp}} \R^{*}_j \Q_{j-1}||$ formally described in the following lemma.
\begin{lemma}  \label{golfing_size}
Consider a fixed matrix $\G \in \T$. Assume that $\displaystyle \max_{\betab}\, \langle \G\,,\v_{\betab} \rangle^{2}\le \frac{\nu}{n^2}c^{2}$
with $c$ satisfying $||\G||_{F}^{2}\le c^{2}$. 
Then, the following estimate holds for all $\displaystyle t\le  \frac{1}{\sqrt{r}}$ with $\displaystyle \kappa_j = \frac{m_j}{nr}$.
  \[
\textrm{Pr}(||\P_{\T^{\perp}}\R^{*}_{j} \G|| > t\,c) \le n \exp\left(-\frac{3t^2\kappa_jr}{8\nu}\right)
\]

\end{lemma}

\begin{proof}[Proof of Lemma \ref{golfing_size}]
We expand $\P_{\T^{\perp}}\R^{*}_{j} \G$
in the dual basis as 
$\displaystyle \P_{\T^{\perp}}\R^{*}_{j} \G = \sum_{\alphab\in\Omega_j} \frac{L}{m_j}  \langle \G\,,\v_{\alphab} \rangle \P_{\T^{\perp}} \w_{\alphab}$.
Let $\X_{\alphab}$ denote the summand.  The proof makes use of Bernstein inequality \eqref{eq:bern}
which mandates appropriate bound on $\|\X_{\alphab}\|$ and  $\|E[\X_{\alphab}^2]\|$. 
First consider the bound on the latter term.  Noting that 
$\displaystyle \X_{\alphab}^2 = \frac{L^2}{m_{j}^{2}}\langle \G\,,\v_{\alphab} \rangle ^{2} (\P_{\T^{\perp}} \w_{\alphab})^2$, 
we have the following bound for $E[\X_{\alphab}]^2$ using Lemma \ref{operator_norm_of_sum} and the fact 
that $\w_{\alphab}^{2}$ is positive semidefinite. 
\begin{align}
  \big|\big|E[\X_{\alphab}^2]\big|\big| &\le  \frac{n^2}{2m_j^{2}}\, \underset{\|\varphi\|_2=1}{\max}\sum_{\alphab\in \Us} 
	\langle \G\,,\v_{\alphab} \rangle ^{2} 
  \langle \varphi\,, \w_{\alphab}^{2}\varphi \rangle 
\le  \frac{n^2}{2m_{j}^2} \, \underset{\alphab\in \Us}{\max}\,\, \langle \v_{\alphab}\,,\G\rangle ^{2}\,
\underset{\|\varphi\|_2=1}{\max} \,  \langle \varphi\,, \left(\sum_{\alphab\in \Us}\w_{\alphab}^{2}\right)\varphi \rangle \label{eq:golfing_var_inequality}
\end{align}
Due to the special structure of $\w_{\alphab}$ in the EDG problem, it is straightforward to verify that
\[
\sum_{\alphab\in \Us} \w_{\alphab}^{2} = 2n\I -2\bm{1}\bm{1}^{T}
\]
It now follows that $\lambda_{\max}(\sum_{\alphab\in \Us} \w_{\alphab}^{2}) = 2n$
which implies that 
$$\big|\big|E[\X_{\alphab}^2]\big|\big| \le \frac{n^3}{m_j^2}\,\underset{\alphab\in \Us}{\max}\,\, \langle \v_{\alphab}\,,\G\rangle ^{2}
\le \frac{n^3}{m_{j}^2} \frac{\nu}{n^2}c^2  = \frac{n\nu}{m_{j}^{2}}c^2 $$
Next, consider a bound on $\|\X_{\alphab}\|$ which results
\begin{equation}
  \|\X_{\alphab}\| \le \frac{n^2}{2m_{j}} |\underset{\alphab\in \Us}{\max}\,\, \langle \v_{\alphab}\,,\G\rangle|\, \|\P_{\T^{\perp}}\w_{\alphab}\|
  \le \frac{n^2}{2m_{j}} \frac{\sqrt{\nu}}{n} \|\w_{\alphab}\|\, c
\end{equation}
We consider two cases. If $\displaystyle \nu\ge \frac{1}{r}$, $\displaystyle\|\X_{\alphab}\|  \le \frac{n\nu\sqrt{r}}{m_{j}}c=R_{1}$
and if $\displaystyle\nu < \frac{1}{r}$, $\displaystyle\|\X_{\alphab}\|  <\frac{n\sqrt{\nu}}{m_{j}}c=R_{2}$. 
Finally, we apply the Bernstein inequality with $R_1,R_2$ and
$\displaystyle \sigma^{2} = \frac{n\nu}{m_{j}}c^2$.
It can be easily verified that $\displaystyle\frac{\sigma^{2}}{R_1}= \frac{1}{\sqrt{r}}c$
and $\displaystyle\frac{\sigma^{2}}{R_2}< \frac{1}{\sqrt{r}}c$.
Therefore, for all $\displaystyle t \le \frac{1}{\sqrt{r}}c$, the following estimate holds. 
\begin{equation}
\textrm{Pr}(\|\P_{\T^{\perp}}\R^{*}_{j} \G\| \ge t) \le n \exp\left(-\frac{3t^2\kappa_{j}r}{8\nu c^2}\right)
\end{equation}
with $\displaystyle \kappa_{j}= \frac{m_{j}}{nr}$. 
This concludes the proof of Lemma \ref{golfing_size}.
\end{proof}

In what follows, it will be argued that $\M$ is a unique solution to \eqref{eq:m_completion} with ``very high'' probability. It is suffices to show all  $\X\in \Sp -\{\M\}$ are solutions with a very small probability by choosing $\Omega$ ``sufficiently large". 
For $\X \in \Sp$, write the deviation 
$\bDelta = \X-\M$. We define $\S_{1}=\big\{\X\in \Sp: ||\bDelta_{\T}||_{F}^{2} \ge \left(n\sqrt{8\frac{L}{m}} ||\bDelta_{\T^{\perp}}||_{F}\right)^{2} \big\}$ and 
$\S_2=\big\{\X\in \Sp:||\bDelta_{\T}||_{F}^{2} < \left(n\sqrt{8\frac{L}{m}}||\bDelta_{\T^{\perp}}||_{F}\right)^{2}\,\&\,\Po(\bDelta)=\bm{0} \}$ be two spaces
defined based on deterministic comparisons of $||\bDelta_{\T}||_{F}^{2}$ and $||\bDelta_{\T^{\perp}}||_{F}^{2}$. Assuming that $\Omega$ is sampled uniformly at random with replacement, we consider the following two cases.
\begin{enumerate}
\item Choose $|\Omega| = m$ ``sufficiently large" such that $\textrm{Pr}\, \left(\left\{\Omega\subset \Us\,\big|\,\, |\Omega| = m, \lambda_{\min}\left(\P_{\T}\,\F\,\P_{\T}\right)> 1/2\right\}\right)  \geq 1-p_1$ based on Lemma \ref{min_eig}. Using Theorem \ref{uniqueness_conditions}(a), all $\X\in\S_{1}$ are feasible solutions to \eqref{eq:m_completion} with probability $p_1$. 

\item Recall that using the golfing scheme $\textrm{Pr}\,\left(\left\{\Omega\subset \Us\, \big|\,\, |\Omega| = m, \Y \in \textrm{range}\,\,\Po^{*}\, \& \, ||\P_{\T}\Y-\textrm{Sgn}\,\M||_{F} \le \frac{1}{4L}\, \&\, ||\P_{\T^{\perp}}\Y||\le \frac{1}{2}\right\}\right) \geq1-\epsilon$ with $\displaystyle \epsilon=  \sum_{i=1}^{l} \left[p_{2}(i)+p_3(i)+p_4(i)\right]$ by choosing $|\Omega| = m$  ``sufficiently large". Using Theorem \ref{uniqueness_conditions}(b),  all $\X\in\S_{2}$ are not solutions to \eqref{eq:m_completion} with probability $\epsilon$.  
\end{enumerate}

Using the union bound, $\X\neq \M$ is a solution to the EDG nuclear minimization problem with probability 
which is at most $p= p_1+\sum_{i=1}^{l}[p_{2}(i)+p_3(i)+p_4(i)]$.
In what follows, the goal is to ensure that $p$ is ``small" with $m$ sufficiently large implying the uniqueness of $\M$  to \eqref{eq:m_completion} with very 
high probability. This is attained by a suitable choice of the parameters $m$, $l$ and $m_i$ which are detailed
below.

The first failure probability is the failure that condition in Theorem \ref{uniqueness_conditions}(a) does not hold 
and is given by
\[
p_{1} = n\exp\left(\frac{-\kappa}{8\nu}\right)
\]
In the proof of Lemma \ref{construction_of_y}, the failure probabilities $p_{2}(i)$, $p_{3}(i)$ and $p_{4}(i)$ are given by
\[
p_{2}(i) = \exp\left(\frac{-\kappa_i}{64\left(\nu+\frac{1}{nr}\right)}+\frac{1}{4}\right)
\quad; \quad p_{3}(i) = n \exp\left(\frac{-3\kappa_i}{128\nu}\right)  \quad ; \quad
p_{4}(i) =  n^{2} \exp\left(\frac{-3\kappa_i}{64\left(\nu+\frac{1}{nr}\right)}\right)  \quad \quad \forall i\in[1,l]
\]
To prove Theorem \ref{thm1}, it remains to specify $k_{i}$ and show that the total
probability of failure is very ``small". In precise terms, this means that the probability of failure 
is bounded above by $n^{-\beta}$ for some $\beta>1$.
$k_i$ is chosen in such a way that all the failure probabilities, $p_1,p_2(i),p_3(i)$ and $p_4(i)$, are
at most $\displaystyle \frac{1}{4l}n^{-\beta}$. With this, one appropriate choice for $k_i$ is
$k_i = 96\big(\nu+\frac{1}{nr}\big)\big(\beta\log(n)+\log(4l)\big)$. 
Using the union bound, it can be verified that the total failure is bounded above
by $n^{-\beta}$. The number of measurement $m=lnrk_i$ must be at least
\begin{equation} \label{eq:m_EDG_basis}
\log_{2}\left(4n\sqrt{8Lr}\right)nr
\left(96\big(\nu+\frac{1}{nr}\big)\big[\,\beta\log(n)+\log\left(4\log_{2}(4L\sqrt{r})\right)\big]\right)
\end{equation}
This finishes the proof of Theorem \ref{thm1} and concludes that the minimization
program in $\eqref{eq:m_completion}$ recovers the true inner product matrix with very high probability.

\subsection{Noisy EDG Completion}

In a practical settings, the available partial information is not exact but noisy. For simplicity, consider
an additive Gaussian noise $\Z$ with mean $\mu$ and variance $\sigma$. The modified nuclear norm minimization
for the EDG problem can now be written as
\begin{align}
&\underset{\X\in \Sp}{\textrm{minimize }} \quad \|\X\|_{*} \nonumber\\
&\textrm{subject to }\quad ||\Po(\X) -\Po(\M)||_{F}\le \delta  \label{eq:m_completion_noisy}
\end{align}
where $\delta$ characterizes the level of noise. 
Unlike the noisy matrix completion problem, the noise parameters for the EDG
problem can not be chosen arbitrarily. The reason is that the perturbed
distance matrix needs to be non-negative. In the context of numerically solving the noisy EDG problem,
details of how to set these parameters will be discussed
in the next section.
Under the assumption that $\mu$ and $\sigma$ are chosen
appropriately, we posit that the theoretical guarantees for the exact case extend 
to noisy EDG completion. Following the analysis in \cite{candes2010matrix} and using the dual
basis framework, it can be surmised that Theorem \ref{thm1} holds true with failure probability proportional to
noise level $\delta$. A sketch of such a theorem is stated below. 
\begin{theorem}
Let $\M\ \in \real^{n\times n}$ matrix of rank $r$ that obeys the coherence conditions
\eqref{eq:coherencew1}, \eqref{eq:coherencev1} and \eqref{eq:jointcoherence1}
with coherence $\nu$.
Assume $m$ measurements $\langle \M\,,\w_{\alphab}\rangle$, sampled uniformly at random with replacement
are corrupted with Gaussian noise of mean $\mu$, variance $\sigma$ and noise level $\delta$.  
For any $\beta>1$, if
  \[
  m\ge \log_{2}\left(4\sqrt{8Lr}n\right)nr\left(96[\nu+\frac{1}{nr}][\beta\log(n)+\log\left(4\log_{2}(4L\sqrt{r})\right)]\right)
 \]
then
\[
||\M-\bar{\M}||_{F} \le f\left(n,\frac{m}{n^2},\delta\right)
\]
where $\bar{\M}$ is a solution to \eqref{eq:m_completion_noisy} with probability at least $1-n^{-\beta}$.
\end{theorem}

\noindent
The interpretation of the above theorem is that the EDG problem is stable to noise. 
Numerical experiments confirm this consideration. However, a precise analysis mandates
characterization of the level of noise and an exact specification of $f$ in the above theorem.
This is left for future work. 

\section{Numerical Algorithms}
\label{sec:Algorithms}
The theoretical analysis so far shows that the nuclear minimization approach yields
a unique solution for the EDG problem. 
In this section, we aim at
developing a practical algorithm to recover the inner product matrix
$\X$ given partial pairwise distance information supported on some
random set $\Omega$.  Since the available partial information 
might be noisy in applications, we also extend the algorithm to this case. 

\subsection{Exact partial information}

An algorithm similar to ours appears in \cite{lai2017solve}
where the authors design an algorithm employing the alternative directional minimization (ADM)
method to recover the Gram matrix. A crucial part of their algorithm uses the hard
thresholding operator which computes eigendecompositions at every iteration
and is computationally intensive. Comparatively, an advantage of our algorithm is that it does not
require an eigendecomposition making it fast and suitable for tests on large data.
Since the nuclear norm of a positive semidefinite matrix
equates to its trace, we consider the following minimization
problem.
\begin{equation*}
\underset{\bar{\X}\in \real^{n\times n}}{\textrm{min}} \quad\textrm{Trace }(\bar{\X})  \quad \textrm{ subject to } \quad \Po(\bar{\X})=
\Po(\M) \,\,,\,\,\bar{\X} \cdot \bm{1}=\bm{0} \,\,\textrm{and}\,\, \bar{\X}\succeq \bm{0}
\end{equation*}
Consider a matrix $\bm{C} \in \real^{n\times (n-1)}$ satisfying $\bm{C}^T\bm{C} = \I$ and $\bm{C}^{T}\bm{1}=\bm{0}$, we rewrite the above minimization problem as follows by a change of variable $\bar{\X} = \bm{C}\X\bm{C}^{T}$.
\begin{equation*}
\underset{\X\in \real^{(n-1)\times (n-1)}}{\textrm{min}} \quad\textrm{Trace }(\C\X\C^{T})  \quad \textrm{ subject to } \quad \Po(\C\X\C^T)=
\Po(\M) \,\,,\,\, \X \succeq \bm{0}
\end{equation*}
Note that the sum to one constraint drops out since $\C^{T}\bm{1}=0$
and $\C\X\C^T\succeq \bm{0}$ if $\X\succeq \bm{0}$. To enforce that $\X$ is positive semidefinite, 
let $\X =\bar{\bm{P}}\bar{\bm{P}}^{T}$ where $\bar{\bm{P}}\in \real^{(n-1)\times q}$ with $q$ unknown a priori.  
Since $\bar{\bm{P}}\bar{\bm{P}}^{T}$ has at most rank $q$, it entails a good estimate for
$q$ which ideally should be a reasonable estimate of the rank. 
Due to the trace regularization, our algorithm only needs a rough guess of $q$. 
The above minimization problem now reduces to
\begin{equation}\label{eq:trace_min_problem_orig}
\underset{\bar{\bm{P}}\in \real^{(n-1)\times q}}{\textrm{min}} \quad \textrm{Trace }(\C\bar{\bm{P}}\bar{\bm{P}}^{T}\C^T)  
\quad \textrm{ subject to } \quad \Po(\C\bar{\bm{P}}\bar{\bm{P}}^{T}\C^T)=\Po(\M) 
\end{equation}
With $\bm{P}= \C\bar{\bm{P}}$, consider the simplified minimization problem.
\begin{equation}\label{eq:trace_min_problem}
\underset{\bm{P}\in \real^{n\times q}}{\textrm{min}} \quad\textrm{Trace }(\bm{P}\bm{P}^{T})  \quad \textrm{ subject to } 
\quad \Po(\bm{P}\bm{P}^{T})=\Po(\M) 
\end{equation}
The above technique, the change of variables employing $\C$, has been previously used \cite{alfakih1999solving,
dokmanic2015euclidean}. In \cite{dokmanic2015euclidean}, the authors remark that the reparamterization 
leads to numerically stable interior point algorithms. Note that $\C$ is simply an orthonormal
basis for the space $\{\x\in \real^{n}~|~\bm{x}^{t}\bm{1}=\bm{0} \}$. For the EDG problem, the goal is to find the Gram matrix
$\bar{\X} = \C\X\C^T = \C\C^T\bm{P}\bm{P}^T\C\C^T$. $\C\C^T$ is simply the orthogonal
projection onto the aforementioned space given by $\C\C^{T} =  \I-\frac{1}{n}\bm{1}\bm{1}^{T}$.
Given $\bar{\X}$, classical MDS can then be used to find the point coordinates. Therefore,
our focus is on solving for $\bm{P}$ in \eqref{eq:trace_min_problem}. 
We employ the method of augmented Lagrangian to solve the minimization problem.
The constraint $\R_{\Omega}(\bm{P}\bm{P}^{T})=\R_{\Omega}(\M) $ 
can be written compactly using the linear operator $\A$ defined as 
$\A :\, \real^{n\times n}\longrightarrow \real^{|\Omega|}$
with $\A(\X) = \bm{f}\in \real^{{|\Omega|}\times 1}$, $\bm{f}_{i}=
\langle \X\,,\w_{\alphab_i} \rangle$ for $\alphab_i \in \Omega$.
For latter use, the adjoint of $\A$ , $\A^{*}$ , can be derived
as follows. For $\bm{y}\in \real^{|\Omega|\times 1} $, $\langle \A \X\,,\bm{y} \rangle = 
\sum_{i} \langle \X\,, \w_{\alphab_i} \rangle \bm{y}_{i}
=   \langle \X\,, \sum_{i} \bm{y}_{i} \w_{\alphab_i}\rangle$.
It follows that $\A^{*}\bm{y} =  \sum_{i} \bm{y}_{i} \w_{\alphab_i} $. Thus, we write 
\eqref{eq:trace_min_problem} as
\begin{equation} \label{eq:alg2_minimization}
\underset{\bm{P}}{\textrm{min}} \quad\textrm{Trace }(\bm{P}\bm{P}^{T}) \quad \textrm{ subject to } 
\quad \A(\bm{P}\bm{P}^{T}) =\bm{b}
\end{equation}
The augmented Lagrangian is given by
\begin{equation*}
\L(\bm{P} ; \Lam)= \textrm{Trace }(\bm{P}\bm{P}^{T}) + 
\frac{r}{2} ||\A(\bm{P}\bm{P}^{T})-\bm{b}+\Lam||_{2}^{2}
\end{equation*}
where $\Lam\in \real^{|\Omega|\times1}$ denotes the Lagrangian multiplier and 
$r$ is the penalty term.
The augmented Lagrangian step is simply
$ \bm{P}^{k} =  \textrm{argmin } \L(\bm{P};\Lam^{k-1})$
followed by updating of the multiplier $\Lam^{k}$.
To solve the first problem with respect to $\bm{P}$, the Barzilai-Borwein
steepest descent method \cite{barzilai1988two} is employed with objective function 
$\textrm{Trace }(\bm{P}\bm{P}^{T}) + \frac{r}{2} ||\A(\bm{P}\bm{P}^{T})-\bm{b}+\Lam^{k-1}||_{2}^{2}$
and gradient $2\bm{P}+2r\A^{*}\left(\A(\bm{P}\bm{P}^{T})-\bm{b}+\Lam^{k-1}\right)\bm{P}$.
The iterative scheme is summarized in Algorithm \ref{alg:trace_min_scheme}. 
\begin{algorithm}
\caption{Augmented Lagrangian based scheme to solve \eqref{eq:alg2_minimization}}
\label{alg:trace_min_scheme}
\begin{algorithmic}[1]
\State \textbf{Initialization:} Set $\Lam^{0}=\bm{0}$\,,
$q = 10$,\,\,\,$\Pb^{0} = \textrm{rand(n, q)}$,\,\,\,$E^{0}=E_{\textrm{Total}}^{0}=0$.
Set $\textrm{maxiterations} \,, \textrm{bbiterations} \,, r \,,\textrm{Tol}$
\For {k = 1:maxiterations}
\State Barzilai-Borwein(BB) descent for $\Pb^{k} =  \textrm{argmin } \L(\bm{P};\Lam^{k-1})$. 
\State $\Lam^{k} = \Lam^{k-1}+ \A(\bm{P}^{k}(\bm{P}^{k})^{T})-\bm{b}$
\State $E^{k}$ = $\frac{r}{2}||\A(\bm{P}^{k}(\bm{P}^{k})^{T})-\bm{b}||_{2}^{2}$
\State $E_{\textrm{Total}}^{k}= \textrm{Trace }(\bm{P}^{k}(\bm{P}^{k})^{T}) + 
\frac{r}{2} ||\A(\bm{P}^{k}(\bm{P}^{k})^{T})-\bm{b}+\Lam^{k}||_{2}^{2}$
\If {$\textrm{E}^{k}  < \textrm{Tol} \,\,\&\,\, E_{\textrm{Total}}^{k} < \textrm{Tol}  $}
\State \textbf{break}
\EndIf
\EndFor
\end{algorithmic}
\end{algorithm}

\subsection{Partial information with Gaussian noise}

Assume that the available partial information is noisy. Formally, $\Po(\X) = \Po(\M)+\Po(\Z)$
where $\Po(\Z)$ is an additive Gaussian noise. Proceeding analogously to the case of
exact partial information, the following minimization problem is obtained.
\begin{equation}\label{eq:trace_min_problem_noisy}
\underset{\bm{P}}{\textrm{min}} \quad \textrm{Trace }(\bm{P}\bm{P}^{T}) +
\frac{\lambda}{2}||\Po(\bm{P}\bm{P}^{T})-\Po(\M)||_{F}^{2}
\end{equation}
Using the operator $\A$ introduced earlier, \eqref{eq:trace_min_problem_noisy} can be rewritten as
\begin{equation} \label{eq:alg2_minimization_noisy}
\underset{\bm{P}}{\textrm{min}} \quad\textrm{Trace }(\bm{P}\bm{P}^{T}) + 
\frac{\lambda}{2}||\A(\bm{P}\bm{P}^{T}) - \bm{b}||_{F}^{2}
\end{equation}
where $\bm{b} = \A(\M)$. The augmented Lagrangian is given by
\begin{equation*}
\L(\bm{P}; \Lam)= \textrm{Trace }(\bm{P}\bm{P}^{T}) + 
\frac{\lambda}{2} ||\A(\bm{P}\bm{P}^{T})-\bm{b}||_{2}^{2} 
\end{equation*}
where $\Lam\in \real^{n\times1}$ denotes the Lagrangian multiplier and $r$ is a penalty term.
The augmented Lagrangian step is simply
$ \bm{P}^{k} =  \textrm{argmin } \L(\bm{P};\Lam^{k-1})$.
As before, $\bm{P}$ is computed using the Barzilai-Borwein
steepest descent method with objective function $\textrm{Trace }(\bm{P}\bm{P}^{T}) + 
\frac{\lambda}{2} ||\A(\bm{P}\bm{P}^{T})-\bm{b}||_{2}^{2} $ and gradient 
$2\bm{P}+2\lambda\A^{*}\left(\A(\bm{P}\bm{P}^{T})-\bm{b}\right)\bm{P}$.
The iterative scheme is summarized in Algorithm  \ref{alg:trace_min_scheme_noisy}.
\begin{algorithm}
\caption{Augmented Lagrangian based scheme to solve \eqref{eq:alg2_minimization_noisy}}
\label{alg:trace_min_scheme_noisy}
\begin{algorithmic}[1]
\State \textbf{Initialization:} Set $q = 10\,, \Pb^{0} = \textrm{rand(n,q)}\,,\,E_{\textrm{Total}}^{0}=0$. 
Set $\textrm{maxiterations} \,,\textrm{bbiterations}\,, r\,, 
\lambda\,,\textrm{Tol}$
\For {k = 1:maxiterations}
\State Barzilai-Borwein(BB) descent for $\bm{P}_{k}$. 
\vspace{0.2em}
\State $E_{\textrm{Total}}^{k}= \textrm{Trace }(\bm{P}^{k}(\bm{P}^{k})^{T})+
\frac{\lambda}{2} ||\A(\bm{P}^{k}(\bm{P}^{k})^{T})-\bm{b}||_{2}^{2} $
\vspace{0.2em}
\If {$E_{\textrm{Total}}^{k} < \textrm{Tol}  $}
\State \textbf{break}
\EndIf
\EndFor
\end{algorithmic}
\end{algorithm}

\section{Numerical Results}
\label{sec:Results}
In this section, we demonstrate the efficiency and accuracy of the proposed algorithms. 
All the numerical experiments are ran in MATLAB on a laptop
with an Intel Core I$7$ $2.7$ GHz processor and a $16$GB RAM. 

\subsection{Experiments on synthetic and real data}
We first test the Euclidean distance geometry problem on different three-dimensional objects.
These objects include a sphere, a cow, and a map of a subset of US cities. Given $n$ points from these objects, the full
$n\times n$ distance matrix $\bm{D}$ has $\displaystyle \frac{n(n-1)}{2}=L$ degrees of freedom. 
The objective is to recover the full point coordinates given $\gamma L$ entries of $\D$, $\gamma \in [0,1]$,  chosen uniformly at random. 
Algorithm \ref{alg:trace_min_scheme} and Algorithm \ref{alg:trace_min_scheme_noisy} output $\bm{P}$ from which
$\X = \bm{P}\bm{P}^{T}$ is constructed. Classical MDS is then used to find the global coordinates. In all of the numerical experiments,
a rank estimation $q=10$ is used. This choice shows that the algorithms recover the ground truth
despite a rough estimate of the true rank. The stopping criterion is a tolerance on the relative total energy defined as 
$(E_{\textrm{Total}}^{k}-E_{\textrm{Total}}^{k-1})/E_{\textrm{Total}}^{k}$. For algorithm \ref{alg:trace_min_scheme}, an additional stopping criterion is
a tolerance on $\displaystyle E^{k} = \frac{r}{2}||\A(\bar{\Pb}_{k})-\bm{b}||_{2}^{2}$.
In all of the numerical experiments, the tolerance is set to $10^{-5}$. The maximum iteration is set to $100$. 
Accuracy of the reconstruction is measured using the relative error of the inner product matrix
$\displaystyle \frac{||\X-\M||_{F}}{||\M||_{F}}$ where $\X$ is the numerically computed inner product matrix and
$\M$ is the ground truth. 

For different sampling rates, Figure \ref{fig:exact_figures_list} displays the reconstructed three-dimensional objects
assuming that exact partial information is provided. For all experiments, the penalty term $r$ is set to $1.0$.
Table \ref{tab:relative_error_exact} shows
the relative error of the inner product matrix for all the different cases. The algorithm 
provides very good reconstruction except the $1\%$ sphere. For this specific
case, the distance matrix for the sphere is comparatively small. This means that $1\%$ provides very little
information and more samples are needed for reasonable reconstruction. 

\begin{table}[h!]
\renewcommand{\arraystretch}{1.7}
\caption{Relative error of the inner product matrix for different three-dimensional objects
under different sampling rates. The partial information is exact. Results are averages over $50$ runs. }
\label{tab:relative_error_exact}
\centering
\begin{tabular}{l|llll}
\hline
  $ \gamma$             & $1\%$ & $2\%$ &  $3\%$ & $5\%$ \\ \hline
Sphere          & $3.27e-01$     & $2.10e-03$      & $7.98e-05$        &  $3.21e-05$       \\ 
Cow             & $1.53e-04$     & $7.08e-05$      & $5.03e-05$        &  $3.65e-05$      \\ 
US Cities       & $1.94e-04$     & $1.13e-04$      & $7.53e-05$        &  $5.82e-05$      \\
\hline
\end{tabular}
\end{table}

\begin{figure*}[t!]
\centering

    \subfloat[1\%]{\includegraphics[width=1.5in]{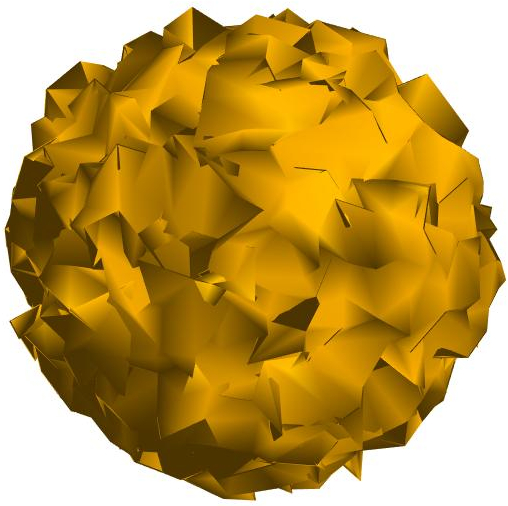}}
\hfil
   \subfloat[2\%]{ \includegraphics[width=1.5in]{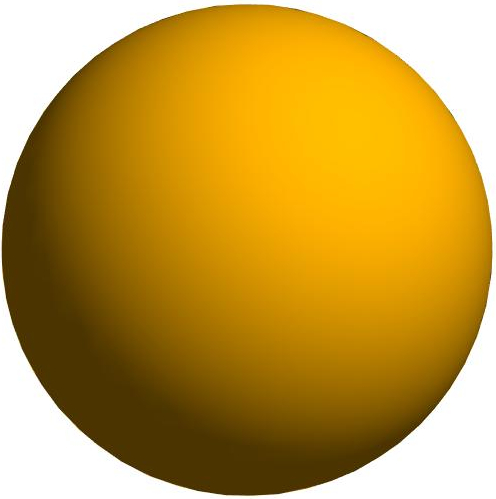}}
\hfil
   \subfloat[3\%]{ \includegraphics[width=1.5in]{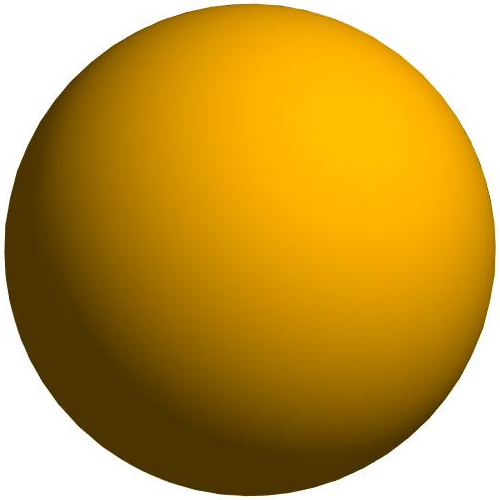}}
\hfil
      \subfloat[5\%]{  \includegraphics[width=1.5in]{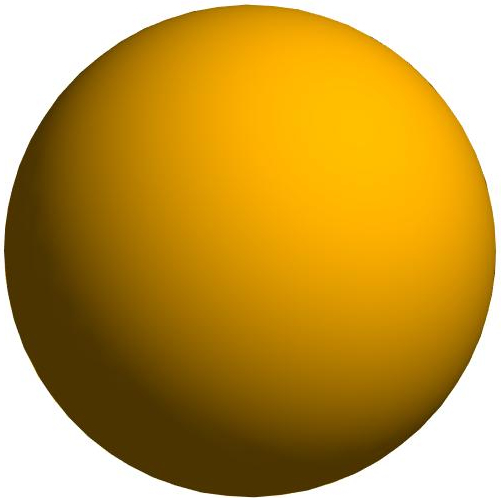}}
\hfil
     \subfloat[1\%]{ \includegraphics[width=1.5in]{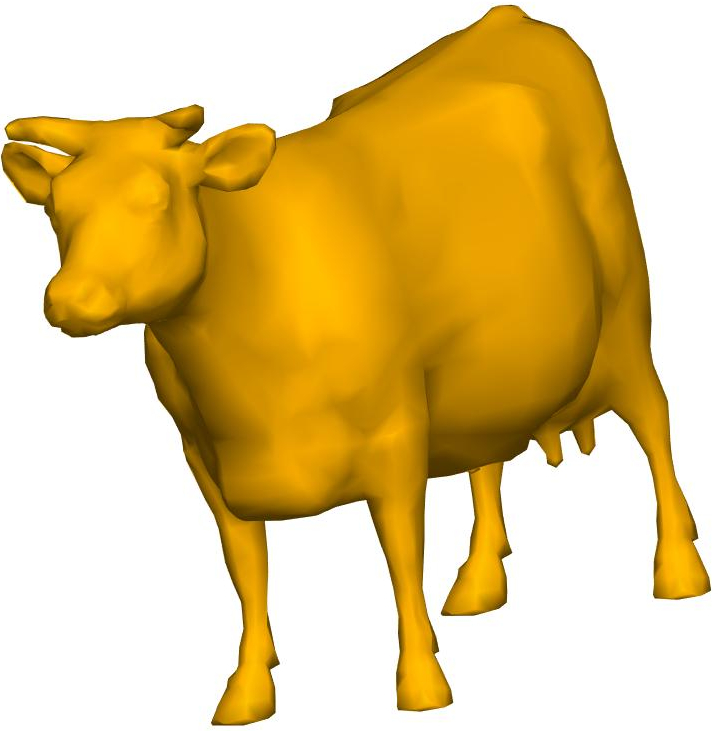}}
\hfil
\subfloat[2\%]{ \includegraphics[width=1.5in]{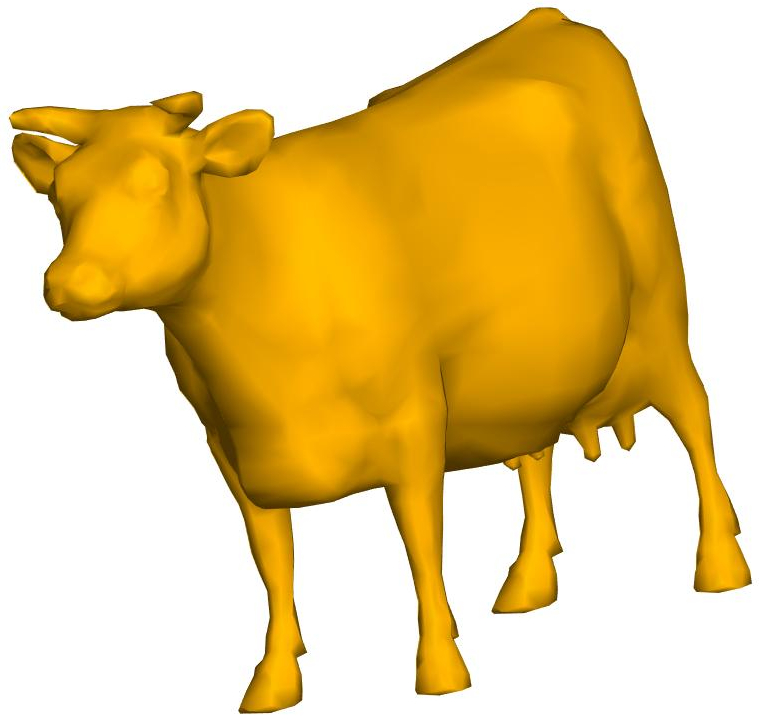}}
\hfil
     \subfloat[3\%]{ \includegraphics[width=1.5in]{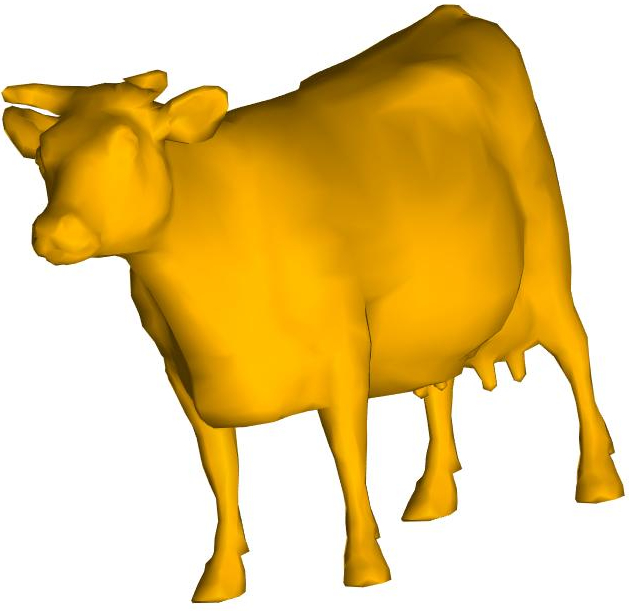}}
\hfil
      \subfloat[5\%]{ \includegraphics[width=1.5in]{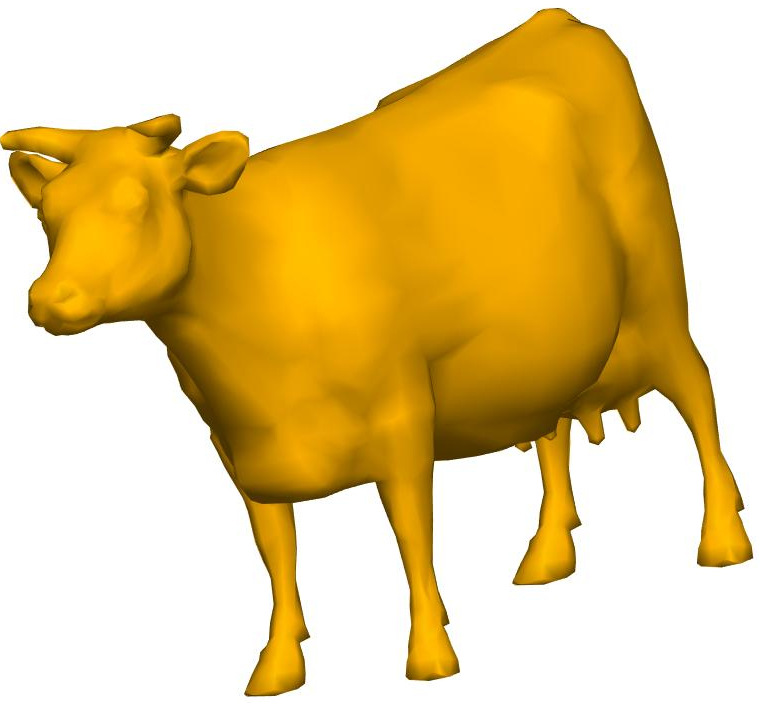}}
\hfil
     \subfloat[1\%]{ \includegraphics[width=1.5in]{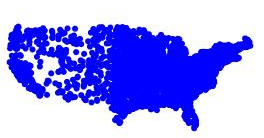}}
\hfil
	\subfloat[2\%]{\includegraphics[width=1.5in]{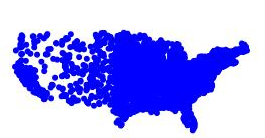}}
\hfil
     \subfloat[3\%]{ \includegraphics[width=1.5in]{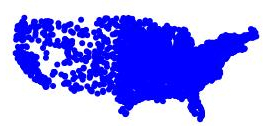}}
\hfil
	%		\hspace{0.5em}
	\subfloat[5\%]{ \includegraphics[width=1.55in]{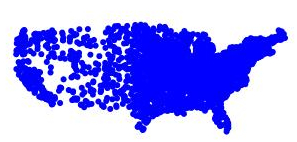}}
 \caption{Reconstruction results of different three-dimensional objects from exact partial information with sampling rates
$1\%$ (1st column), $2\%$ (2nd column), $3\%$ (3rd column) and $5\%$ (4th column), respectively. }
\label{fig:exact_figures_list}
\end{figure*}

For the case of partial information with the additive Gaussian noise $N(\mu,\sigma)$, the noisy distance matrix
can be written as $\bar{\D} = \D +  N(\mu,\sigma)$.
The standard deviation $\sigma$ is a critical parameter determining the extent to which
the underlying information is noisy. In the EDG problem, the perturbed distance matrix $\bar{\D}$
must have non-negative entries. Thus $\mu$ and $\sigma$ need to
be chosen carefully to satisfy this condition.  In our numerical experiments, $\sigma$ is set to be the minimum non-zero value
of the partial distance matrix and $\mu$ is $3\sigma$.  It can be easily verified that,
with very high probability, these choices ensure a non-negative noisy distance matrix.
The choice of $\sigma$ corresponds to the case where the error of the measurement is in the order of 
the minimum distance. While this choice might not reflect practical measurements,
the setting allows us to test the extent to which the algorithm handles a noisy data.
The parameter $\lambda$ is a penalty term which needs to be chosen carefully. 
It can be surmised that $\lambda$ needs to increase with increasing sampling rate. 
For our numerical experiments, a simple heuristic is to set $\lambda = 100\gamma$.
While a more sophisticated analysis might result an optimal choice of $\lambda$, the simple
choice is found to be sufficient and effective in the numerical tests. 
For different sampling rates, Figure \ref{fig:noisy_figures_list} displays the reconstructed three-dimensional objects
assuming that a noisy partial information is provided. 
Table \ref{tab:relative_error_noisy} shows
the relative error of the inner product matrix for all the different cases. 
The algorithm results good reconstruction except the $1\%$ sphere. As discussed
earlier, this specific case requires relatively more samples
for reasonable reconstruction.

\begin{table}[h!]
\renewcommand{\arraystretch}{1.7}
\caption{Relative error of the inner product matrix for different three-dimensional objects
under different sampling rates. The relative error is an average of $50$ runs and the partial information is noisy.}
\label{tab:relative_error_noisy}
\centering
\begin{tabular}{l|llll}
\hline
                & $1\%$ & $2\%$ &  $3\%$ & $5\%$ \\ \hline
Sphere          & $5.21e-01$     & $5.55e-02$     & $1.86e-02$        &  $8.80e-03$       \\ 
Cow             & $2.84e-02$     & $6.60e-03$     & $3.80e-03$        &  $2.00e-03$      \\ 
US Cities       & $2.80e-02$     & $9.00e-03$     & $5.50e-03$        & $3.10e-03$        \\
\hline
\end{tabular}

\end{table}

\begin{figure*}[!t]
\centering
   \subfloat[1\%]{\includegraphics[width=1.5in]{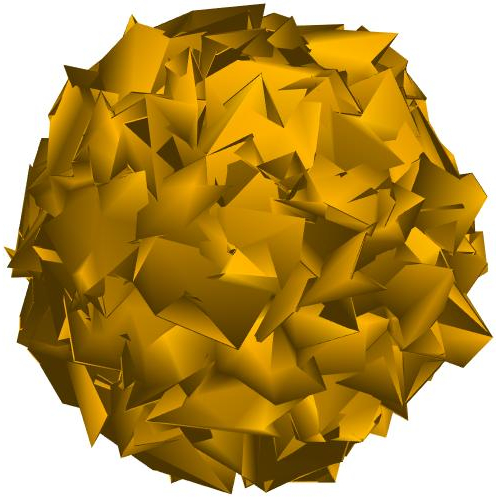}}
\hfil
   %\caption{2\%}
   \subfloat[2\%]{\includegraphics[width=1.5in]{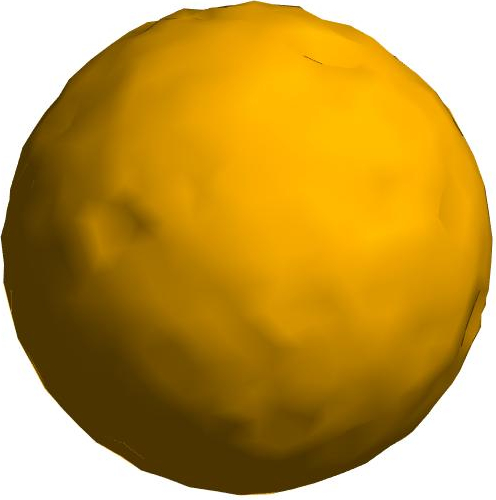}}
\hfil
      %\caption{3\%}
    \subfloat[3\%]{\includegraphics[width=1.5in]{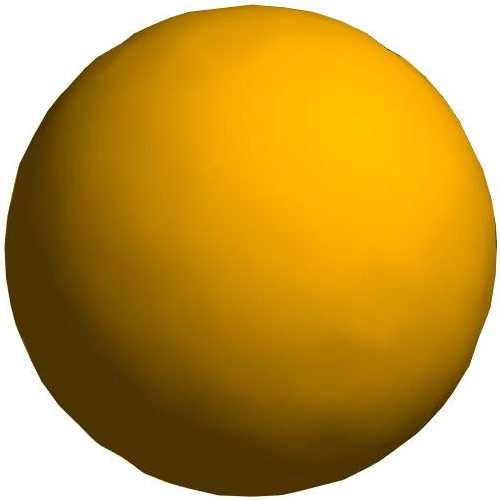}}
\hfil
      %\caption{5\%}
    \subfloat[5\%]{\includegraphics[width=1.5in]{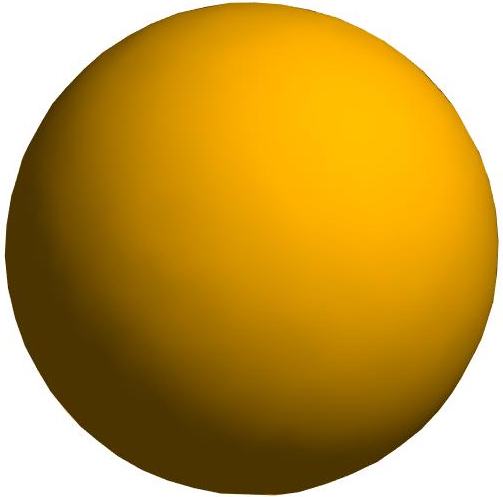}}
\hfil
%\vspace{2em}
    \subfloat[1\%]{\includegraphics[width=1.5in]{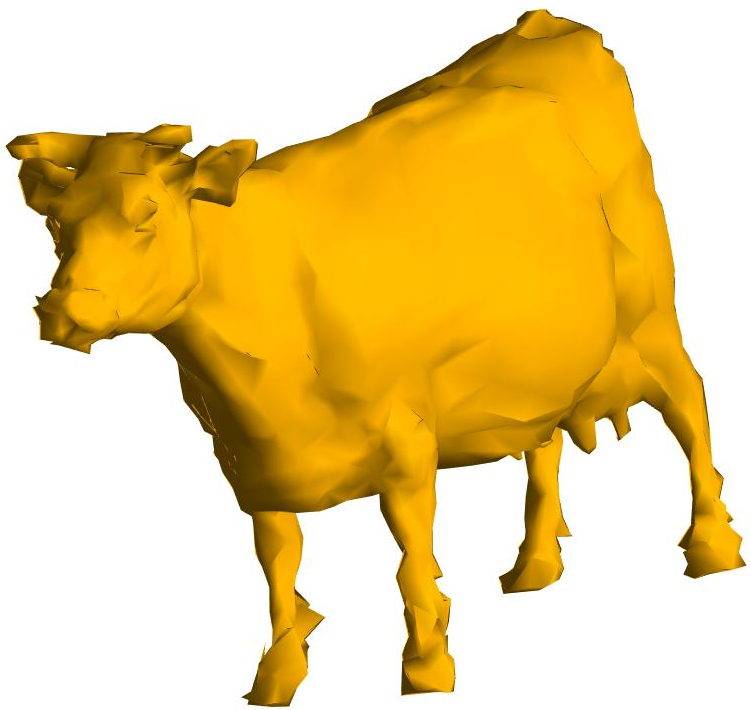}}
\hfil
    \subfloat[2\%]{\includegraphics[width=1.5in]{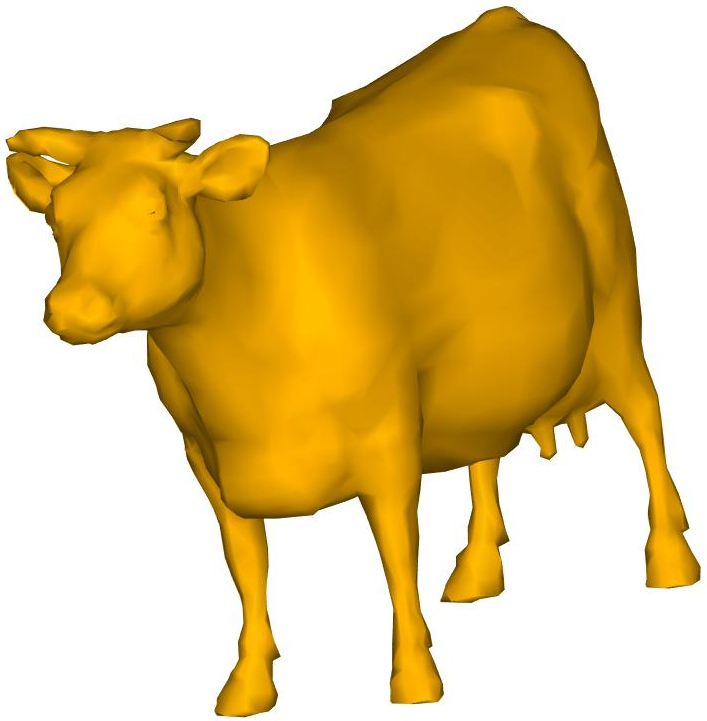}}
\hfil
    \subfloat[3\%]{\includegraphics[width=1.5in]{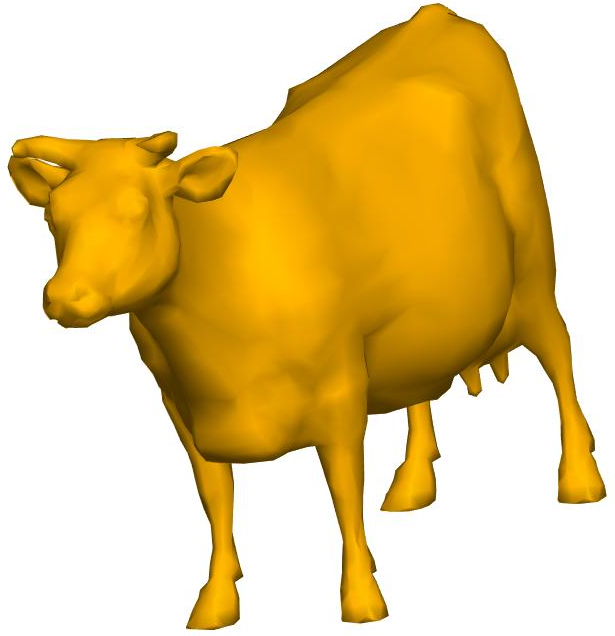}}
\hfil
    \subfloat[5\%]{\includegraphics[width=1.5in]{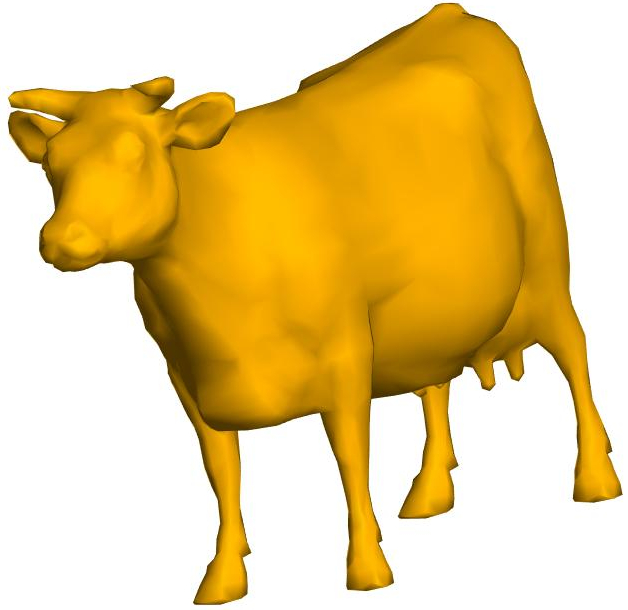}}
\hfil
%\vspace{3em}
    \subfloat[1\%]{\includegraphics[width=1.5in]{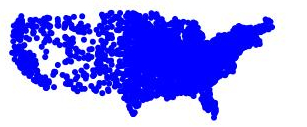}}
\hfil
%\hspace{0.5em}
    \subfloat[2\%]{\includegraphics[width=1.5in]{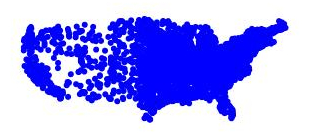}}
\hfil
%\hspace{0.5em}
    \subfloat[3\%]{\includegraphics[width=1.5in]{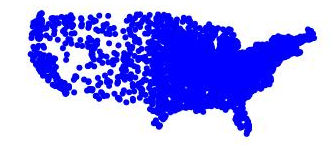}}
\hfil
%\hspace{0.5em}
    \subfloat[5\%]{\includegraphics[width=1.5in]{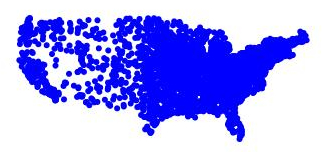}}
%\vspace{-1cm}
\caption{Reconstruction results of different three-dimensional objects from noisy partial information with different sampling rates
$1\%$ (1st column), $2\%$ (2nd column), $3\%$ (3rd column) and $5\%$ (4th column), respectively. }
\label{fig:noisy_figures_list}
\end{figure*}

We further apply the proposed algorithm to the molecular conformation problem~\cite{glunt1993molecular}. 
In this problem, the aim is to determine the three-dimensional
structure of a molecule given partial information on pairwise distances between the atoms.
An instance of this problem is determining the structure of proteins
using nuclear magnetic resonance (NMR) spectroscopy or X-ray diffraction experiments. The problem
is challenging since the partial distance matrix obtained from experiments is
sparse, non-uniform, noisy and prone to outliers~\cite{glunt1993molecular,leung2009sdp}.
We test our method on a simple version of the problem to illustrate that the algorithm can also work
on real data. Our numerical experiment considers two protein molecules identified
as 1AX8 and 1RGS obtained from the Protein Data Bank~\cite{berman2000protein}.
We use a pre-processed version of the data taken from
\cite{leung2009sdp}. Given the full distance matrix,
the partial data is a random uniform sample with $3\%$ sampling rate.
The setup of the numerical experiments is the same as before.  
Figure \ref{fig:protein_figures_list} displays the reconstructed three-dimensional structure of 1AX8
and 1RGS under exact partial information and noisy partial information. The results demonstrate that
the algorithm provides good reconstruction of the underlying three-dimensional structures.

\begin{figure}[t!]
\centering
\subfloat[1AX8]{ \includegraphics[width=1.5in]{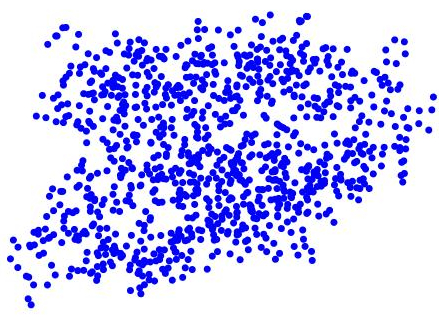}}
\hfil
  \subfloat[1AX8]{ \includegraphics[width=1.5in]{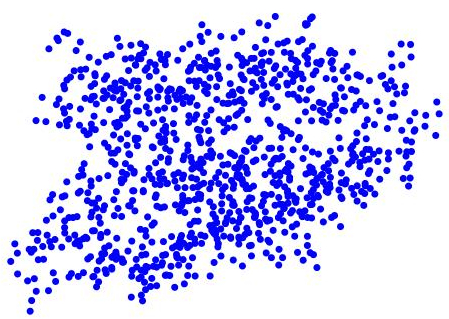}}
\hfil
     \subfloat[1RGS]{ \includegraphics[width=1.5in]{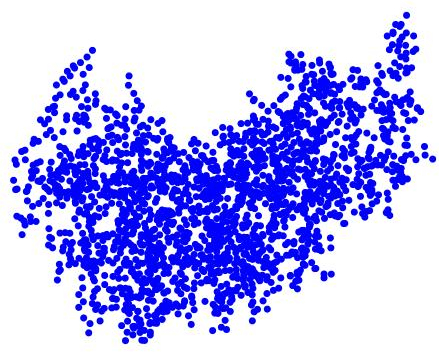}}
\hfil
     \subfloat[1RGS]{ \includegraphics[width=1.5in]{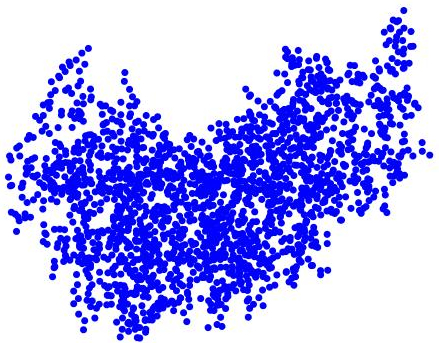}}
\caption{Reconstructed three-dimensional structure of proteins identified as 1AX8 and 1RGS in the Protein Data Bank.
Exact partial information ($3\%$ sample rate) for the first and third, noisy partial information for the second and fourth.}
\label{fig:protein_figures_list}
\end{figure}

\subsection{Computational comparisons of Algorithms}
To evaluate the efficiency of the proposed algorithms, 
numerical tests are carried out using Algorithm\ref{alg:trace_min_scheme} and Algorithm \ref{alg:trace_min_scheme_noisy}. 
For the test, the input is a random uniform sample of the exact/noisy distance matrix of one
of the three-dimensional objects discussed earlier(a sphere, a cow and a map of a subset of US cities). 
The sampling rate is set to $5\%$ and the algorithms are run until a stopping criterion discussed before
is met. In what follows, the reported results are averages of $50$ runs
and the number of iterations is the ceiling of the average number
of iterations. Tables \ref{tab:computational_time} summarizes the
results of the computational experiments for the three-dimensional
objects. It can be concluded that the algorithm is fast and converges in few iterations to
the desired solution.

We also conduct comparisons with the algorithm for the EDG problem proposed in \cite{lai2017solve}  
that also employs the augmented Lagrangian method. As remarked earlier, the main difference between
the two algorithms is on the way the positive semidefinite condition
is imposed on the inner product matrix $\X$. In \cite{lai2017solve}, the positive semidefinite condition
is imposed on $\X$ directly while algorithm \ref{alg:trace_min_scheme} uses the factorization $\bm{P}\bm{P}^{T}$ to enforce
this condition. To compare the two algorithms, we use 
the same input of data which is a random uniform sample of the distance matrix of one
of the three-dimensional objects. The sampling rate is set
to $5\%$. For both algorithms, the stopping criterion is the relative error of the total energy and the tolerance is
set $10^{-5}$. Table \ref{tab:computational_comparison} summarizes the
comparison of these two algorithms. The reported results are averages of $10$ runs.
We see that our algorithm converges to the desired solution faster and with 
significantly less number of iterations.
 
Finally, we consider the molecular conformation problem and compare our algorithm to the DISCO algorithm proposed in \cite{leung2009sdp}. The DISCO algorithm is an SDP based divide and conquer algorithm.
In \cite{leung2009sdp}, the authors demonstrate that
the algorithm is effective on sparse and highly noisy molecular conformation problems.
For the numerical tests, the input for both algorithms is a protein molecule. The sampling rate is set
to $5\%$ and it is assumed that the underlying partial information is exact. We downloaded
the DISCO code, a MATLAB mex code version $1.4$, from \url{http://www.math.nus.edu.sg/~mattohkc/disco.html} which provides
an input file and executables.  For our algorithm, the stopping criterion is maximum number of iterations set to $200$. We emphasize that the rank estimate using our method
for all experiments is $5$. This means our method has 5/3 times variables to the method used in DISCO which compute molecular coordinates directly (i.e. rank number is $3$). The DISCO algorithm has a radius parameter which implies that the input distance matrix consists pairwise distances less than or equal to the radius. For consistent comparison
with the EDG problem and our algorithm, the radius is set large. 
In lack of a source code for DISCO, under the above setups, we run the algorithm as it is. Table \ref{tab:computational_comparison_protein} summarizes the
comparison of these two algorithms. The reported results are averages of $20$ runs.
We see that our algorithm attains a relative error of the same order as DISCO but is  faster on all the tests. 
 Some caveats about the comparison are the assumption
on the radius and a very sparse partial exact information. In \cite{leung2009sdp}, the radius is set
to $\SI{6}{\angstrom}$ since NMR measurements have a limited range
of validity estimated to be $\SI{6}{\angstrom}$. With this choice, the problem departs from
the EDG problem since there is localization. For this localization problem with
a noisy input data and relatively sparse input ($20\%$ of distance within the radius),
we note that DISCO results in excellent reconstruction of the protein molecules. The above
comparison is meant to illustrate that, for a simplistic setup, our algorithm is very fast
and has the potential to handle tests on large protein molecules

\begin{table*}[h!]
\renewcommand{\arraystretch}{1.5}
\caption{Computational summary of Algorithm \ref{alg:trace_min_scheme} from $50$ runs. The data are the different three-dimensional
objects. The sampling rate is $5\%$.
}
\label{tab:computational_time}
\centering
\begin{tabular}{l|llll}
                       & $3$D Object & Number of Points & Computational Time(Sec) & Number of iterations \\\hline
\multirow{4}{*}{Exact} & Sphere     & $1002   $ & $3.27$       &  $9$       \\
                       & Cow        & $2601   $ & $68.08$      &  $26$       \\
                       & US Cities  & $2920   $ & $101.92$       &  $34$     \\ \hline
\multirow{4}{*}{Noisy} & Sphere     & $ 1002  $  &$10.56$       & $19$      \\
                       & Cow        & $ 2601  $ & $62.63$     & $17$      \\
                       & US Cities  & $ 2920  $  &$62.33$      & $18$      \\ \hline
\end{tabular}

\end{table*}

\begin{table*}[h!]
\renewcommand{\arraystretch}{1.5}
\caption{Comparison of the proposed Algorithm~\ref{alg:trace_min_scheme} and the algorithm in \cite{lai2017solve} with $10$ runs.
The data are the different three-dimensional objects. It is assumed that the partial information is exact. The sampling rate is $5\%$.}
\label{tab:computational_comparison}
\centering
\begin{tabular}{llllll}
  $3$D Object & Number of Points & \multicolumn{2}{c}{Computational Time(Sec)} & \multicolumn{2}{c}{Number of iterations} \\
              &                  & Alg.~\ref{alg:trace_min_scheme} & Alg.~\cite{lai2017solve}  &     Alg.~\ref{alg:trace_min_scheme} & Alg.\cite{lai2017solve}\\ \hline
 Sphere     & $1002$    & $3.06$    &$101.85$      &  $9$   & $204$        \\
 Cow        & $2601$    & $51.40$   &$1387.80$      &  $20$   & $257$       \\
 US Cities  & $2920$    & $72.51$  & $2417.30$     &  $23$   & $250$       \\ \hline

\end{tabular}
\end{table*}

\begin{table*}[h!]
\renewcommand{\arraystretch}{1.5}
\caption{Comparison of the proposed Algorithm~\ref{alg:trace_min_scheme} and the algorithm in 
\cite{leung2009sdp} with $20$ runs.
The data are the different protein molecules. It is assumed that the partial information is exact. The sampling rate is $5\%$.}
\label{tab:computational_comparison_protein}
\centering
\begin{tabular}{llllll}
  Protein Molecule & Number of Points & \multicolumn{2}{c}{Computational Time(Sec)} & 
\multicolumn{2}{c}{Relative error of the inner product matrix} \\
  &                  & Alg.~\ref{alg:trace_min_scheme} & Alg.~\cite{leung2009sdp}  &  Alg.~\ref{alg:trace_min_scheme} & Alg.~\cite{leung2009sdp}\\ \hline
 1PTQ       & $402$     & $3.38$    &$11.97$      &  $8.03e-07$   & $1.63e-07$        \\
 1AX8       & $1003$    & $9.44$    &$45.94$      &  $1.02e-08$   & $8.01e-07$       \\
 1RGS       & $2015$    & $52.62$   &$214.30$      &  $5.41e-09$   & $1.62e-07$       \\
 1KDH       & $2923$    & $104.60$   &$438.04$      &  $2.86e-09$   & $8.16e-09$       \\
 1BPM       & $3672$    & $223.14$   &$392.29$      &  $3.46e-09$   & $3.27e-08$       \\
\end{tabular}
\end{table*}

\subsection{Phase transition of Algorithm \ref{alg:trace_min_scheme}}
 
Last but not least, we numerically investigate the optimality of the proposed algorithm by plotting the phase transition. 
Given the number of points
and the underlying rank, the theory provides the sampling rate
which leads to successful recovery with very high probability (see
Theorem \ref{thm1}). To investigate the optimality of Algorithm \ref{alg:trace_min_scheme},
the following numerical experiment was carried out. Consider sampling
rates ranging from $1\%$ to $100\%$ and rank ranging from
$1$ to $40$. For each pair, Algorithm \ref{alg:trace_min_scheme} is run
$50$ times. Successful recovery refers to the case where the relative error of the inner product
matrix is within tolerance. As remarked earlier, the tolerance is set to $10^{-5}$.
Out of the $50$ runs, the number of times the algorithm succeeds
provides us with a probability of success. This procedure is repeated
for all combination of sampling rate and rank. Figure \ref{fig:borderline}
shows the optimality result of Algorithm \ref{alg:trace_min_scheme}. Namely, for a large portion in the sampling rate-rank domain, the proposed algorithm can provide successful reconstruction. 

\begin{figure}[h!]
\centering
\subfloat{\includegraphics[width=5in]{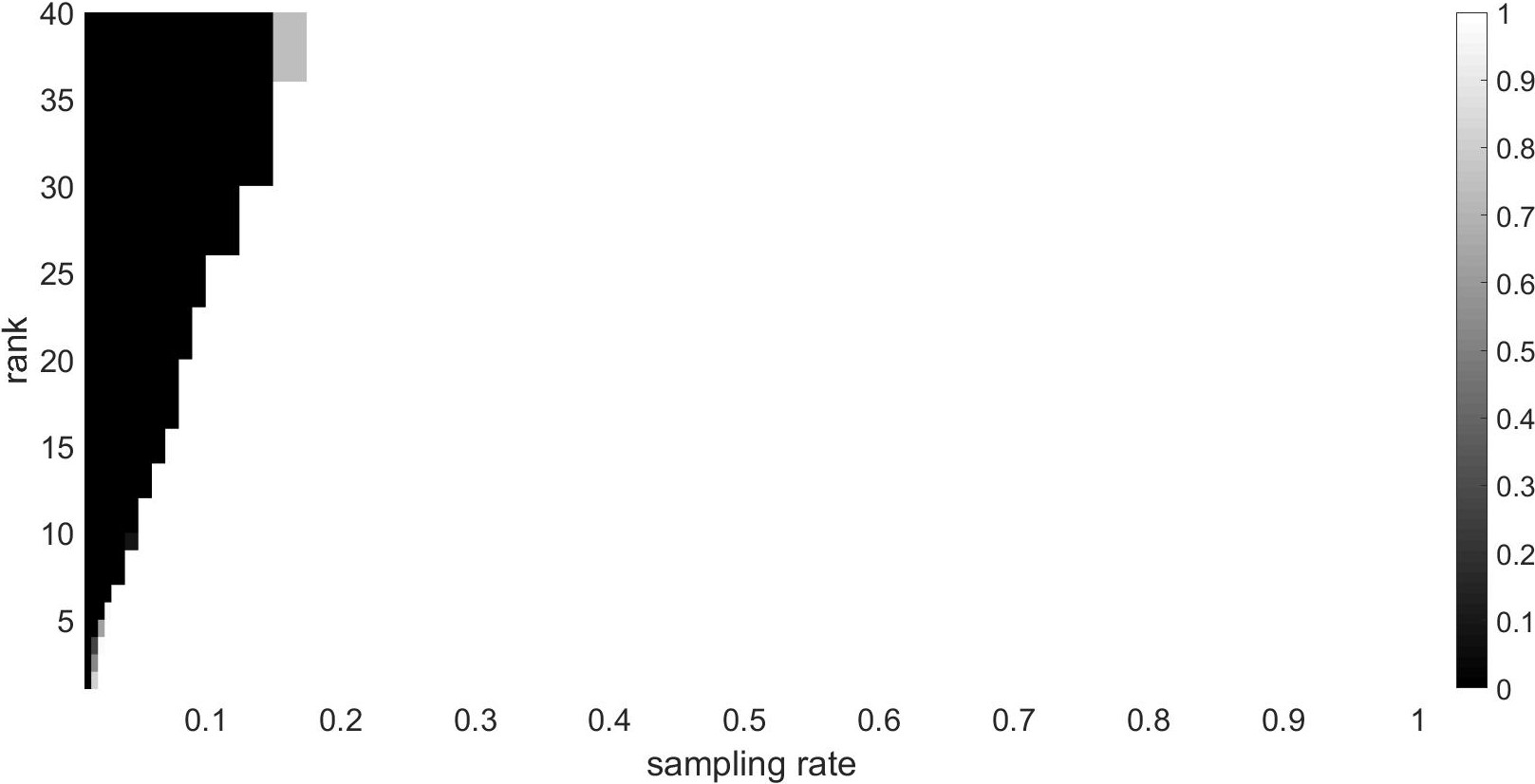}}
\caption{Success probability of Algorithm \ref{alg:trace_min_scheme} given sampling rate and
rank. For each sampling rate, Algorithm \ref{alg:trace_min_scheme} is run for ranks ranging
from $1$ to $40$. The average success probability, of $50$ runs, is shown above. White
indicates perfect recovery and black indicates failure of the algorithm. }
\label{fig:borderline}
\end{figure}

\section{Conclusion}

In this paper, we formulate the Euclidean distance geometry (EDG) problem 
as a low rank matrix recovery problem. Adopting the matrix completion
framework, our approach can be viewed as completing the Gram matrix with
respect to a suitable basis given few uniformly random distance samples. 
However, the existing analysis based on the restricted isometry property (RIP) does not hold for our problem. 
Alternatively, we conduct analysis by introducing the dual basis approach to
formulate the EDG problem.
Our main result shows that the underlying configuration of points can be recovered
with very high probability from $O(nr\nu\log^{2}(n))$ measurements if the underlying Gram matrix obeys the coherence condition with parameter $\nu$. 
Numerical algorithms are designed to solve the EDG problem
under two scenarios, exact and noisy partial information. 
Numerical experiments on various test data demonstrate that the algorithms are simple,
fast and accurate. The technique in this paper is not specifically limited to the EDG problem.
The extension of this result to the low rank recovery of a matrix given few measurements with respect to
any non-orthogonal basis is a work in preparation. 
%In our future work, we will extend our result and explore the low rank recovery
%of a matrix given few measurements with respect to any non-orthogonal basis. 

\section{Acknowledgment}
The authors would like to thank Dr.\,\,Jia Li for his discussions in the early stage of this project. 
Abiy Tasissa would like to thank Professor David Gross for correspondence over email
regarding the work in \cite{gross2011recovering}. Particularly, the proof of Lemma \ref{operator_norm_of_sum}
is a personal communication from Professor David Gross. The authors also would like to thanks Professor Peter Kramer and Professor Alex Gittens for their comments and suggestions.
The authors' gratitude is further extended to the anonymous reviewers for their valuable feedback which has improved
the manuscript.

\bibliographystyle{IEEEtran} 
\bibliography{IEEEabrv,matrix_completion}

%\clearpage
\appendix 
\section{Appendix A}
\begin{lemma} \label{signx_prop}
If $\X\in \S$, $\textrm{Sgn}\,\X \in \S$. If $\X \in \T$, $\textrm{Sgn}\,\X \in \T$
\end{lemma}
\begin{proof}
Using the eigenvalue decomposition of $\X$, $\X = \U \Sigma \U^{T}$. $\textrm{Sgn}\,\X$ is simply
$\textrm{Sgn}\,\X = \U(\textrm{Sgn}\,\Sigma) \U^{T} = \U \D \U^{T}$
where $\D$ is the diagonal matrix resulting from applying the
sign function to $\Sigma$. To show that $\textrm{Sgn}\,\X \in \S$, we need to verify
$\textrm{Sgn}\,\X = (\textrm{Sgn}\,\X)^T$ and $\textrm{Sgn}\,\X\cdot \bm{1}=\bm{0}$.
Symmetry of $\textrm{Sgn}\,\X$ is apparent from its definition, $ (\textrm{Sgn}\,\X)^T = \U \D^T \U^T = \U \D \U^T = \textrm{Sgn}\,\X$.
To show that $\textrm{Sgn}\,\X \cdot \bm{1} =\bm{0} $, consider $\X \cdot \bm{1}=\bm{0}$ 
using the spectral decomposition of $\X$.
\[
\left(\sum_{i} \lambda_{i} \bm{u}_{i} \bm{u}_{i}^T \right) \cdot \bm{1} = \bm{0} \longrightarrow
\bm{u}_{j}^T \left(\sum_{i} \lambda_{i} \bm{u}_{i} \bm{u}_{i}^T \right)  \cdot \bm{1} = \bm{0} \longrightarrow
(\lambda_{j} \bm{u} _{j}^T) \cdot \bm{1} = \bm{0}
\]
The implication is that $\lambda_j=0$ or $\bm{u} _{j}^T \cdot \bm{1} = \bm{0}$.
With this, consider the spectral decomposition of the
symmetric matrix $\textrm{Sgn}\,\X$.
\[
\textrm{Sgn}\,\X = \sum_{j} \textrm{sgn}\,(\lambda_j)\, \bm{u}_{j} \bm{u}_{j}^T
\]
$\textrm{Sgn}\,\X \cdot \bm{1} =\bm{0} $ follows in the following way. If $\lambda_j =0$,
$\textrm{sgn}\,(\lambda_j)=0$. Otherwise, from above, $\bm{u} _{j}^T \cdot \bm{1} = \bm{0}$.
It can now be concluded that $\textrm{Sgn}\,\X \in \S$. Next, we show that for
$\X \in \T$, $\textrm{Sgn}\,\X \in \T$. 
Using the eigenvalue decomposition of $\textrm{Sgn } \X$,
consider $\P_{\T^{\perp}}\, \textrm{Sgn } \X$.
\begin{align*}
\P_{\T^{\perp}} \textrm{Sgn } \X & = \textrm{Sgn } \X - \P_{\T}\, \textrm{Sgn } \X\\
                                 & = \U\D\U^T - [ \P_{\U}\, \textrm{Sgn } \X +  \textrm{Sgn } \X \,\P_{\U}
                                    - \P_{\U} \, \textrm{Sgn } \X \,\P_{\U} ] \\
                                & = \U\D\U^T - [ \U\U^T\U\D\U^T + \U\D\U^T \U\U^T - \U\U^T \U\D\U^T \U\U^T] = \bm{0}
\end{align*}
where the last step simply follows from the fact that $\U^T \U = \I$. This confirms that
$\textrm{sgn } \X \in \T$ and concludes the proof. 
\end{proof}

In the next two lemmas, we state and prove some facts about $\H$ and $\H^{-1}$. 
We start with the latter first deriving an explicit form of $\H^{-1}$. This will be the focus
of Lemma \ref{H_inverse} to follow shortly. In the proof of Lemma \ref{H_inverse}, a certain form of the basis $\v_{\alphab}$ is used. The form is conjectured by inspection of the basis $\v_{\alphab}$ for the case when $n$ is small. 
We start by stating and proving this form.

\begin{lemma}
 Given an index $(i,j)$ with $1 \le  i< j\le n$, the matrix $\v_{i,j}$ has the following form. 
\[
\v_{i,j} = \wt_{i,j}+\p_{i,j}+\q_{i,j}
\]
The matrices $\wt_{ij}, \p_{i,j}$ and $\q_{ij}$ are respectively defined as follows. 
\[
\wt_{i,j} = \frac{n-1}{n^2} \e_{i,i}+\frac{n-1}{n^2}\e_{j,j}-\frac{(n-1)^2+1}{2n^2}\e_{i,j} -\frac{(n-1)^2+1}{2n^2}\e_{j,i}
\]
where $\e_{\alpha_1,\,\alpha_2}$
is a matrix of zeros except a $1$ at the location $(\alpha_1,\alpha_2)$. $\p_{ij}$ has the following form.
\[
\p_{i,j}= \sum_{t,t\neq i, t \neq j} \frac{2n-4}{4n^2}\e_{i,t} +  \sum_{s,s\neq i, s \neq j} \frac{2n-4}{4n^2} \e_{s,j}
\]
$\q_{ij}$ is defined as follows.
\[
\q_{i,j} = \sum_{(t,s),(t,s)\neq (i,j)}-\frac{1}{n^2} \e_{s,t} 
\]
where $(t,s)\neq (i,j)$ is defined as $\{t,s\}\cap \{i,j\}= \varnothing$.

\end{lemma}

\begin{proof}

To make the proposed form concrete, before proceeding with the proof, consider the following example.

\paragraph{Example $1$: The form of $\v_{1,2}$}
Consider the case where $n= 5$. Using the proposed form, the matrix $\v_{1,2}$ can be written as follows
\[
\v_{1,2} = \wt_{1,2}+ \p_{1,2} + \q_{1,2}
\]
where
\[
\wt_{1,2}= \frac{1}{25}
\begin{pmatrix}
4 & -\frac{17}{2} & 0 & 0 & 0\\
-\frac{17}{2} & 4 & 0 & 0 & 0\\
0 & 0 & 0 & 0 & 0\\
0 & 0 & 0 & 0 & 0\\
0 & 0 & 0 & 0 & 0\\
\end{pmatrix}
\]
\[
\p_{1,2} = \frac{1}{25}
\begin{pmatrix}
0 & 0 & 1.5 & 1.5 & 1.5 \\
0 & 0 & 1.5 & 1.5 & 1.5 \\
1.5 & 1.5 & 0 & 0 & 0\\
1.5 & 1.5 & 0 & 0 & 0\\
1.5 & 1.5 & 0 & 0 & 0\\
\end{pmatrix}
\]
and
\[
\q_{1,2} = \frac{1}{25}
\begin{pmatrix}
0 & 0 & 0 & 0 & 0\\
0 & 0 & 0 & 0 & 0\\
0 & 0 & -1 & -1 & -1\\
0 & 0 &-1 & -1 & -1\\
0 & 0 &-1 & -1 & -1\\
\end{pmatrix}
\] 
Therefore, $\v_{1,2}$ has the following explicit form
\[
\v_{1,2}= \frac{1}{25}
\begin{pmatrix}
4 & -\frac{17}{2} & 1.5 & 1.5 & 1.5\\
-\frac{17}{2} & 4 & 1.5 & 1.5 & 1.5\\
1.5 & 1.5 & -1 & -1 & -1\\
1.5 & 1.5 & -1 & -1 & -1\\
1.5 & 1.5 & -1 & -1 & -1\\
\end{pmatrix}
\]
As a first step simple check, consider $\langle \v_{1,2}\,,\w_{1,2}\rangle$ which results $\frac{1}{25}(8+17)=1$ as desired. \\*

The proof of Lemma $1$ relies on showing that $\v_{i,j}$ is dual to $\w_{i,j}$. Since the dual basis is unique, establishing duality
will conclude the proof. The result will be shown considering different cases.\\*

\noindent
Case $1$: $\langle \v_{i,j}\,,\w_{i,j}\rangle$ 
\begin{align*}
\langle \v_{i,j}\,,\w_{i,j}\rangle = \langle \wt_{i,j}+\p_{i,j}+\q_{i,j}\,,\w_{i,j}\rangle
&= \langle \wt_{i,j}\,,\w_{i,j}\rangle + \langle \p_{i,j}\,,\w_{i,j}\rangle + \langle \q_{i,j}\,,\w_{i,j}\rangle \\
& = \frac{2(n-1)}{n^2}+\frac{2(n-1)^2+2}{2n^2} + 0 + 0 = 1
\end{align*}
The second to last equality follows since $\langle \p_{i,j}\,,\w_{i,j}\rangle = 0$.
and $\langle \q_{i,j}\,,\w_{i,j}\rangle =0$.\\*

\noindent
Case $2$: $\langle \v_{i,j}\,,\w_{\alpha,\,\beta}\rangle$ for $(i,j)\neq (\alpha,\beta)$
\begin{align*}
\langle \v_{i,j}\,,\w_{\alpha,\,\beta}\rangle = \langle \wt_{i,j}+\p_{i,j}+\q_{i,j}\,,\w_{\alpha,\,\beta}\rangle
&= \langle \wt_{i,j}\,,\w_{\alpha,\,\beta}\rangle + \langle \p_{i,j}\,,\w_{\alpha,\,\beta}\rangle + \langle \q_{i,j}\,,\w_{\alpha,\,\beta}\rangle \\
& = 0 + 0 + 0 = 0
\end{align*}

\noindent
Case $3$: $\langle \v_{i,j}\,,\w_{\alpha,\,\beta}\rangle$ for $\{i,j\}\cap\{\alpha,\beta\}\neq \varnothing$

\vspace{0.5em}

A. $i =\alpha$ and $j \neq \beta$
\begin{align*}
\langle \v_{i,j}\,,\w_{\alpha,\,\beta}\rangle = \langle \wt_{i,j}+\p_{i,j}+\q_{i,j}\,,\w_{\alpha,\,\beta}\rangle
&= \langle \wt_{i,j}\,,\w_{\alpha,\,\beta}\rangle + \langle \p_{i,j}\,,\w_{\alpha,\,\beta}\rangle + \langle \q_{i,j}\,,\w_{\alpha,\,\beta}\rangle \\
& = \frac{n-1}{n^2}-\frac{2n-4}{2n^2}-\frac{1}{n^2} = 0 
\end{align*}

B. $i\neq \alpha$ and $j = \beta$
\begin{align*}
\langle \v_{i,j}\,,\w_{\alpha,\,\beta}\rangle = \langle \wt_{i,j}+\p_{i,j}+\q_{i,j}\,,\w_{\alpha,\,\beta}\rangle
&= \langle \wt_{i,j}\,,\w_{\alpha,\,\beta}\rangle + \langle \p_{i,j}\,,\w_{\alpha,\,\beta}\rangle + \langle \q_{i,j}\,,\w_{\alpha,\,\beta}\rangle \\
& = \frac{n-1}{n^2}-\frac{2n-4}{2n^2}-\frac{1}{n^2} = 0
\end{align*}

C. $i =\beta$ and $j \neq  \alpha$
\begin{align*}
\langle \v_{i,j}\,,\w_{\alpha,\,\beta}\rangle = \langle \wt_{i,j}+\p_{i,j}+\q_{i,j}\,,\w_{\alpha,\,\beta}\rangle
&= \langle \wt_{i,j}\,,\w_{\alpha,\,\beta}\rangle + \langle \p_{i,j}\,,\w_{\alpha,\,\beta}\rangle + \langle \q_{i,j}\,,\w_{\alpha,\,\beta}\rangle \\
& = \frac{n-1}{n^2}-\frac{2n-4}{2n^2}-\frac{1}{n^2} = 0
\end{align*}

D. $i \neq \beta $ and $j = \alpha$
\begin{align*}
\langle \v_{i,j}\,,\w_{\alpha,\,\beta}\rangle = \langle \wt_{i,j}+\p_{i,j}+\q_{i,j}\,,\w_{\alpha,\,\beta}\rangle
&= \langle \wt_{i,j}\,,\w_{\alpha,\,\beta}\rangle + \langle \p_{i,j}\,,\w_{\alpha,\,\beta}\rangle + \langle \q_{i,j}\,,\w_{\alpha,\,\beta}\rangle \\
& = \frac{n-1}{n^2}-\frac{2n-4}{2n^2}-\frac{1}{n^2} = 0
\end{align*}

With this, it can be concluded that the basis $\{\v_{i,j}\}$ is dual to $\{\w_{i,j}\}$ and the proposed form is established. 
\end{proof}

\noindent
The next Lemma uses the result of the above Lemma to derive an explicit form of the matrix $\H^{-1}$. 

\begin{lemma}\label{H_inverse}
The inverse of the matrix $\H$, $\H^{-1}$, has the following explicit form
\[
\H^{\alphab,\,\betab} = \langle \v_{\alphab}\,,\v_{\betab}\rangle=
\begin{cases}
\displaystyle\frac{(n-1)^{2}+1}{2n^2} & \textrm{if }\alpha_1 = \beta_1 \,\&\, \alpha_2 = \beta_2 \\
\displaystyle\frac{1}{n^2} & \textrm{if } \alpha_1 \neq \beta_1 \,\&\, \alpha_1 \neq \beta_2\,\&\,\alpha_2 \neq \beta_1 \,\&\, \alpha_2 \neq \beta_2\\
\displaystyle\frac{4-2n}{4n^2} & \textrm{otherwise}
\end{cases}
\]
\end{lemma}

\begin{proof} We consider three different cases. 

\noindent
Case $1$: $\langle \v_{i,j}\,,\v_{i,j}\rangle$
\begin{align*}
\langle \v_{i,j}\,,\v_{i,j}\rangle & = \langle \wt_{i,j}+\p_{i,j} +\q_{i,j}\,, \wt_{i,j}+\p_{i,j} +\q_{i,j} \rangle \\
                                                & =  \langle \wt_{i,j}\,,\wt_{i,j}\rangle + \langle \p_{i,j}\,,\p_{i,j}\rangle + \langle \q_{i,j}\,,\q_{i,j}\rangle\\
                                                & = 2\frac{(n-1)^2}{n^4}+2\bigg[\frac{(n-1)^2+1}{2n^2} \bigg]^2 +\left(\frac{2n-4}{4n^2} \right)^24(n-2)+\frac{1}{n^4}(n-2)^{2} \\
                                                & = \frac{2n^2-4n+2}{n^4}+2\left( \frac{n^4-4n^3+8n^2-8n+4}{4n^4} \right)+\frac{16n^3-96n^2+192n-128}{16n^4}+\frac{n^2-4n+4}{n^4}\\
                                                & = \frac{1}{2}-\frac{1}{n}+\frac{1}{n^2}=  \frac{(n-1)^{2}+1}{2n^2}
\end{align*} 
The second equality uses the fact that $\langle \wt_{i,j}\,,\p_{i,j}\rangle=0$, $\langle \p_{i,j}\,,\q_{i,j}\rangle=0$ and $\wt_{i,j}\,,\q_{i,j}=0$. 
The third inequality simply uses the definition of $\wt_{i,j}, \p_{i,j}$ and $\q_{i,j}$.
The second to last result follows after some algebraic calculations. \\*

\noindent
Case $2$: $\langle \v_{i,j}\,,\v_{\alpha,\,\beta}\rangle$ for $(i,j)\neq (\alpha,\beta)$, $\{i,j\}\cap \{\alpha,\,\beta\} = \varnothing$
\begin{align*} 
\langle \v_{i,j}\,,\v_{\alpha,\,\beta}\rangle  =& \langle \wt_{i,j}+\p_{i,j}+\q_{i,j}\,,  \wt_{\alpha,\,\beta}+\p_{\alpha,\,\beta}+\q_{\alpha,\,\beta}\\
                                                                  = & \langle \wt_{i,j}\,, \wt_{\alpha,\,\beta}\rangle +  \langle \wt_{i,j}\,,\p_{\alpha,\,\beta}\rangle +  \langle \wt_{i,j}\,,\q_{\alpha,\,\beta}\rangle+\\
                                                                         &\langle \p_{i,j}\,,\wt_{\alpha,\,\beta}\rangle +    \langle \p_{i,j}\,,\p_{\alpha,\,\beta}\rangle +
                                                                            \langle \p_{i,j}\,,\q_{\alpha,\,\beta}\rangle +    \langle \q_{i,j}\,,\wt_{\alpha,\,\beta}\rangle +   \langle \q_{i,j}\,,\p_{\alpha,\,\beta}\rangle +
                                                                               \langle \q_{i,j}\,,\q_{\alpha,\,\beta}\rangle 
\end{align*}
Each term will be evaluated separately. \\

\noindent
$1$. $ \langle \wt_{i,j}\,, \wt_{\alpha,\,\beta}\rangle  = 0$ since for $(i,j)\neq (\alpha,\beta)$, $ \wt_{i,j}$ and $\wt_{\alpha,\,\beta}$ have disjoint 
supports.  \\*

\noindent
$2$. $\langle \wt_{i,j}\,,\p_{\alpha,\,\beta}\rangle = 0$ since for $(i,j)\neq (\alpha,\beta)$, $ \wt_{i,j}$ and $\p_{\alpha,\,\beta}$ have disjoint 
supports.  \\*

\noindent
$3$. Consider $\langle \wt_{i,j}\,,\q_{\alpha,\,\beta}\rangle$. 
\begin{align*}
\langle \wt_{i,j}\,,\q_{\alpha,\,\beta}\rangle &= [\wt_{i,j}]_{i,i}[\q_{\alpha,\,\beta}]_{i,i}+[\wt_{i,j}]_{j,j}[\q_{\alpha,\,\beta}]_{j,j}+[\wt_{i,j}]_{i,j}[\q_{\alpha,\,\beta}]_{i,j}+
[\wt_{i,j}]_{j,i}[\q_{\alpha,\,\beta}]_{j,i} \\
&= -2\frac{(n-1)}{n^4}+\frac{(n-1)^2+1}{n^4}
\end{align*}

\noindent
$4$. $\langle \p_{i,j}\,,\wt_{\alpha,\,\beta}\rangle = 0$ follows by a similar argument as $2$. \\*

\noindent
$5$. Consider $\langle \p_{i,j}\,,\p_{\alpha,\,\beta}\rangle $. 
\begin{align*}
\langle \p_{i,j}\,,\p_{\alpha,\,\beta}\rangle = & \sum_{s,s\neq i, s \neq j } [\p_{i,j}]_{i,s}  [\p_{\alpha,\,\beta}]_{i,s} +  \sum_{t,t\neq i, t \neq j } [\p_{i,j}]_{j,t}  [\p_{\alpha,\,\beta}]_{j,t} +
\sum_{s,s\neq i, s \neq j } [\p_{i,j}]_{s,i}  [\p_{\alpha,\,\beta}]_{s,i} +  \sum_{t,t\neq i, t \neq j } [\p_{i,j}]_{t,j}  [\p_{\alpha,\,\beta}]_{t,j} \\
 = &   [\p_{i,j}]_{i,\,\alpha}[\p_{\alpha,\,\beta}]_{i,\,\alpha} + [\p_{i,j}]_{i,\,\beta}[\p_{\alpha,\,\beta}]_{i,\,\beta} +
 [\p_{i,j}]_{j,\,\alpha}  [\p_{\alpha,\,\beta}]_{j,\,\alpha}+[\p_{i,j}]_{j,\,\beta}  [\p_{\alpha,\,\beta}]_{j,\,\beta}+\\
   &[\p_{i,j}]_{\alpha,\,i}  [\p_{\alpha,\,\beta}]_{\alpha,\,i} +  [\p_{i,j}]_{\beta,\,i}  [\p_{\alpha,\,\beta}]_{\beta,\,i}+
    [\p_{i,j}]_{\alpha,\,j}  [\p_{\alpha,\,\beta}]_{\alpha,\,j} +  [\p_{i,j}]_{\beta,\,j}  [\p_{\alpha,\,\beta}]_{\beta,\,j}\\
    = & 8\left(\frac{2n-4}{4n^2}\right)^{2} = \frac{2(n-2)^2}{n^4}
\end{align*}

\noindent
$6$. Consider $\langle \p_{i,j}\,,\q_{\alpha,\,\beta}\rangle$. 
\begin{align*}
\langle \p_{i,j}\,,\q_{\alpha,\,\beta}\rangle &= \sum_{t,t \neq i,t\neq j} [\p_{i,j}]_{i,t} [\q_{\alpha,\,\beta}]_{i,t} + \sum_{s,s\neq i,s\neq j} [\p_{i,j}]_{j,s} [\q_{\alpha,\,\beta}]_{j,s}+
\sum_{t,t \neq i,t\neq j} [\p_{i,j}]_{t,i} [\q_{\alpha,\,\beta}]_{t,i} + \sum_{s,s\neq i,s\neq j} [\p_{i,j}]_{s,j} [\q_{\alpha,\,\beta}]_{s,j}\\
& = (4n-16)\frac{2n-4}{4n^2}\left(-\frac{1}{n^2}\right) = -\frac{2(n-4)(n-2)}{n^4}
\end{align*}
The first equality follows since it suffices to consider $\langle \p_{i,j}\,,\q_{\alpha,\,\beta}\rangle$ on the support of $\p_{i,j}$.
The second and last equality result from the following analysis. 
Restricted to the support of $\p_{i,j}$, the matrix $\q_{\alpha,\,\beta}$ is non-zero except at the entries
$[\q_{\alpha,\,\beta}]_{i,\,\alpha}, [\q_{\alpha,\,\beta}]_{i,\,\beta}, [\q_{\alpha,\,\beta}]_{\alpha,\,i}$, $[\q_{\alpha,\,\beta}]_{\beta,\,i},
[\q_{\alpha,\,\beta}]_{j,\,\alpha}, [\q_{\alpha,\,\beta}]_{j,\,\beta}, [\q_{\alpha,\,\beta}]_{\alpha,j}$ and $[\q_{\alpha,\,\beta}]_{\beta,j}$. 
Using this and the fact that $\p_{i,j}$ has $4(n-2)$ entries, $|\textrm{supp}(\p_{i,j})\cap \textrm{supp}(\q_{\alpha,\,\beta})| = 4(n-2)-8 = 4n-16$.
With this, the final form above follows. \\*

\noindent
$7$. $\displaystyle  \langle \q_{i,j}\,,\wt_{\alpha,\,\beta}\rangle =  -2\frac{(n-1)}{n^4}+\frac{(n-1)^2+1}{n^4}$ follows by a similar argument as  $3$. \\*
 
\noindent
$8$. $\displaystyle  \langle \q_{i,j}\,,\p_{\alpha,\,\beta}\rangle= -\frac{2(n-4)(n-2)}{n^4}$ follows by a similar argument as $6$. \\*

\noindent
$9$. Consider $\langle \q_{i,j}\,,\q_{\alpha,\,\beta}\rangle$. Since $(i,j)\neq (\alpha,\,\beta)$, both $\q_{i,j}$ and $\q_{\alpha,\,\beta}$
have non-zero entry at $(s,t)$ if and only if $s\neq i, s\neq j, s\neq \alpha, s\neq \beta$, $t\neq i, t\neq j, t\neq \alpha$ and $t\neq \beta$.
Therefore, $|\textrm{supp}(\q_{i,j})\cap \textrm{supp}(\q_{\alpha,\,\beta})| = (n-4)(n-4)=  (n-4)^2$ given the $n-4$ choices
for $s$ and $n-4$ choices for $t$.  $\langle \q_{i,j}\,,\q_{\alpha,\,\beta}\rangle$ can now be written as follows
\begin{align*}
\langle \q_{i,j}\,,\q_{\alpha,\,\beta}\rangle = \sum_{(s,t)\in\, \textrm{supp}(\q_{i,j})\cap \textrm{supp}(\q_{\alpha,\,\beta}) } [\q_{i,j}]_{s,t}[\q_{\alpha,\,\beta}]_{s,t}
=
(n-4)^2\left(-\frac{1}{n^2}\right)\left(-\frac{1}{n^2}\right)= \frac{(n-4)^2}{n^4}
\end{align*}
Therefore, $\langle \v_{i,j}\,,\v_{\alpha,\,\beta}\rangle$ is the sum of the above terms.
\[
\langle \v_{i,j}\,,\v_{\alpha,\,\beta}\rangle = -4\,\frac{n-1}{n^4}+2\,\frac{(n-1)^2+1}{n^4}+\frac{2(n-2)^2}{n^4}-\frac{4(n-4)(n-2)}{n^4}+\frac{(n-4)^2}{n^4}
= \frac{1}{n^2}
\]
The last equality follows after some algebraic manipulations. \\*

\noindent
Case $3$: $\langle \v_{i,j}\,,\v_{\alpha,\,\beta}\rangle$ for $\{i,j\}\cap \{\alpha,\,\beta\} \neq \varnothing$\\*

\noindent
A. $i=\alpha$ and $j \neq\beta$. 
\begin{align*} 
\langle \v_{i,j}\,,\v_{\alpha,\,\beta}\rangle  =& \langle \wt_{i,j}+\p_{i,j}+\q_{i,j}\,,  \wt_{\alpha,\,\beta}+\p_{\alpha,\,\beta}+\q_{\alpha,\,\beta}\rangle \\
                                                                  = & \langle \wt_{i,j}\,, \wt_{\alpha,\,\beta}\rangle +  \langle \wt_{i,j}\,,\p_{\alpha,\,\beta}\rangle +  \langle \wt_{i,j}\,,\q_{\alpha,\,\beta}\rangle+\\
                                                                         &\langle \p_{i,j}\,,\wt_{\alpha,\,\beta}\rangle +    \langle \p_{i,j}\,,\p_{\alpha,\,\beta}\rangle +
                                                                            \langle \p_{i,j}\,,\q_{\alpha,\,\beta}\rangle +    \langle \q_{i,j}\,,\wt_{\alpha,\,\beta}\rangle +   \langle \q_{i,j}\,,\p_{\alpha,\,\beta}\rangle +
                                                                               \langle \q_{i,j}\,,\q_{\alpha,\,\beta}\rangle 
\end{align*}
Each term will be evaluated separately. \\

\noindent
$1$. $\displaystyle \langle \wt_{i,j}\,, \wt_{\alpha,\,\beta}\rangle  = [\wt_{i,j}]_{i,i}[\wt_{\alpha,\,\beta}]_{i,i}= \frac{(n-1)^2}{n^4}$. \\*

\noindent
$2$. Consider $\displaystyle \langle \wt_{i,j}\,,\p_{\alpha,\,\beta}\rangle$.
\begin{align*}
\displaystyle \langle \wt_{i,j}\,,\p_{\alpha,\,\beta}\rangle  = [\wt_{i,j}]_{i,j} [\p_{i,j}]_{i,j}+ [\wt_{i,j}]_{j,i} [\p_{i,j}]_{j,i} &= -\left(\frac{2(n-1)^2+2}{2n^2}\right)\left( \frac{2n-4}{4n^2}\right)\\
&= -\frac{(n-2)[(n-1)^2+1]}{2n^4}  
\end{align*}

\noindent
$3$. $\displaystyle \langle \wt_{i,j}\,,\q_{\alpha,\,\beta}\rangle = [\wt_{i,j}]_{j,j} [\q_{\alpha,\,\beta}]_{j,j} = \frac{n-1}{n^2}\left(-\frac{1}{n^2}\right) = -\frac{n-1}{n^4}$.  \\*

\noindent
$4$. $\displaystyle \langle \p_{i,j}\,,\wt_{\alpha,\,\beta}\rangle =  -\frac{(n-2)[(n-1)^2+1]}{2n^4}$ follows by a similar argument as $2$. \\*

\noindent
$5$. Consider $\langle \p_{i,j}\,,\p_{\alpha,\,\beta}\rangle $. 
\begin{align*}
\langle \p_{i,j}\,,\p_{\alpha,\,\beta}\rangle = & \sum_{s,s\neq i, s \neq j } [\p_{i,j}]_{i,s}  [\p_{\alpha,\,\beta}]_{i,s} +  \sum_{t,t\neq i, t \neq j } [\p_{i,j}]_{j,t}  [\p_{\alpha,\,\beta}]_{j,t} +
\sum_{s,s\neq i, s \neq j } [\p_{i,j}]_{s,i}  [\p_{\alpha,\,\beta}]_{s,i} +  \sum_{t,t\neq i, t \neq j } [\p_{i,j}]_{t,j}  [\p_{\alpha,\,\beta}]_{t,j} \\
 = &   \sum_{s,s\neq i, s \neq j } [\p_{i,j}]_{i,s}  [\p_{\alpha,\,\beta}]_{i,s} +  [\p_{i,j}]_{j,\,\beta}[\p_{\alpha,\,\beta}]_{j,\,\beta} +
\sum_{s,s\neq i, s \neq j } [\p_{i,j}]_{s,i}  [\p_{\alpha,\,\beta}]_{s,i} +  [\p_{i,j}]_{t,\,\beta}  [\p_{\alpha,\,\beta}]_{t,\,\beta} \\
= &  (n-3)\left(\frac{2n-4}{4n^2}\right)^2+  \frac{2n-4}{4n^2} +
 (n-3)\left(\frac{2n-4}{4n^2}\right)^2 +  \frac{2n-4}{4n^2} \\
    = & 2(n-2) (n-3)\left(\frac{2n-4}{4n^2}\right)^2 = \frac{1}{2}\frac{(n-2)^{3}}{n^{4}}
\end{align*}
The second line follows since $[\p_{i,\,\beta}]_{j,t} =0$ for all $t$ except $t=\beta$ and $t=i$ and  $[\p_{i,\,\beta}]_{t,j} =0$ for all $t$ except $t=\beta$ and $t=i$.
The third line uses the fact that  $[\p_{i,\,\beta}]_{i,s} \neq 0$ for all $t$ except $t=i$ and $t=\beta$ and  $[\p_{i,\,\beta}]_{s,i} \neq 0$ for all $t$ except $t=i$ and $t=\beta$.
The final equality results after some algebraic manipulations.  \\*

\noindent
$6$. Consider $\langle \p_{i,j}\,,\q_{\alpha,\,\beta}\rangle$. 
\begin{align*}
\langle \p_{i,j}\,,\q_{\alpha,\,\beta}\rangle &= \sum_{t,t \neq i,t\neq j} [\p_{i,j}]_{i,t} [\q_{\alpha,\,\beta}]_{i,t} + \sum_{s,s\neq i,s\neq j} [\p_{i,j}]_{j,s} [\q_{\alpha,\,\beta}]_{j,s}+
\sum_{t,t \neq i,t\neq j} [\p_{i,j}]_{t,i} [\q_{\alpha,\,\beta}]_{t,i} + \sum_{s,s\neq i,s\neq j} [\p_{i,j}]_{s,j} [\q_{\alpha,\,\beta}]_{s,j}\\
& = 0 + \sum_{s,s\neq i,s\neq j} [\p_{i,j}]_{j,s} [\q_{\alpha,\,\beta}]_{j,s} + 0 +  \sum_{s,s\neq i,s\neq j} [\p_{i,j}]_{s,j} [\q_{\alpha,\,\beta}]_{s,j}\\
& = 2(n-3)\frac{2n-4}{4n^2}\left(-\frac{1}{n^2}\right) = -\frac{(n-3)(n-2)}{n^4}
\end{align*}
The first equality follow since it suffices to consider $\langle \p_{i,j}\,,\q_{\alpha,\,\beta}\rangle$ on the support of $\p_{i,j}$.
The second equality results since $[\q_{i,\,\beta}]_{i,t}=[\q_{i,\,\beta}]_{t,i}=0$ for any $t$ . 
The third and last equality result from the following analysis. 
Restricted to the support of $\p_{i,j}$, the matrix $\q_{i,\,\beta}$ is zero except at the entries
$[\q_{i,\,\beta}]_{j,s}$ and $[\q_{i,\,\beta}]_{s,j}$ for all $j\neq \beta$.
With this, $|\textrm{supp}(\p_{i,j})\cap \textrm{supp}(\q_{\alpha,\,\beta})| = 2[(n-2)-1] = 2(n-3)$
and the final form follows. \\*

\noindent
$7$. $\displaystyle \langle \q_{i,j}\,,\wt_{\alpha,\,\beta}\rangle =  -\frac{n-1}{n^4}$ follows by a similar argument as  $3$. \\*
 
\noindent
$8$. $\displaystyle \langle \q_{i,j}\,,\p_{\alpha,\,\beta}\rangle= -\frac{(n-3)(n-2)}{n^4}$ follows by a similar argument as $6$. \\*

\noindent
$9$. Consider $\langle \q_{i,j}\,,\q_{\alpha,\,\beta}\rangle$. Since $i=\alpha$ and $j\neq (\alpha,\,\beta)$, both $\q_{i,j}$ and $\q_{\alpha,\,\beta}$
have a non-zero entry at $(s,t)$ if and only if $s\neq i, s\neq \alpha, s\neq \beta$, $t\neq i, t\neq j$, and $t\neq \beta$.
Therefore, $|\textrm{supp}(\q_{i,j})\cap \textrm{supp}(\q_{\alpha,\,\beta})| = (n-3)(n-3)=  (n-3)^2$ given the $n-3$ choices
for $s$ and $n-3$ choices for $t$.  $\langle \q_{i,j}\,,\q_{\alpha,\,\beta}\rangle$ can now be written as follows
\begin{align*}
\langle \q_{i,j}\,,\q_{\alpha,\,\beta}\rangle = \sum_{(s,t)\in\, \textrm{supp}(\q_{i,j})\cap \textrm{supp}(\q_{\alpha,\,\beta}) } [\q_{i,j}]_{s,t}[\q_{\alpha,\,\beta}]_{s,t}
=
(n-3)^2\left(-\frac{1}{n^2}\right)\left(-\frac{1}{n^2}\right)= \frac{(n-3)^2}{n^4}
\end{align*}
Therefore, $\langle \v_{i,j}\,,\v_{\alpha,\,\beta}\rangle$ is the sum of the above terms.
\begin{align*}
\langle \v_{i,j}\,,\v_{\alpha,\,\beta}\rangle &= \frac{(n-1)^2}{n^4}+ -\frac{(n-2)[(n-1)^2+1]}{n^4}
-2\frac{n-1}{n^4}+  \frac{1}{2}\frac{(n-2)^{3}}{n^{4}} -2\frac{(n-3)(n-2)}{n^4}+\frac{(n-3)^2}{n^4}\\
&= \frac{2-n}{2n^2}
\end{align*}
The last equality follows after some algebraic manipulations. \\*

\noindent
Now the following three cases remain: $i\neq \alpha$ and $j = \beta$, $i =\beta$ and $j \neq  \alpha$ and $i \neq \beta $ and $j = \alpha$. 
However, since $\v_{i,j}$ is symmetric the above argument can be adapted to these three cases by interchanging indices
as appropriate. This concludes the proof. 
       
\end{proof}
\begin{remark}
A short proof of the form of $\H^{-1}$ might be plausible. The main technical challenge has been the locations of $1$ and $0$'s in $\H$ which does not lend itself to simple analysis.
\end{remark}

\noindent
Next, in Lemma \ref{H_prop}, we state and prove the spectral properties of the matrix $\H$. 

\begin{lemma} \label{H_prop}
  The matrix $\H\in \real^{L\times L}$ is symmetric and positive definite.
	The minimum eigenvalue of $\H$ is at least $1$ and the maximum eigenvalue is $2n$. 
	The absolute sum of each row of $\H^{-1}$ is given by $\displaystyle 2-\frac{15}{2n}+\frac{8}{n^2}$.
\end{lemma} 
\begin{proof}
The symmetry of $\H$ is trivial as $\H_{\alphab,\,\betab}=\langle \w_{\alphab}\,,\w_{\betab}\rangle
= \langle \w_{\betab}\,,\w_{\alphab}\rangle = \H_{\betab,\,\alphab}$. $\H$ is positive definite
since $\H = \W^{T}\W$ and $\W$ has linearly independent columns. With $\alphab=(\alpha_1,\alpha_2)$
and $\betab = (\beta_1,\beta_2)$, $\H$ is defined as follows
\[
\H_{\alphab,\,\betab} = \langle \w_{\alphab}\,,\w_{\betab}\rangle=
\begin{cases}
4 & \textrm{if }\alpha_1 = \beta_1 \,\&\, \alpha_2 = \beta_2 \\
0 & \textrm{if }\alpha_1 \neq \beta_1 \,\&\, \alpha_1 \neq \beta_2\,\&\,\alpha_2 \neq \beta_1 \,\&\, \alpha_2 \neq \beta_2\\
1 & \textrm{otherwise}
\end{cases}
\]
After minor analysis of $\H$, the number of zeros in any given row is
given by $\displaystyle \frac{(n-2)(n-3)}{2}$. The number of ones is simply
\[
L - \# \textrm{ number of zeros} - \# \textrm{ number of fours}  =
\frac{n(n-1)}{2} - \frac{(n-2)(n-3)}{2} - 1 = 2n-4
\]
To find the maximum eigenvalue
of $\H$, note that $\bm{1}$ is an eigenvector of $\H$ with eigenvalue $2n$. 
From Gerschgorin theorem, the upper bound for an eigenvalue is simply $4+(2n-4) = 2n$.
It follows that the maximum eigenvalue of $\H$ is $2n$. 
Using the form of $\H^{-1}$ from Lemma \ref{H_inverse}, 
consider the absolute sum of any row of $\H^{-1}$, $\sum_{i} |\H^{-1}_{ij}|$.
\[
\sum_{i} |\H^{-1}_{i,j}| = 
\frac{1}{n^2}\bigg[\frac{(n-1)^{2}+1}{2}\bigg]+ \frac{1}{n^2}\bigg[\frac{(n-2)(n-3)}{2}\bigg] + \frac{1}{n^2}(2n-4) \big|1-\frac{1}{2}n\big|
= 2 - \frac{15}{2n} + \frac{8}{n^2}
\]
Finally, we make use of a variant of Gerschgorin's theorem to bound the maximum eigenvalue of $\H^{-1}$.
If $\bm{P}$ is invertible, $\bm{P}^{-1}\H^{-1}\bm{P}$ and $\H^{-1}$ have the same eigenvalues.
For simplicity, let $\bm{P} = \textrm{diag } (d_1,d_2,...,d_L)$. 
In \cite[p. $347$]{horn1990matrix}, this fact was used to state the following variant of Gerschgorin's theorem.
Below, we restate this result, albeit minor changes, for ease of reference.  

\begin{corollary}[{\cite[p. $347$]{horn1990matrix}}]

Let $\bm{A} = [a_{ij}]\in \real^{L\times L}$ and  let $d_1,d_2,...,d_L$ be positive real numbers. Then
all eigenvalues of $\bm{A}$ lie in the region
\[
\bigcup_{i=1}^{L} \bigg\{ \bm{z}\in \real:|z-a_{ii}|\le \frac{1}{d_i}\sum_{j=1,j\neq i}^{L} d_{j}\,|a_{ij}|\bigg\}
\]
\end{corollary}

\noindent
Using the corollary, set $d_1 = 9$ and $d_2=d_3=...=d_L = 3$. Apply the corollary on the matrix $\H^{-1}$.
After minor calculation, we have $\lambda_{\max}(\H^{-1})\le 1$. We remark here that with suitable choice of
$(d_1,d_2,...,d_L)$, the bound can be tightened to $\lambda_{\max}(\H^{-1})=\frac{1}{2}$ but the current bound is sufficient for our analysis. 

\end{proof}

Assume that the underlying inner product matrix $\M$ has coherence $\nu$ with respect to the standard basis.
Let $\e_{i}\in \real^{n}$ be the standard vector, a vector of zeros except a $1$ in the $i$th position.
For all $i$, $1\le i\le n$, the  coherence definition \cite{candes2009exact} states that
\begin{equation}
\label{eqn:coherenceSB}
||\P_{\U}\,\e_{i}||_{2}^{2}\le \frac{\nu r}{n} 
\end{equation}
It suffices to consider the condition on $\U$ since $\M$ is symmetric. 
Given this, could one derive coherence conditions for the EDG problem? The answer is affirmative
and is given in Lemma \ref{coherence_conditions_derivation} below.

\begin{lemma} \label{coherence_conditions_derivation}
  If the underlying inner product matrix $\M$ has coherence $\nu$ with respect to the standard basis, i.e. $\M$ satisfies \eqref{eqn:coherenceSB}, then
  the following coherence conditions hold for the EDG problem. 
\[
||\P_{\T}\,\w_{\alphab}||_{F}^{2} \le  \frac{8\nu r}{n} \quad ; \quad
||\P_{\T}\,\v_{\alphab}||_{F}^{2} \le  \frac{8\nu r}{n}
\]
\end{lemma}

\begin{proof}

  Using the definition of $\P_{\T}$ and the fact that $\w_{\alphab}$ is symmetric for any $\alpha$,
	we have
\begin{align*}
||\P_{\T}\,\w_{\alphab}||_{F}^{2} &=
\langle \w_{\alphab}\,,\U\U^{T}\w_{\alphab} \rangle +\langle \w_{\alphab}\,,\w_{\alphab}\U\U^{T}
\rangle - \langle \w_{\alphab}\,,\U\U^{T}\w_{\alphab}\U\U^{T}\rangle\\
& = 2\langle \w_{\alphab}\,,\U\U^{T}\w_{\alphab} \rangle - \langle \U\U^{T}\w_{\alphab}\U\U^{T}\,,\U\U^{T}\w_{\alphab}\U\U^{T}\rangle
 \le 2\langle \w_{\alphab}\,,\U\U^{T}\w_{\alphab} \rangle
\end{align*} 
Note that $\langle \w_{\alphab}\,,\U\U^{T}\w_{\alphab}\rangle = \langle \U\U^{T}\,\w_{\alphab}\,,\U\U^{T}\,\w_{\alphab}\rangle\ge 0$.
Using the definition of $\langle \X\,,\w_{\alphab}\rangle$ and the fact that $\w_{\alphab}^{2}=2\w_{\alphab}$ for
any $\alphab$, $\langle \w_{\alphab}\,,\U\U^{T}\w_{\alphab}\rangle  =
2[(\U\U^{T})_{\alpha_1,\,\alpha_1}+(\U\U^{T})_{\alpha_2,\,\alpha_2}-2(\U\U^{T})_{\alpha_1,\,\alpha_2}]\le
2[(\U\U^{T})_{\alpha_1,\,\alpha_1}+(\U\U^{T})_{\alpha_2,\,\alpha_2}+2|(\U\U^{T})_{\alpha_1,\,\alpha_2}|]\le
4 [(\U\U^{T})_{\alpha_1,\,\alpha_1}+(\U\U^{T})_{\alpha_2,\,\alpha_2}]$. The last inequality holds 
since $\U\U^{T}$ is positive semidefinite.
This motivates a bound on $ \underset{ij}\max \,|\U\U^{T}|_{i,j}$. 
\begin{align*}
|\U\U^{T}|_{i,j} = \bigg|\sum_{k=1}^{r} \U_{i,k} \U_{j,k}\bigg| \le \sum_{k=1}^{r} |\U_{i,k}|\, |\U_{j,k}|
\le \sqrt{\sum_{k=1}^{r} \U_{i,k}^{2}}  \sqrt{\sum_{k=1}^{r} \U_{j,k}^{2}} \le 
\sqrt{\frac{\nu r}{n}} \sqrt{\frac{\nu r}{n}} = \frac{\nu r}{n}
\end{align*}
Using the above bound, $\langle \w_{\alphab}\,,\U\U^{T}\w_{\alphab}\rangle \le 8\frac{\nu r}{n}$
resulting the following bound for $||\P_{\T}\,\w_{\alphab}||_{F}^{2}$.
\[
||\P_{\T}\,\w_{\alphab}||_{F}^{2} \le  \frac{16\nu r}{n}
\]
To bound $||\P_{\T}\,\v_{\alphab}||_{F}^{2}$, note that 
$||\P_{\T}\,\v_{\alphab}||_{F} \le \sum_{\betab\in \Us} ||\H^{\alphab,\,\betab}\,\P_{\T} \w_{\betab}||_{F} = \sum_{\betab\in \Us} |\H^{\alphab,\,\betab}|\,\,||\P_{\T} \w_{\betab}||_{F}$.
Using Lemma \ref{H_prop} and the bound for $||\P_{\T}\,\w_{\alphab}||_{F}$,
\[||\P_{\T}\,\v_{\alphab}||_{F}^{2} \le  \frac{64\nu r}{n}\]

\end{proof}

In addition, the standard matrix completion analysis assumes that 
$\displaystyle \underset{ij}{\max}\,|(\U\U^T)_{i,j}| \le \frac{\mu_1\sqrt{r}}{n}$
for some constant $\mu_1$ \cite{candes2009exact}. If this holds, for some $\alphab$, it follows
that
\begin{align*}
\langle \U\U^{T},\w_{\alphab}\rangle &= (\U\U^{T})_{\alpha_1,\,\alpha_1}+(\U\U^{T})_{\alpha_2,\,\alpha_2}-2(\U\U^{T})_{\alpha_1,\,\alpha_2}\\
& \le (\U\U^{T})_{\alpha_1,\,\alpha_1}+(\U\U^{T})_{\alpha_2,\,\alpha_2}+2|(\U\U^{T})_{\alpha_1,\,\alpha_2}|\\
&\le 2 [(\U\U^{T})_{\alpha_1,\,\alpha_1}+(\U\U^{T})_{\alpha_2,\,\alpha_2}] \le 4 \,\frac{\mu_1\sqrt{r}}{n}
\end{align*}
With this, it can be concluded that
\begin{equation}\label{eq:jointcoherencew}
\underset{\alphab\in \Us}{\max} \,\, \langle \w_{\alphab}\,,\U\U^{T}\rangle^{2} \le 16\,\frac{\mu_{1}^{2}r}{n^2}
\end{equation}
Another short calculation results a bound on $ \langle \v_{\alphab}\,,\U\U^{T}\rangle^{2}$
\[
 \langle \v_{\alphab}\,,\U\U^{T}\rangle^{2}	= \langle \sum_{\betab\in \Us} \H^{\alphab,\,\betab}\w_{\betab}\,,\U\U^{T}\rangle^{2}	\le
\underset{\betab\in \Us}{\max} \,\, \langle \w_{\betab}\,,\U\U^{T}\rangle^{2}\,\left(\sum_{\betab\in \Us} |\H^{\alphab,\,\betab}|\right)^{2}
\le 64\frac{\mu_1 r}{n^2}
\]
The last inequality uses \eqref{eq:jointcoherencew} and Lemma  \ref{H_prop}. From the above inequality, it follows that
\[
\underset{\alphab\in \Us}{\max} \,\, \langle \v_{\alphab}\,,\U\U^{T}\rangle^{2} \le 64\,\frac{\mu_{1}^{2}r}{n^2}
\]

Lemma \ref{coherence_conditions_derivation} and the discussion above show that the coherence conditions with respect
to standard basis lead to comparable EDG coherence conditions.
Specifically, we obtain conditions equivalent up to constants
to \eqref{eq:coherencew}, \eqref{eq:coherencev} and \eqref{eq:jointcoherence}.
We remark here that the condition \eqref{eq:coherencev1}
does not simply follow from the coherence conditions
with respect to the standard basis. We speculate that the equivalence is possible
under certain assumptions but a rigorous analysis is left for future work.

\begin{lemma}\label{norm_equivalent}
Given any $\X \in \S$, the following estimates hold. 
\[
\frac{1}{2n}||\X||_{F}^{2}\le \sum_{\alphab\in \Us}  \langle \X\,,\v_{\alphab}\rangle^{2} \le ||\X||_{F}^{2} \quad;
\quad
||\X||_{F}^{2} \le \sum_{\alphab\in \Us}  \langle \X\,,\w_{\alphab}\rangle^{2} \le 2n||\X||_{F}^{2}
\]
\end{lemma}

\begin{proof}
Vectorize the matrix $\X$ and each dual basis $\v_{\alphab}$. It follows that
\[ 
\sum_{\alphab\in \Us} \langle \X\,,\v_{\alphab}\rangle^{2} = \sum_{\alphab\in \Us}  \x^{T}\v_{\alphab}\v_{\alphab}^{T} \x 
= \x^{T} \V\V^{T} \x
\]
Orthogonalize $\V$, $\overline{\V} = \V(\sqrt{\H^{-1}})^{-1}$. With this,
$\sum_{\betab\in \Us}  \langle \X\,,\v_{\betab}\rangle^{2}   
= \x^{T}\, \overline{\V}\H^{-1}\overline{\V}^{T} \x$. It follows that
\[
\lambda_{\min}(\H^{-1})\,||\x||_{2}^{2}=\frac{1}{2n}\,||\x||_{2}^{2} \le 
\sum_{\betab\in \Us}  \langle \X\,,\v_{\betab}\rangle^{2} \le \lambda_{\max}(\H^{-1})\,||\x||_{2}^{2} \le ||\x||_{2}^{2}
\]
where the above result follows from the min-max theorem and Lemma \ref{H_prop}. Proceeding analogously
as above, noting that $\lambda_{\max}(\H)= n$ and $\lambda_{\min}(\H)\ge 1$ from Lemma \ref{H_prop}, we obtain
\[
||\X||_{F}^{2}\le \sum_{\alphab\in \Us}  \langle \X\,,\w_{\alphab}\rangle^{2} \le 2n||\X||_{F}^{2}
\]
This concludes the proof.
\end{proof}

\begin{lemma}\label{operator_norm_of_sum}
Let $\c_{\alpha}\ge 0$. Then, 
\[
\bigg|\bigg|\sum_{\alphab} \c_{\alpha}\, (\P_{\T^{\perp}}\,\w_{\alphab})^{2}\bigg|\bigg| 
\le \bigg|\bigg|\sum_{\alpha} \c_{\alpha}\, \w_{\alphab}^{2}\bigg|\bigg|
\]
\end{lemma}

\begin{proof}
Using the definition of $\P_{\T^{\perp}}\,\w_{\alphab}$, 
$\bigg|\bigg|\sum_{\alphab} \c_{\alpha}\, (\P_{\T^{\perp}}\,\w_{\alphab})^{2}\bigg|\bigg|$ can be written
as follows.
\[
\bigg|\bigg|\sum_{\alphab} \c_{\alpha}\, (\P_{\T^{\perp}}\,\w_{\alphab})^{2}\bigg|\bigg| = 
\bigg|\bigg|\sum_{\alphab} \c_{\alpha}\, \P_{\U^{\perp}}\,\w_{\alphab}\,\P_{\U^{\perp}}\,\w_{\alphab}\,\P_{\U^{\perp}}\bigg|\bigg|
= \bigg|\bigg|\P_{\U^{\perp}}\left( \sum_{\alphab} \c_{\alpha}\, \w_{\alphab}\,\P_{\U^{\perp}}\,\w_{\alphab}\right)\P_{\U^{\perp}}\bigg|\bigg|
\]
Using the fact that the operator norm is unitarily invariant and $||\P\X\P||\le ||\X||$ for any $\X$ and a projection $\P$,
$\bigg|\bigg|\sum_{\alphab} \c_{\alpha}\, (\P_{\T^{\perp}}\,\w_{\alphab})^{2}\bigg|\bigg|$ can be upper bounded as follows
\[
\bigg|\bigg|\sum_{\alphab} \c_{\alpha}\, (\P_{\T^{\perp}}\,\w_{\alphab})^{2}\bigg|\bigg| 
\le \bigg|\bigg|\sum_{\alphab} \c_{\alpha}\, \w_{\alphab}\,\P_{\U^{\perp}}\,\w_{\alphab}\bigg|\bigg|
= \bigg|\bigg|\sum_{\alphab} \c_{\alpha}\, \w_{\alphab}\,\left(\w_{\alphab}-\P_{\U}\,\w_{\alphab}\right)\bigg|\bigg|
= \bigg|\bigg|\sum_{\alphab} [\c_{\alpha}\, \w_{\alphab}^{2}-\c_{\alpha}\,\w_{\alphab}\,\P_{\U}\,\w_{\alphab}]\bigg|\bigg|
\]
where the first equality follows from the relation $\P_{\U^{\perp}} = \I -\P_{\U}$. Since $\w_{\alphab}=\w_{\alphab}^{T}$
and $\c_{\alpha}\ge 0$, $\sum_{\alphab} \c_{\alpha}\, \w_{\alphab}^{2}$ is positive semidefinite.
Using the relation $\P_{\U}^{2} = \P_{\U}$ and the assumption that $\c_{\alpha}\ge 0$, 
$\sum_{\alphab} \c_{\alpha}\w_{\alphab}\,\P_{\U}\,\w_{\alphab}=\sum_{\alphab}\c_{\alpha}  \w_{\alphab}\,\P_{\U}\,\P_{\U}\,\w_{\alphab}$ is 
also positive semidefinite. Repeating the same argument, 
$\sum_{\alphab} \c_{\alpha}\, \w_{\alphab}\,\P_{\U^{\perp}}\,\w_{\alphab}$ is also
positive semidefinite. Finally, the norm inequality, $||\bm{A}+\bm{B}||\ge \max(||\bm{A}||,||\bm{B}||)$, for
positive semidefinite matrices $\bm{A}$ and $\bm{B}$, concludes the proof. 
\end{proof}

\begin{lemma} \label{joint_coherence_size_golfing}
	Define $\displaystyle \eta(\X) = \max_{\betab\in\Us} \,|\langle \X\,,\v_{\betab} \rangle|$.
  For a fixed $\X \in \T$, the following estimate holds. 
  \begin{equation}
    \textrm{Pr}\,(\max_{\betab\in\Us}\,\,|\langle \P_{\T}\R^{*}_{j}\X-\X\,,\v_{\betab}\rangle|\ge t) \le n^{2} \exp\left(-\frac{3t^2\kappa_j}{16\eta(\X)^{2}\left(\nu+\frac{1}{nr}\right)}\right)
  \end{equation}
  for all $t\le \eta(\X)$ with $\displaystyle\kappa_{j} = \frac{m_{j}}{nr}$. 
  \end{lemma}

\begin{proof}

  For some $\v_{\betab}$, expand $\langle \P_{\T}\R^{*}_{j}\X-\X\,,\v_{\betab}\rangle$
  in the following way: $$\langle \P_{\T}\R^{*}_{j}\X-\X\,,\v_{\betab}\rangle = 
	\langle \sum_{\alphab\in\Omega_j} \frac{L}{m_j}\langle \X\,,\v_{{\alphab}}\rangle \P_{\T}\,\w_{\alphab}-\X\,,\v_{\betab}\rangle
  = \sum_{\alphab\in\Omega_j} \left(\frac{L}{m_j}\langle \X\,,\v_{\alphab}\rangle \langle \P_{\T}\,\w_{\alphab}\,,\v_{\betab}\rangle -\frac{1}{m_j} \langle \X\,,\v_{\betab}\rangle\right)$$
  Note that the summand can be written as $Y_{\alphab} = X_{\alphab}-E[X_{\alphab}]$. 
	By construction, $E[Y_{\alphab}]=0$. 
To apply Bernstein inequality, it remains to compute a suitable bound for $|Y_{\alphab}|$ and $|E[Y_{\alphab}^{2}]|$.
$|Y_{\alphab}|$ is bounded as follows.
  \begin{align*}
  |Y_{\alphab}| &= \bigg|\frac{L}{m_j}\langle \X\,,\v_{\alphab}\rangle \langle \P_{\T}\,\w_{\alphab}\,,\v_{\betab}\rangle -\frac{1}{m_j} \langle \X\,,\v_{\betab}\rangle\bigg|
  \le \frac{L}{m_j} \eta(\X)\,4\frac{\nu r}{n}+\frac{1}{m_j}\eta(\X)\\
  & < \frac{1}{m_j}\eta(\X)\left(1+2nr\nu\right) 
  \end{align*}
  Above, the first inequality follows from the coherence estimates \eqref{eq:coherencew} and \eqref{eq:coherencev}.
	Next, we consider a bound $|E[Y_{\alphab}^{2}]|$. Since
  $E[Y_{\alphab}^{2}] = E[X_{\alphab}^{2}]-(E[X_{\alphab}])^{2}$,  $E[Y_{\alphab}^{2}]$ can be upper bounded as follows
  \begin{align*}
  E[Y_{\alphab}^{2}]& \le   E\bigg[\frac{L^2}{m_{j}^2}\langle \X\,,\v_{\alphab}\rangle^{2} \langle \P_{\T}\,\w_{\alphab}\,,\v_{\betab}\rangle^{2}\bigg]+
                                     \frac{1}{m_j^2} \langle \X\,,\v_{\betab}\rangle^{2}\\
                                     & = \frac{L}{m_{j}^2}\sum_{\alphab\in \Us} \langle \X\,,\v_{\alphab}\rangle^{2} \langle \P_{\T}\,\w_{\alphab}\,,\v_{\betab}\rangle^{2}
                                     +     \frac{1}{m_j^2} \langle \X\,,\v_{\betab}\rangle^{2}\\
                                     & \le  \eta(\X)^{2} \frac{n^2}{2m_{j}^2}\sum_{\alphab\in \Us} \langle \P_{\T}\,\w_{\alphab}\,,\v_{\betab}\rangle^{2}
                                     +  \frac{1}{m_j^2} \eta(\X)^{2}  \\                                     
                                     & \le \eta(\X)^{2} \frac{2nr\nu}{m_j^2} + \frac{2}{m_j^2} \eta(\X)^{2}     = \frac{2\eta(\X)^2}{m_j^2}\left(1+nr\nu\right)                                                 
    \end{align*}
Above, the last inequality follows from the coherence estimate \eqref{eq:coherencev1}.
 Finally, apply the scalar Bernstein inequality
with $\displaystyle R= \frac{\eta(\X)\left(1+2nr\nu\right)}{m_j} $ and $\displaystyle \sigma^2=\frac{2\eta(\X)^2\left(1+nr\nu\right)}{m_j}  $. For $t\le \frac{\sigma^{2}}{R}>\eta(\X)$, 
\begin{equation}
    \textrm{Pr}(|\langle \P_{\T}\R^{*}_{j}\X-\X\,,\v_{\betab}\rangle|\ge t) \le  \exp\left(-\frac{3t^2\kappa_j}{16\eta(\X)^{2}\left(\nu+\frac{1}{nr}\right)}\right)
  \end{equation}
The proof of Lemma \ref{joint_coherence_size_golfing} concludes by simply applying the
union bound.
\end{proof}

\end{document}